\documentclass[11pt,leqno]{article}
\usepackage[T1]{fontenc}
\usepackage{mathtools}
\usepackage{amssymb}
\usepackage{amsthm}
\usepackage{mathrsfs}
\usepackage{faktor}
\usepackage{graphicx}
\usepackage{float}
\usepackage{xcolor}
\usepackage{enumitem}
\usepackage[font=small,labelfont=bf]{caption}
\usepackage[font=small]{subcaption}
\usepackage[margin=1in]{geometry}
\usepackage[colorlinks=true, linkcolor=blue!80!black, citecolor=magenta]{hyperref}
\usepackage{mathpazo}

\numberwithin{equation}{section}
\newtheorem{theorem}{Theorem}[section]
\newtheorem{lemma}{Lemma}[section]
\newtheorem{corollary}{Corollary}[section]
\theoremstyle{definition}
\newtheorem{definition}{Definition}[section]
\theoremstyle{remark}
\newtheorem{remark}{Remark}[section]

\newcommand{\e}[1]{{\varepsilon_#1}}
\newcommand{\w}[1]{{\omega_#1}}
\newcommand{\g}[1]{{\gamma_#1}}

\newcommand{\Cmainw}{\mathcal{C}_0}
\newcommand{\Cotherw}{\mathcal{C}_1}
\newcommand{\Cmainz}{\mathcal{C}_0'}
\newcommand{\Cotherz}{\mathcal{C}_1'}
\newcommand{\Cmainwres}{\widetilde{\mathcal{C}_0}}
\newcommand{\Cotherwres}{\widetilde{\mathcal{C}_1}}
\newcommand{\Cmainzres}{\widetilde{\mathcal{C}_0}'}
\newcommand{\Cotherzres}{\widetilde{\mathcal{C}_1}'}
\newcommand{\Cmainwm}{\mathcal{C}_{0,m}}
\newcommand{\Cotherwm}{\mathcal{C}_{1,m}}
\newcommand{\Cmainzm}{\mathcal{C}_{0,m}'}
\newcommand{\Cotherzm}{\mathcal{C}_{1,m}'}
\newcommand{\bconst}{B}
\newcommand{\circlecontour}{C}

\title{Local correlation functions of the two-periodic weighted Aztec diamond in mesoscopic limit}
\author{Emily Bain\footnote{University of California, Berkeley and Yau Mathematical Sciences Center, Tsinghua University}}
\date{}

\begin{document}
\maketitle

\begin{abstract}
Here we study the two-periodic weighted dimer model on the Aztec diamond graph. In the thermodynamic limit when the size of the graph goes to infinity while weights are fixed, the model develops a limit shape with frozen regions near corners, a flat ``diamond'' in the center with a noncritical (ordered) phase, and a disordered phase separating this diamond and the frozen phase. We show that in the mesoscopic scaling limit, when weights scale in the thermodynamic limit such that the size of the ``flat diamond'' is of the same order as the correlation length inside the diamond, fluctuations of the height function are described by a new process. We compute asymptotics of the inverse Kasteleyn matrix for vertices in a local neighborhood in this mesoscopic limit.
\end{abstract}

\section{Introduction}

\subsection{Overview}

A \textit{dimer configuration} is a perfect matching on a graph. A \textit{dimer model} on a weighted planar graph is a probability measure on dimer configurations, such that the probability of a dimer configuration is proportional to the product of its edge weights. Early work on dimer models concerned counting the number of dimer configurations on a graph \cite{KASTELEYN1961,Temperley1961DimerPI,Kasteleyn1963}. 

 A \textit{tiling} of a planar graph is a covering of the faces of the graph by tiles consisting of two adjacent faces. Dimer configurations of a planar graph are in bijection with tilings of the dual graph, and so for every dimer model there is an associated tiling model. We will talk about a dimer model or the corresponding tiling model interchangeably as convenient. 
 
The two most commonly studied classes of tiling models are \textit{domino tilings}, which correspond to dimer models on a square grid, and \textit{lozenge tilings}, which correspond to dimer models on a hexagonal grid. Here we focus on domino tiling models. For a thorough discussion of lozenge tiling models and some general theory see \cite{gorin_2021}.

An important tool in the study of dimer models is the Kasteleyn method. The \textit{Kasteleyn matrix} is essentially a weighted oriented adjacency matrix for the dimer graph. The absolute value of its determinant is the partition function of the dimer model \cite{Kasteleyn1963}. For the uniform weighted case, this is just the number of dimer configurations. We can also use the Kasteleyn matrix to find the correlation functions between edges, where the \textit{correlation function} of a set of edges is the probability that a random dimer configuration has dimers at each of these edges. These correlation functions can be written in terms of the Kasteleyn matrix and its inverse \cite{Kenyon_1997}. As a result, finding the inverse Kasteleyn matrix for a dimer model is of considerable interest.

Work on domino tiling models intensified after Thurston's 1990 paper \cite{Thurston1990} on height functions for dimer models on bipartite planar graphs. A \textit{height function} assigns a height to every face of the dimer graph, or equivalently every to vertex of the tiling graph. In this way, a tiling model can be considered as a random surface in three-dimensional space. In \cite{Cohn_2000} the authors proved a variational principle for domino tiling models that can be used to find the limit shape of the associated random surface as the graph size tends to infinity.

The \textit{Aztec diamond graph} is part of a square grid with a boundary at a 45 degree angle to the grid (Figure \ref{fig:aztecdiamondgraph}). Domino tilings of the Aztec diamond graph were first studied in \cite{elkies1992}, and subsequently in \cite{Cohn_1996, Jockush1998, Cohn_2000, Johansson2005}.

Dimer models exhibit up to three different phases: \textit{frozen} (or solid), \textit{disordered} (or liquid or rough) and \textit{ordered} (or gas or smooth) \cite{kenyon_okounkov_sheffield_2006}. These are characterized by the rate of decay of correlation functions between dimers when the distance between them is increasing, or equivalently by the variance of the height function. The Aztec diamond tiling model with uniform weights exhibits frozen and disordered regions. In \cite{Jockush1998}, the authors showed that the frozen-disordered boundary in the continuous limit can be described by an algebraic curve -- the `arctic circle' -- and in \cite{Johansson2005} it was shown that the fluctuations around the arctic circle at the disordered-frozen boundary converge to the Airy process \cite{Prahofer_scaleinvariance} under suitable rescaling. See \cite{johansson2018} for a survey of results on boundary fluctuations of dimer models. Equations for phase boundaries for lozenge tiling models were studied in \cite{kenyon2007limit} using tools from algebraic geometry.

\begin{figure}[ht]
\centering
 \includegraphics[width = 0.3\textwidth]{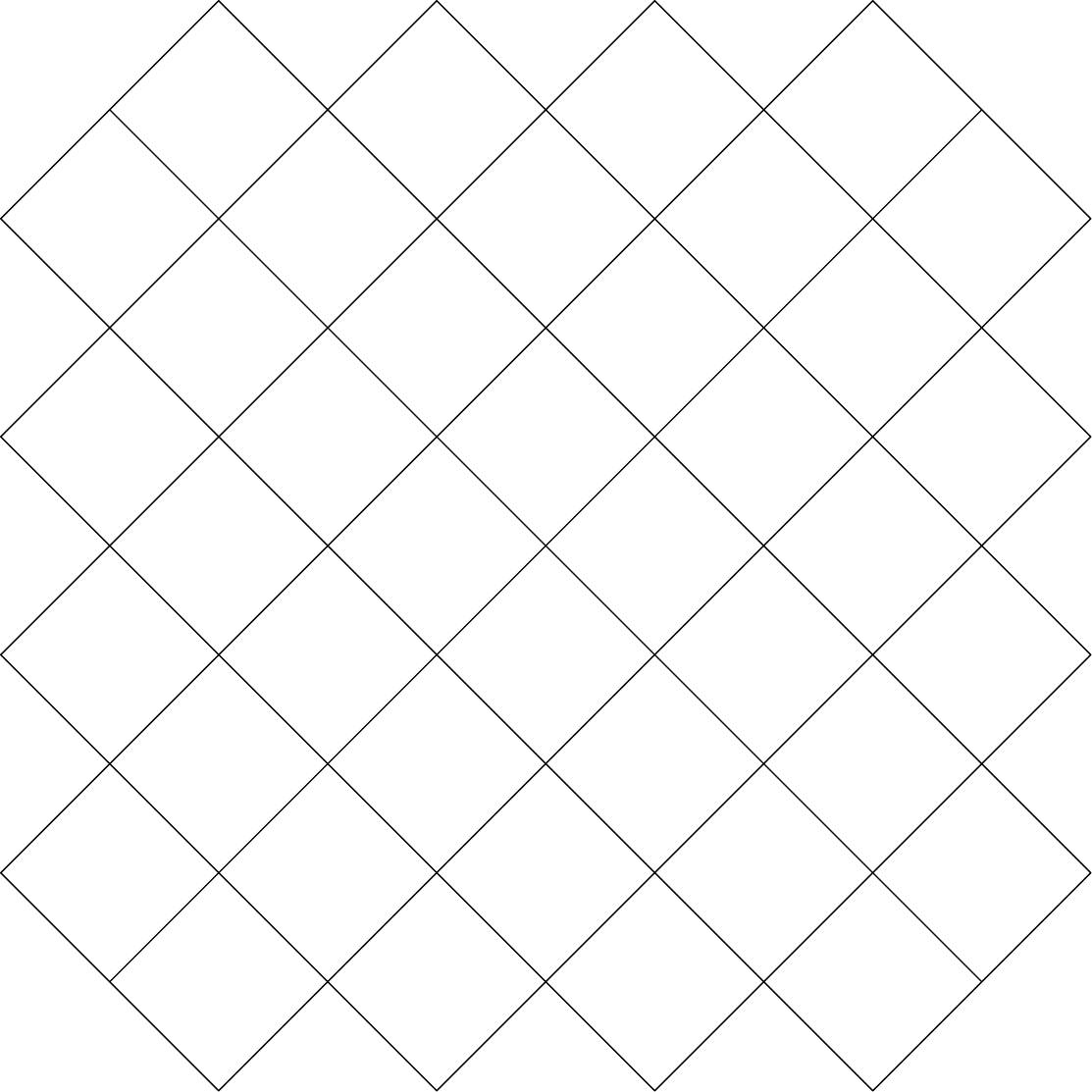}
\caption{The Aztec diamond of size $n=4$.}
\label{fig:aztecdiamondgraph}
\end{figure}

The \textit{two-periodic weighted Aztec diamond} \cite{chhita2014, Francesco_2014} is a probability measure on tilings on the Aztec diamond graph with doubly-periodic weights and a $2\times 2$ fundamental domain (smallest repeating block). An example is shown in Figure \ref{fig:dominotiling}. It is one of the simplest models to exhibit all three phases: frozen, disordered and ordered. There are phase boundaries between the frozen and disordered phases, and between the disordered and ordered phases. A formula for the entries of the inverse Kasteleyn matrix was found in \cite{chhita2014} and simplified in \cite{chhita2016domino}. Recently, other approaches have also been used to find the correlation functions; see \cite{Duits2021, BERGGREN2019}. The behavior at the ordered-disordered boundary, where the Airy kernel point process has been observed \cite{chhita2016domino}, is of particular interest \cite{chhita2016domino, johansson2018, Beffara2018, Beffara2022, JohanssonMason2021}. These Airy kernel point processes are expected to be a universal phenomenon. 

Currently, most results are on the statistical properties of the ordered-disordered boundary, but there has been some progress on finding a microscopic description of this boundary in the form of lattice paths \cite{Beffara2022}. A generalization of the two-periodic weighted Aztec diamond where there is a bias towards horizontal dominos has recently been studied in \cite{Borodin2022}.

In this paper, we go in a different direction and look at the rough-smooth boundary analytically in a particular scaling limit where the macroscopic size of the ordered region tends to zero. We find new behavior not seen before in other scaling limits. A similar phenomenon was observed numerically for the six-vertex model with $\Delta < -1$ in \cite{Belov2022TheTC} where it was called the mesoscopic scaling limit. We expect other models to exhibit similar behavior.

\begin{figure}[htb]
\centering
\includegraphics[width=0.6\textwidth]{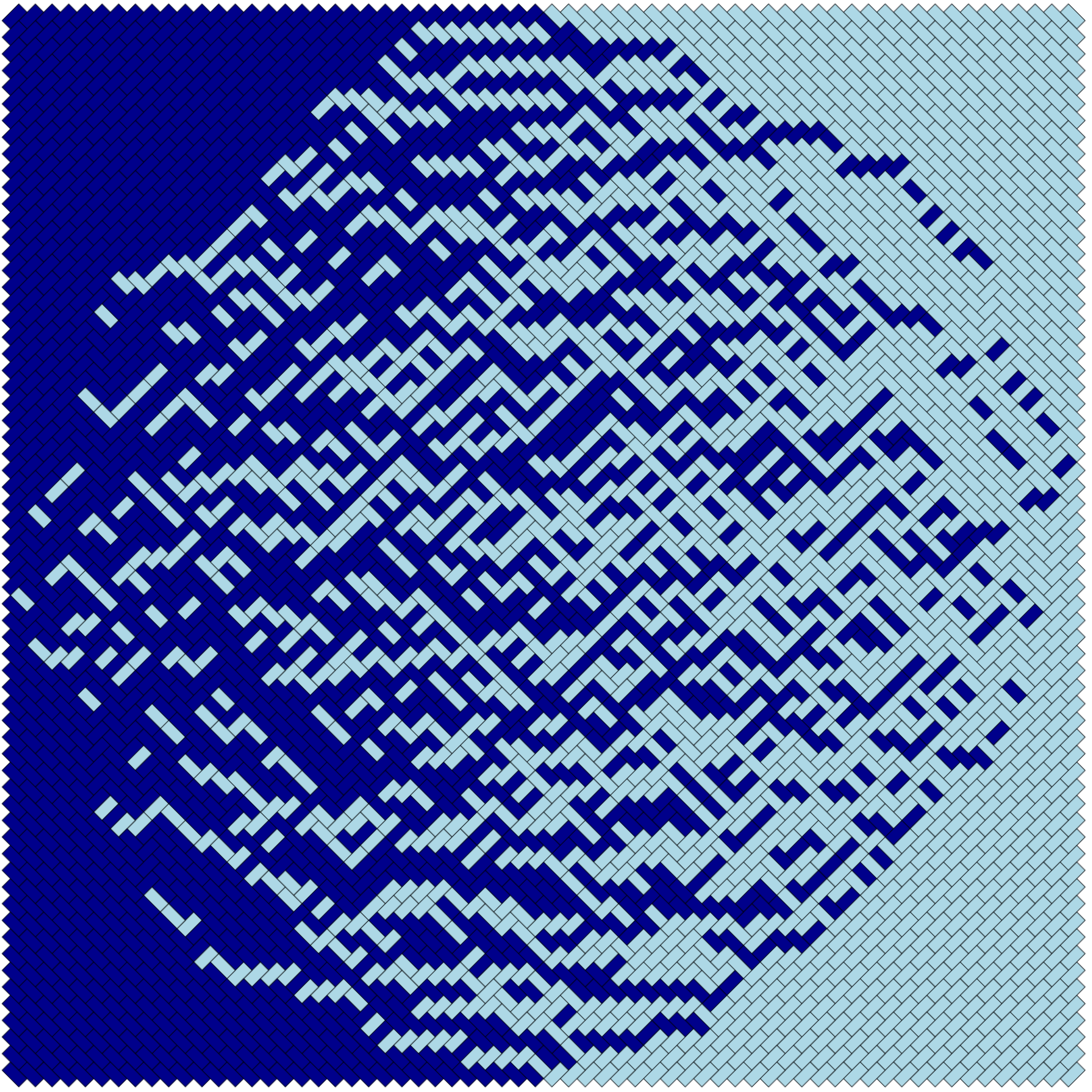}
\caption{Tiling of two-periodic weighted Aztec diamond with $n=64$ and $a = 0.8$}
\label{fig:dominotiling}
\end{figure}

\subsection{Informal description of our results}

We look at the limit when the size of the Aztec diamond graph tends to infinity and the weights tend to the uniform weights in such a way that the \textit{correlation length inside the ordered region is of the order of the linear size of the ordered region}.

More specifically, as in \cite{chhita2016domino} we assume the edge weights for the dimer graph are 1 and $a<1$. These weights repeat as shown in Figure \ref{fig:aztec diamond diagram}. If $n$ is the linear size of the whole diamond we consider the limit when $1-a$ is of order $n^{-1/2}$.  The region of the Airy asymptotic of correlation functions found in \cite{chhita2016domino} in this case is of the same order as the width of the ordered region. We study the correlation functions in this ``mesoscopic'' limit \cite{Belov2022TheTC}.

The ordered region macroscopically shrinks to a point in this limit and there is no longer an ordered region as such. We show that the leading order term of the one-point correlation functions along a diagonal at a distance of order $n^{1/2}$ from the center of the Aztec diamond is $1/4$. This means that at any vertex, the probabilities of having a dimer on the four edges incident to the vertex are equal. The subleading order term is a constant term of order $n^{-1/2}\log n$. The next term has order $n^{-1/2}$ and it depends on the rescaled coordinates. This last term is the focus of this paper. 

We show that terms of order $n^{-1/2}$ in the one-point correlation function are no longer described by an Airy kernel point process; instead they are given by another double integral. We give precise integral formulas for the subleading order terms. We expect to see similar asymptotic behavior away from the diagonal, except possibly on the horizontal and vertical lines through the center, where for finite weights there are cusp points on the disordered-ordered boundary.

In an upcoming paper, we will study two-point correlations for dimers that grow further apart microscopically as $n$ tends to infinity.

\begin{figure}[ht]
\centering
 \includegraphics[width = 0.8\textwidth]{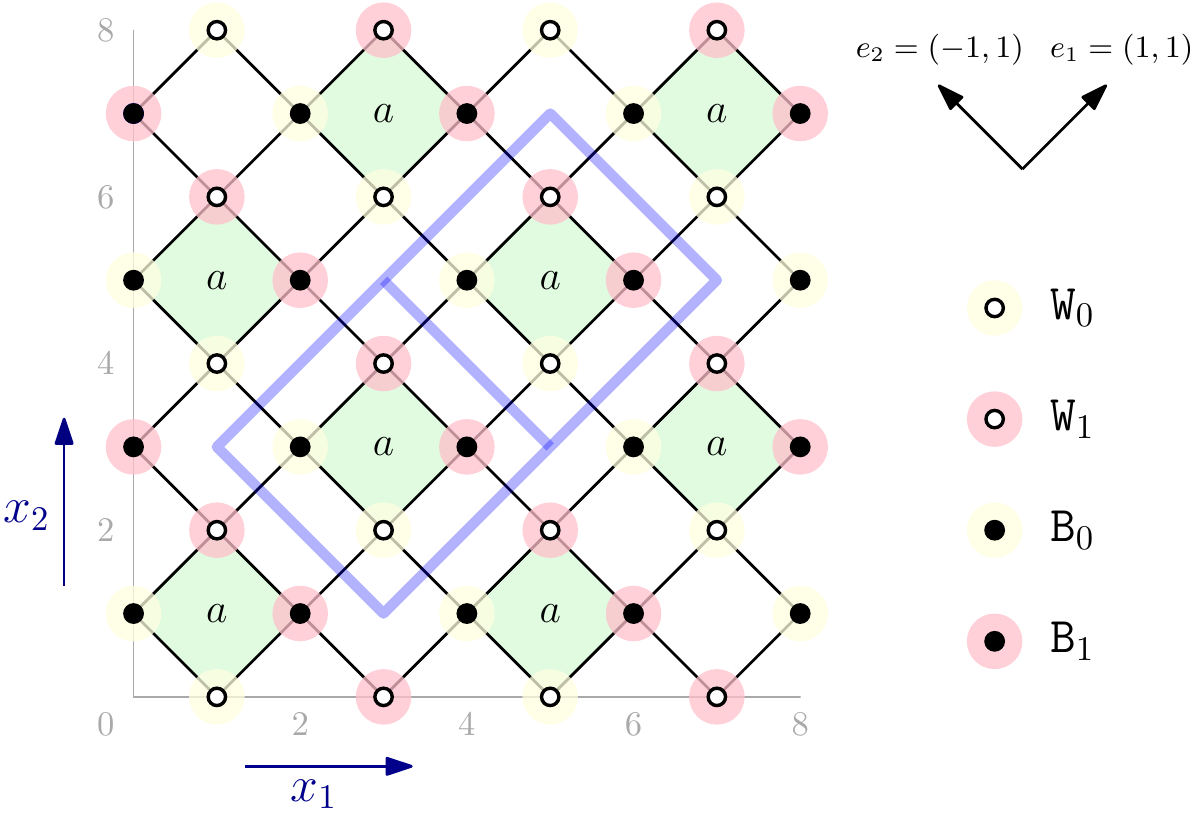}
\caption{The two periodic Aztec diamond graph of size $n=4$ showing the subgraphs $\mathtt{W}_0$, $\mathtt{W}_1$, $\mathtt{B}_0$ and $\mathtt{B}_1$. Edges surrounding a green square have weight $a$; other edges have weight 1. Two fundamental domains are marked in blue.}
\label{fig:aztec diamond diagram}
\end{figure}

\subsection{Acknowledgements}
This research was partly supported by NSF grant DMS-1902226. We thank Yau Mathematical Sciences Center, Tsinghua University, where this work was completed. We are grateful to Nicolai Reshetikhin for guidance and many helpful discussions during the course of this work, and Zitong Cheng for the computations to pass from Equation~\ref{eq:kgasinverserealunsimplified} to Equation~\ref{eq:kgasinversesimplified}. We thank the reviewer who's careful reading and constructive comments were critical to improving the clarity of this manuscript.

We used code based on the source code developed by Keating and Sridhar \cite{keating2018code} to simulate domino tilings. 

This research used the Savio computational cluster resource provided by the Berkeley Research Computing program at the University of California, Berkeley (supported by the UC Berkeley Chancellor, Vice Chancellor for Research, and Chief Information Officer). 

\section{Main result}

Before we state the main result we must cover some preliminaries.

\subsection{Definition of model}\label{sec:coords}

\begin{figure}[ht]
\centering
 \includegraphics[width = 0.7\textwidth]{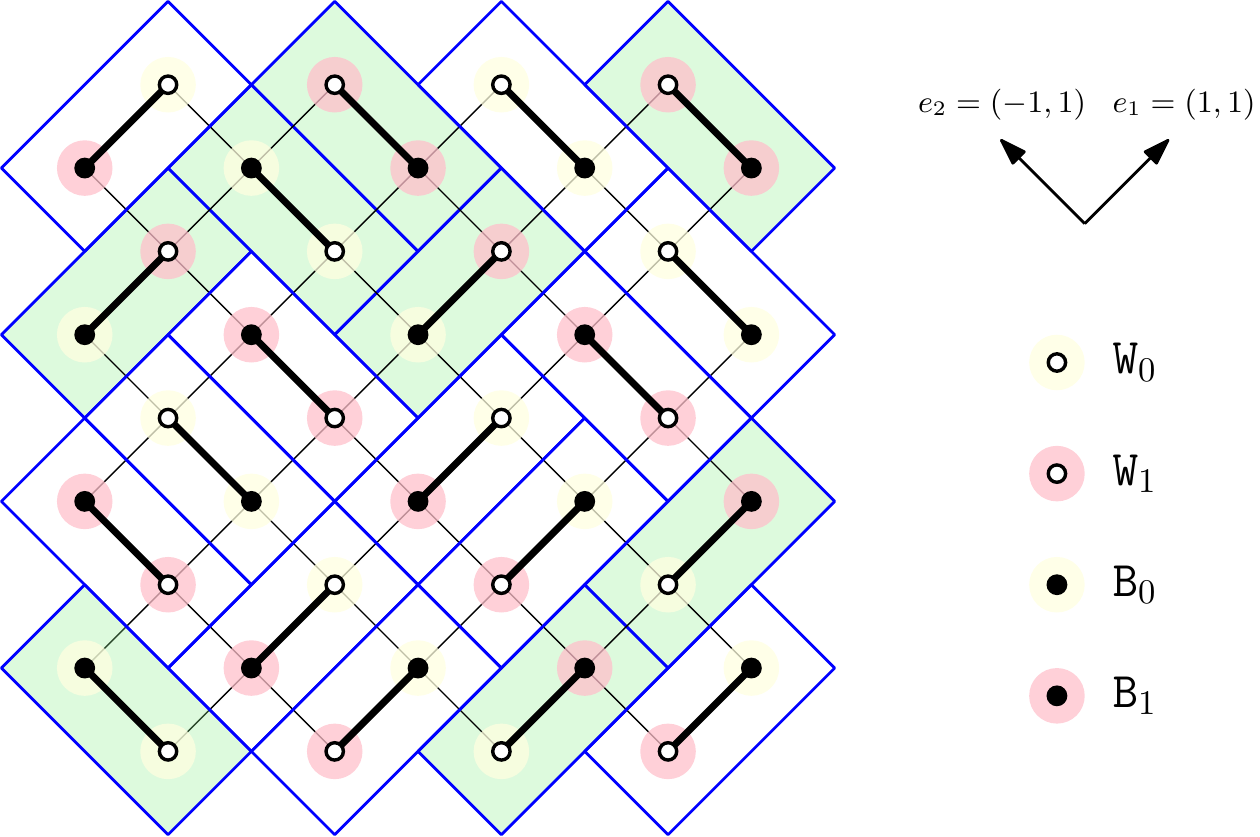}
\caption{The two periodic Aztec diamond graph of size $n=4$ with a dimer matching, and its dual domino tiling in blue. Weight $a$ dominos are shaded in green.}
\label{fig:aztec diamond domino}
\end{figure}

We study the limit shape for random domino tilings of an Aztec diamond with two-periodic weights. Let $\Gamma$ denote the dimer graph, and let $\Gamma^\vee$ denote the dual graph. Each vertex of $\Gamma$ marks the center of a face of $\Gamma^\vee$. The graph with its bipartite structure and the structure of weights is shown in Figure~\ref{fig:aztec diamond diagram}. An example of a tiling of $\Gamma$ and the corresponding dimer configuration is given in Figure~\ref{fig:aztec diamond domino}.

Because each fundamental domain \cite{kenyon_okounkov_sheffield_2006} is a $2\times 2$ square region, the linear size of the diamond should be a multiple of 2. In this paper for simplicity we only consider the case where the linear size is a multiple of 4. Let $n=4m$ denote the linear size of $\Gamma$.\footnote{We will mostly follow the same setup as Chhita and Johansson \cite{chhita2016domino}.} So on each slice containing either all black or all white vertices we will have either $4m$ segments between vertices or $4m-1$ segments between vertices and two half segments at the sides. We will use Euclidean coordinates $(x_1, x_2)$ with $x_i \in \mathbb{Z}$ ranging from 0 to $2n=8m$ as in Figure~\ref{fig:aztec diamond diagram}, with the vertices being at points satisfying $x_1 + x_2 \equiv 1 \mod 2$. Thus the linear Euclidean length of the diamond along the axes $x_1$ and $x_2$ is $2n=8m$.

We fix a bipartite structure on $\Gamma$ by coloring the vertices black and white (see Figure~\ref{fig:aztec diamond diagram}). We denote the set of white vertices by $\mathtt{W}$ and the set of black vertices by $\mathtt{B}$ where we set \begin{align*}
\mathtt{W} &= \{ (x_1,x_2) \in \Gamma : x_1 \equiv 1 \mod 2,
\, x_2 \equiv 0 \mod 2,\, 0\leq x_1,x_2 \leq 2n\}\\
\mathtt{B} &= \{ (y_1,y_2) \in \Gamma : y_1 \equiv 0 \mod 2,
\, y_2 \equiv 1 \mod 2,\, 0\leq y_1,y_2 \leq 2n\}.\end{align*} Then we split the white vertices $\mathtt{W}$ into two subgraphs $\mathtt{W}_0$ and $\mathtt{W}_1$, and the black vertices $\mathtt{B}$ into $\mathtt{B}_0$ and $\mathtt{B}_1$ according to the values of $x_1 + x_2 \mod 4$ and $y_1 + y_2 \mod 4$ respectively as follows (and as is shown in Figure~\ref{fig:aztec diamond diagram}).\begin{align*}
\mathtt{W}_\e1 &= \{ (x_1,x_2) \in \mathtt{W} : x_1 + x_2 \equiv 2\e1 + 1 \mod 4\} \\
\mathtt{B}_\e2 &= \{ (y_1,y_2) \in \mathtt{B} : y_1 + y_2 \equiv 2\e2 + 1 \mod 4\} 
\end{align*} 
Let $e_1 = (1,1)$ and $e_2 = (-1,1)$. There is a $2\times 2$ fundamental domain which when embedded in the graph consists of a vertex $\mathtt{w} \in \mathtt{W}_0$ along with $\mathtt{w} + e_1$, $\mathtt{w} + e_2$ and $\mathtt{w} + e_1 + e_2$. 


We assign a weight to each edge. The weights of edges connecting vertices in one fundamental domain are $a$, and the weights of edges connecting vertices in neighboring fundamental domains are 1. This is shown diagrammatically in Figure~\ref{fig:aztec diamond diagram}, where the edges of weight $a$ are the edges surrounding a face labelled with $a$, and the other edges have weight 1.

Let $\Omega$ denote the set of all dimer configurations on $\Gamma$. Define the partition function \[
Z = \sum_{\mathcal{D} \in \Omega }\prod_{e\in \mathcal{D}}w(e).
\] Then the Boltzmann measure is defined as \[
\mathrm{Prob}(\mathcal{D}) = \frac{\prod_{e\in \mathcal{D}}w(e)}{Z}
\] for $\mathcal{D} \in \Omega$. 

Let $e_i = (\mathbf{w}_i, \mathbf{b}_i)$, $1 \leq i \leq k$ be edges of $\Gamma$. The $k$-point correlation function is \[
\rho(\mathbf{w}_1,\mathbf{b}_1;\,\mathbf{w}_2,\mathbf{b}_2;\,\ldots ;\, \mathbf{w}_k,\mathbf{b}_k) = \sum_{\mathcal{D}\in \Omega} \mathrm{Prob}(\mathcal{D}) \sigma_{\mathbf{w}_1,\mathbf{b}_1}(\mathcal{D})\ldots \sigma_{\mathbf{w}_k,\mathbf{b}_k}(\mathcal{D})
\] where \[\sigma_{\mathbf{w},\mathbf{b}}(\mathcal{D}) = \begin{cases} 1 &\text{if } (\mathbf{w},\mathbf{b})\in \mathcal{D} \\
0 &\text{otherwise}. \end{cases}\]  This is the probability of a random dimer configuration containing the dimers on edges $\{e_i\}_{1\leq i \leq k}$.

\subsection{Kasteleyn solution}
Here we give a brief overview of the Kasteleyn method. The Kasteleyn matrix for the two periodic Aztec diamond is defined as follows. For each edge $e$ of $\Gamma$, let $w(e)$ be the weight of the edge $e$, as shown in Figure~\ref{fig:aztec diamond diagram}. We define a Kasteleyn-Percus orientation $\epsilon$ on the edges of $\Gamma$. For $(x,y)$ an edge of $\Gamma$ with $x\in \mathtt{W}$ and $y\in \mathtt{B}$, let
\[
\epsilon(x,y) = \begin{cases}
1 & y = x \pm e_1\\
i & y = x \pm e_2
\end{cases}
\] This is shown diagrammatically in Figure~\ref{fig:orientation}. Then for $x\in \mathtt{W}$ and $y\in \mathtt{B}$ we define the Kasteleyn matrix $K_a$ \footnote{In \cite{chhita2016domino} $K_{a,1}$ is used instead of $K_a$.} as having entries \[
K_a(y,x) = \begin{cases}
w(x,y)\epsilon(x,y) & \text{if there is an edge between }x\text{ and }y \\
0 & \text{otherwise}
\end{cases}. \] Explicitly,\[
K_a(y,x) = \begin{cases}
a(1-\varepsilon) + \varepsilon & y = x + e_1,\, x\in W_\varepsilon \\
(1-\varepsilon) + a \varepsilon & y = x - e_1,\, x\in W_\varepsilon \\
i(a(1-\varepsilon) + \varepsilon) & y = x + e_2,\, x\in W_\varepsilon \\
i((1-\varepsilon) + a\varepsilon) & y = x - e_2,\, x\in W_\varepsilon \\
0 & \text{if } (x,y) \text{ is not an edge}
\end{cases}
\] This is shown diagrammatically for a fundamental domain in Figure~\ref{fig:kmatrix}. Let $K_a^{-1}$ denote the inverse matrix. 

For edges $e_i = (\mathbf{w}_i, \mathbf{b}_i)$, $1 \leq i \leq k$, the $k$-point correlation function has been shown \cite{Kenyon_1997} to be  \[
\rho(\mathbf{w}_1,\mathbf{b}_1;\,\mathbf{w}_2,\mathbf{b}_2;\,\ldots \mathbf{w}_k,\mathbf{b}_k) = \left(\prod_{i=1}^{k} K_a(\mathbf{b}_i,\mathbf{w}_i)\right) \det (K_a^{-1}(\mathbf{w}_i, \mathbf{b}_j))_{1 \leq i,j \leq k}
\]  

When we are looking at a single edge, we will denote this edge by $(x,y)$, where $x\in\mathtt{W}$ and $y\in\mathtt{B}$. Then the one-point correlation is given by $\rho(x,y) = K_a(y,x) K_a^{-1}(x,y)$ . This is the probability that a random dimer configuration contains the dimer $(x,y)$. 

\begin{figure}[ht]
\centering
 \includegraphics[width = 0.3\textwidth]{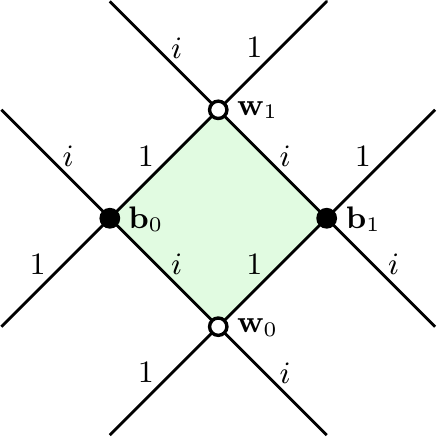}
\caption{The Kasteleyn-Perkus orientation shown on a fundamental domain.}
\label{fig:orientation}
\end{figure}

\begin{figure}[ht]
\centering
 \includegraphics[width = 0.3\textwidth]{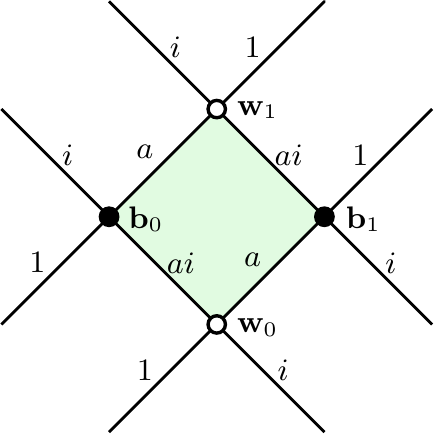}
\caption{The entries $K_a(y,x)$ of the Kasteleyn matrix for vertices $x \in  \mathtt{W}$, $y  \in \mathtt{B}$ connected by an edge, shown for a fundamental domain.}
\label{fig:kmatrix}
\end{figure}

 \subsection{Main theorem}\label{sec:overviewofresults}
We look at the inverse Kasteleyn matrix of vertices in a local neighborhood that is a Euclidean distance of order $m^{1/2}$ from the center $(4m, 4m)$ of the Aztec diamond, near the diagonal in the third quadrant (see Figure \ref{fig:halfdiagregion}), in the limit where the weight $a$ is given by $a=1-Bm^{-1/2}$ for some constant $B>0$. We define the asymptotic coordinate $\alpha < 0$ as follows. For $\e1,\e2\in \{0,1\}$, let $x = (x_1,x_2) \in \mathtt{W}_\e1$ and $y= (y_1,y_2) \in \mathtt{B}_\e2$ be vertices of an edge of the graph $\Gamma$, with \begin{align}\begin{split}\label{eq:asymptoticcoord}
x_1 &= [4m + 2m^{1/2}\alpha \bconst] + \overline{x_1} \\ 
x_2 &= [4m + 2m^{1/2}\alpha \bconst] + \overline{x_2} \\ 
y_1 &= [4m + 2m^{1/2}\alpha \bconst] + \overline{y_1} \\ 
y_2 &= [4m + 2m^{1/2}\alpha \bconst] + \overline{y_2},
\end{split}
\end{align} where the integral parts $\overline{x}_i, \overline{y}_i \in \mathbb{Z}$ do not grow with $m$. 

\begin{figure}[htbp]
\centering
\includegraphics[width = 0.4\textwidth]{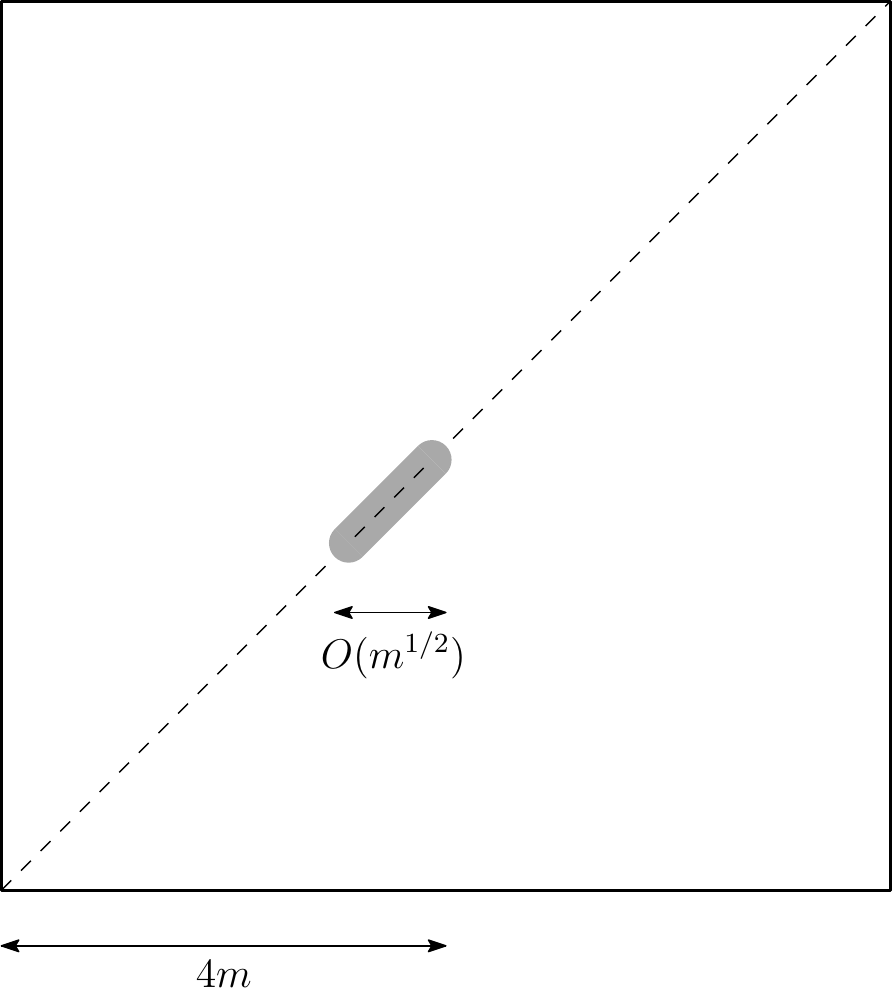}
\caption{In this paper, we study the one-point correlation functions of the dominos in the area shaded in gray.}
\label{fig:halfdiagregion}
\end{figure}

Define the matrix $\zeta$ to have entries \begin{equation}\label{eq:zeta}
\zeta(x,y)=(-1)^{(y_2-x_1)/2}.
\end{equation} We also define the matrix $\Sigma$ by \begin{equation}\label{eq:Sigmaxy}
\Sigma(x,y) = \begin{cases}
1 & y = x + (2k+1)e_1 + 2le_2, \text{ some } k,l\in \mathbb{Z}\\
i & y = x + 2ke_1 + (2l+1)e_2, \text{ some } k,l\in \mathbb{Z}
\end{cases}
\end{equation} Note that when $(x,y)$ is an edge of $\Gamma$, this agrees with the Kasteleyn-Percus orientation $\epsilon(x,y)$.

In what follows, the square roots refer to the principal branch of the square root. Let $\eta = \eta(\alpha)$ be defined to be the unique complex number with non-negative real part and non-negative imaginary part that satisfies \begin{equation}
\frac{1}{\sqrt{1/2 - 2 i \eta}} + \frac{1}{\sqrt{1/2+ 2i \eta}} = -2/\alpha.\label{eq:eta}\end{equation} When $-1/\sqrt{2} < \alpha < 0$, $\eta\in i(0,1/4)$; when $\alpha < -1/\sqrt{2}$, $\eta \in (0, \infty)$; and when $\alpha=-1/\sqrt{2}$, $\eta = 0$.

Let $\eta'$ be the unique complex number that satisfies \begin{equation}
\frac{1}{\sqrt{1/2 - 2 i \eta'}}- \frac{1}{\sqrt{1/2+ 2i \eta'}} = 2/\alpha.\label{eq:eta'}\end{equation} For all $\alpha < 0$ we have $\eta' \in i(0, 1/4)$.

Let \begin{equation}\label{eq:fpm}
f^\pm(w) = \sqrt{1/2-2iw} \pm \sqrt{1/2+2iw}.
\end{equation}and for $j,k \in \{0,1\}$ and $\e1,\e2 \in \{0,1\}$, let
\begin{multline}
A^{j,k}_{\e1,\e2}(w,z) =  \frac{-(-1)^{\e1+\e2}}{\sqrt{1/2-2iw}\sqrt{1/2+2iw}\sqrt{1/2-2iz}\sqrt{1/2+2iz}}\Bigg(2 i(w-z) \\
+ (-1)^{\e1+\e2}\left(\sqrt{1/2-2iw} +  (-1)^{j}\sqrt{1/2-2iz}\right)\left((-1)^{k}\sqrt{1/2+2iw} + \sqrt{1/2+2iz}\right) \\
+\Big((-1)^\e1\sqrt{1/2-2iw} + (-1)^{\e2+k}\sqrt{1/2+2iw} 
+ (-1)^\e2\sqrt{1/2+2iz} + (-1)^{\e1+j}\sqrt{1/2-2iz}\Big)\\
\times \left(\sqrt{1/2-2iw}\sqrt{1/2+2iz} + (-1)^{j+k}\sqrt{1/2+2iw}\sqrt{1/2-2iz}\right)\Bigg).\label{eq:Aall}
\end{multline} 
We define the following double integrals, where $x = (x_1,x_2) \in \mathtt{W}_\e1$ and $y= (y_1,y_2) \in \mathtt{B}_\e2$ and $\alpha$ is as in Equation~\ref{eq:asymptoticcoord}.
\begin{align}\begin{split}
I_1(\alpha, \e1, \e2) &=  \int_{\Cmainw} dw\int_{\Cmainz}dz \,\frac{A_{\e1,\e2}^{0,0}(w,z)}{i(z-w)} \exp(\bconst^2(-2i(w-z)+ \alpha ( f^-(w)  - f^-(z)))),\\
I_2(\alpha, \e1, \e2) &=  \int_{\Cmainw} dw\int_{\Cotherz}dz \,\frac{A_{\e1,\e2}^{1,0}(w,z)}{i(z-w)} \exp(\bconst^2(-2i(w-z)+ \alpha  (f^-(w)  +  f^+(z)))) ,\\
I_3(\alpha, \e1, \e2) &=  \int_{\Cotherw} dw\int_{\Cmainz}dz \,\frac{A_{\e1,\e2}^{0,1}(w,z)}{i(z-w)} \exp(\bconst^2(-2i(w-z)+ \alpha (f^+(w)  -  f^-(z)))) ,\\
I_4(\alpha, \e1, \e2) &=  \int_{\Cotherw} dw\int_{\Cotherz}dz \,\frac{A_{\e1,\e2}^{1,1}(w,z)}{i(z-w)} \exp(\bconst^2(-2i(w-z)+ \alpha  (f^+(w)  +  f^+(z)))),\label{eq:I1to4}\end{split}
\end{align} where the functions $A_{\e1,\e2}^{j,k}(w,z)$ are defined in Equation~\ref{eq:Aall} and the contours are defined below. We also define the single integral
 \begin{equation}
I_0(\alpha, \e1, \e2) =  \int_{-\eta(\alpha)}^{\eta(\alpha)}\frac{1 + (-1)^\e2\sqrt{1/2-2iw} + (-1)^\e1\sqrt{1/2+2iw}}{\sqrt{1/2-2iw}\sqrt{1/2+2iw}}dw
\label{eq:I0}\end{equation} where $\eta(\alpha)$ is defined in Equation~\ref{eq:eta}. These integrals all evaluate to real quantities.

The contours $\Cmainw$, $\Cmainz$, $\Cotherw$ and $\Cotherz$ are defined as follows. They are shown in Figure~\ref{fig:C1C2} for some different values of $\alpha$.

Recall that along a steepest descent contour of a holomorphic function (contour where the real part decreases most rapidly), its imaginary part is constant. A function has a saddle point when its second derivative is 0. The function $-2iw + \alpha f^-(w)$ has saddle points at $w = \pm \eta$ and the function $-2iw + \alpha f^+(w)$ has a saddle point at $w = -\eta'$.

All contours are oriented in the direction of decreasing real part.

For $-1/\sqrt{2}<\alpha < 0$, let $\Cmainw$ be the steepest descent contour for $-2iw + \alpha_x f^-(w)$ that is contained in the negative half plane and passes through the saddle point $w = -\eta$.

For $\alpha = -1/\sqrt{2}$ let $\Cmainw$ be the steepest descent contour for $-2iw + \alpha_x f^-(w)$ that passes through the saddle point $w = 0$ and enters the negative half plane at angles of $-\pi/6$ and $-5\pi/6$.

For $\alpha < -1/\sqrt{2}$, let $\Cmainw$ consist of the steepest descent contour for $-2iw + \alpha_x f^-(w)$ that starts from the branch cut $i(1/4,\infty)$, passes through the saddle point $w = -\eta$ and goes to infinity in the third quadrant; the reflection in the imaginary axis of this contour; and a contour that goes around the branch cut $i(1/4, \infty)$. This contour is shown in detail in Figure~\ref{fig:C1}.

Let $\Cmainz$ be the reflection of $\Cmainw$ in the real axis.

Let $\Cotherw$ be the steepest descent contour for $ -2iw + \alpha_x f^+(w)$. This passes through $w = -\eta'$ and goes to infinity in the negative half plane.

Let $\Cotherz$ be the reflection of $\Cotherw$ in the real axis.

\begin{figure}[htbp]
\centering
\includegraphics[width = 0.6\textwidth]{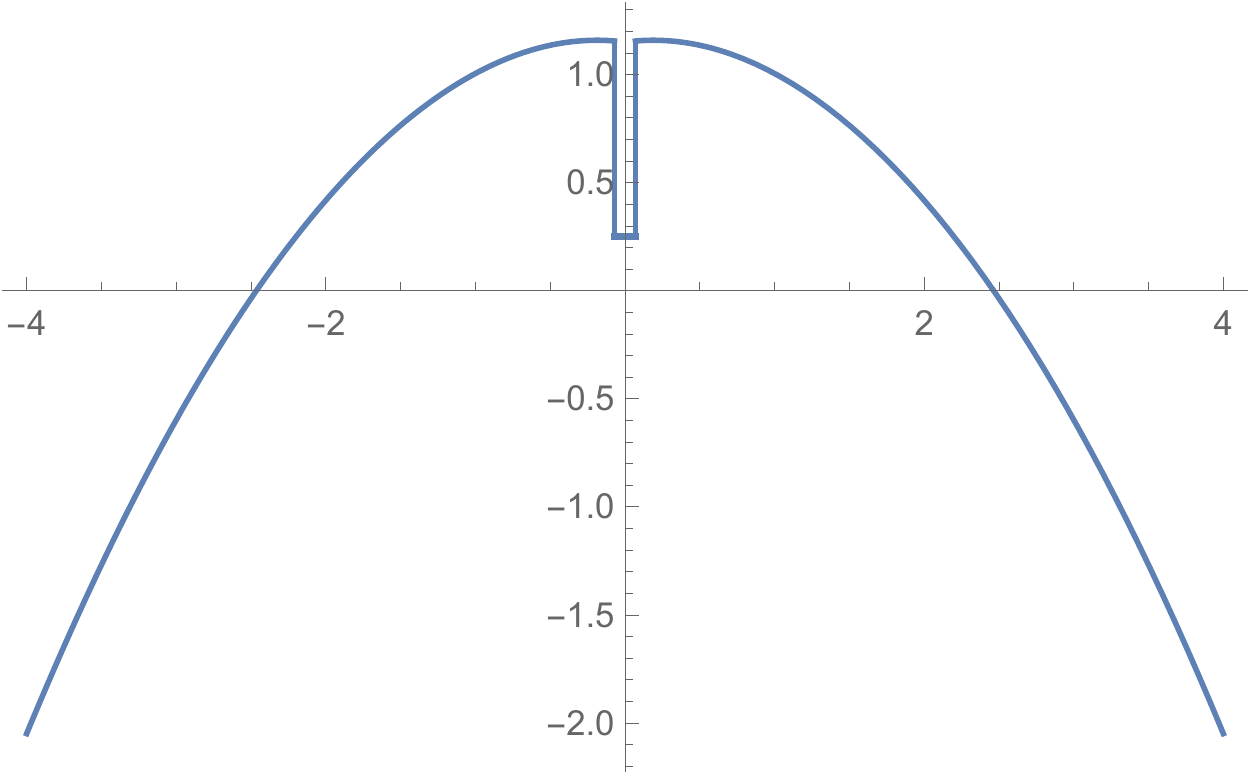}
\caption{The contour $\Cmainw$ for $\alpha = -3$.}
\label{fig:C1}
\end{figure}

In the limit as $m$ tends to infinity with $a = 1-\bconst m^{-1/2}$ we prove the following theorem for asymptotic formulas for the inverse Kasteleyn matrix $K_a^{-1}$. 

\begin{theorem}\label{thm:mainresult}
 For $-1/\sqrt{2} \leq \alpha < 0$, if $x = (x_1,x_2) \in \mathtt{W}_\e1$ and $y= (y_1,y_2) \in \mathtt{B}_\e2$ with $\e1,\e2\in \{0,1\}$ we have
\begin{multline}\label{eq:mainresult1}
K_a^{-1}(x,y) = \frac{1}{\Sigma(x,y)}\Bigg(c_0(y-x)\left(1 + \frac{\bconst m^{-1/2}}{2}\right) \\
+\zeta (x,y)\bconst m^{-1/2}\left(\frac{ \log (\bconst m^{-1/2})}{2\pi} + c_2(y-x)
+ \psi(\alpha,\e1,\e2)\right)\Bigg) + O(m^{-1}\log m )
\end{multline} and for $\alpha < -1/\sqrt{2}$ we have
\begin{multline}\label{eq:mainresult2}
K_a^{-1}(x,y) = \frac{1}{\Sigma(x,y)}\Bigg(c_0(y-x)\left(1 + \frac{\bconst m^{-1/2}}{2}\right) \\
+\zeta (x,y)\bconst m^{-1/2}\bigg(\frac{\log (\bconst m^{-1/2})}{2\pi} + c_2(y-x)  + \frac{I_0(\alpha, \e1, \e2)}{4\pi}
+ \psi(\alpha,\e1,\e2)\bigg)\Bigg)+ O(m^{-1}\log m )
\end{multline} where \begin{equation}\label{eq:psixy} \psi(\alpha, \e1, \e2) = \frac{1}{32 \pi^2}(I_1(\alpha, \e1, \e2) - I_2(\alpha, \e1, \e2) - I_3(\alpha, \e1, \e2) + I_4(\alpha, \e1, \e2)).
\end{equation} and the integrals $I_0(\alpha, \e1, \e2)$, $I_1(\alpha, \e1, \e2)$, $I_2(\alpha, \e1, \e2)$, $I_3(\alpha, \e1, \e2)$ and $I_4(\alpha, \e1, \e2)$ are defined above in Equations \ref{eq:I1to4} and \ref{eq:I0}, and $c_0$ and $c_2$ are functions that depend only the vector $y-x$, defined in Equations~\ref{eq:c0} and \ref{eq:c2} respectively.
\end{theorem}
The proof is given in Section~\ref{sec:proofofmainresult}.

\begin{figure}[htbp]
\centering
\begin{subfigure}[t]{0.5\textwidth}
\centering
\includegraphics[width = \textwidth]{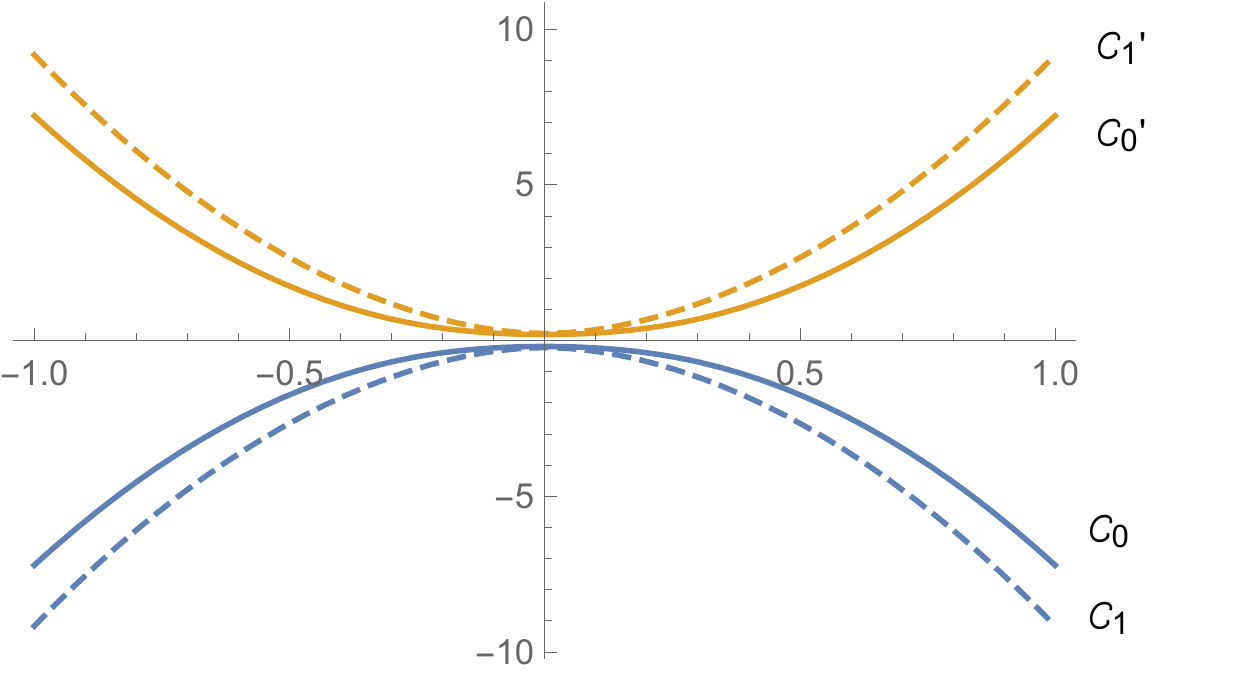}
\caption{$\alpha = -0.5$}
\end{subfigure}%
\begin{subfigure}[t]{0.5\textwidth}
\centering
\includegraphics[width = \textwidth]{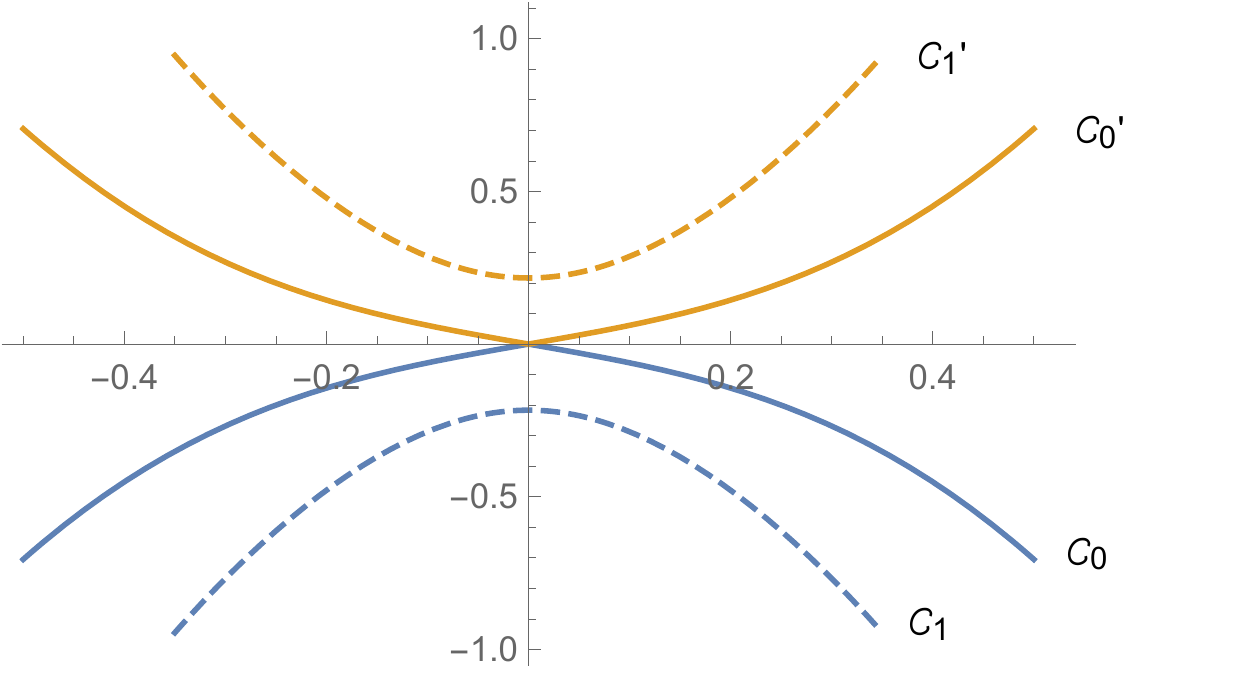}
\caption{$\alpha = -1/\sqrt{2}$}
\end{subfigure}
\vspace{0.3cm}

\begin{subfigure}[t]{0.5\textwidth}
\centering
\includegraphics[width = \textwidth]{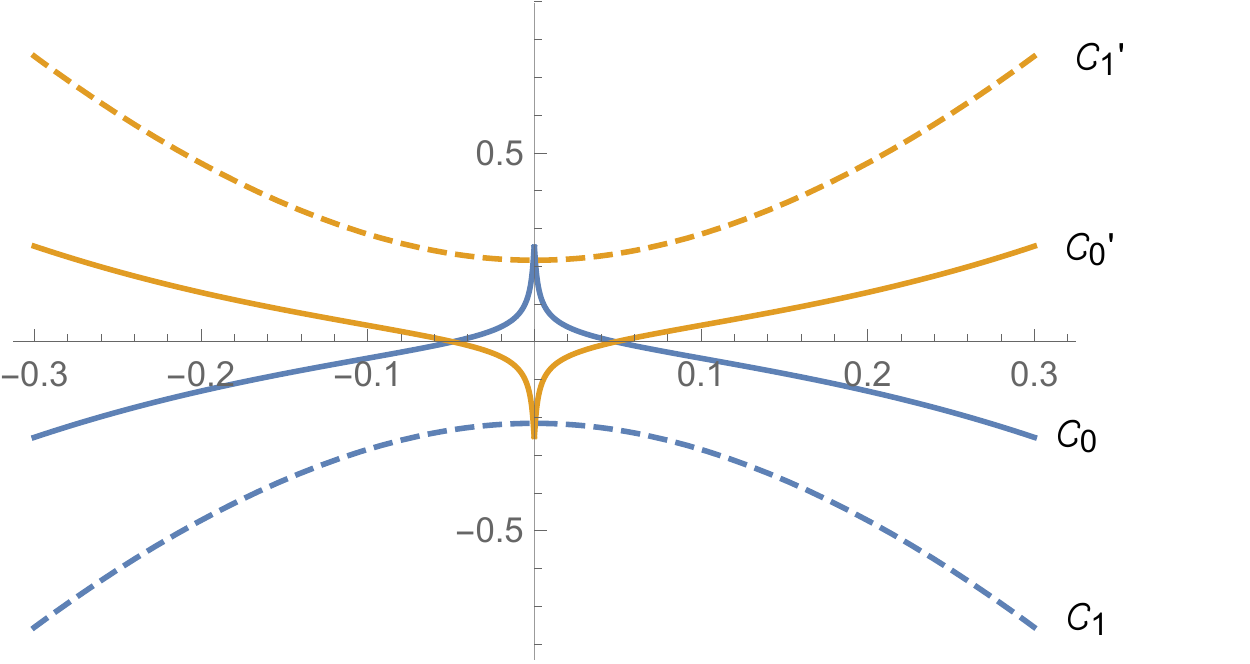}
\caption{$\alpha = -1/\sqrt{2} - 0.01$}
\end{subfigure}%
\begin{subfigure}[t]{0.5\textwidth}
\centering
\includegraphics[width = \textwidth]{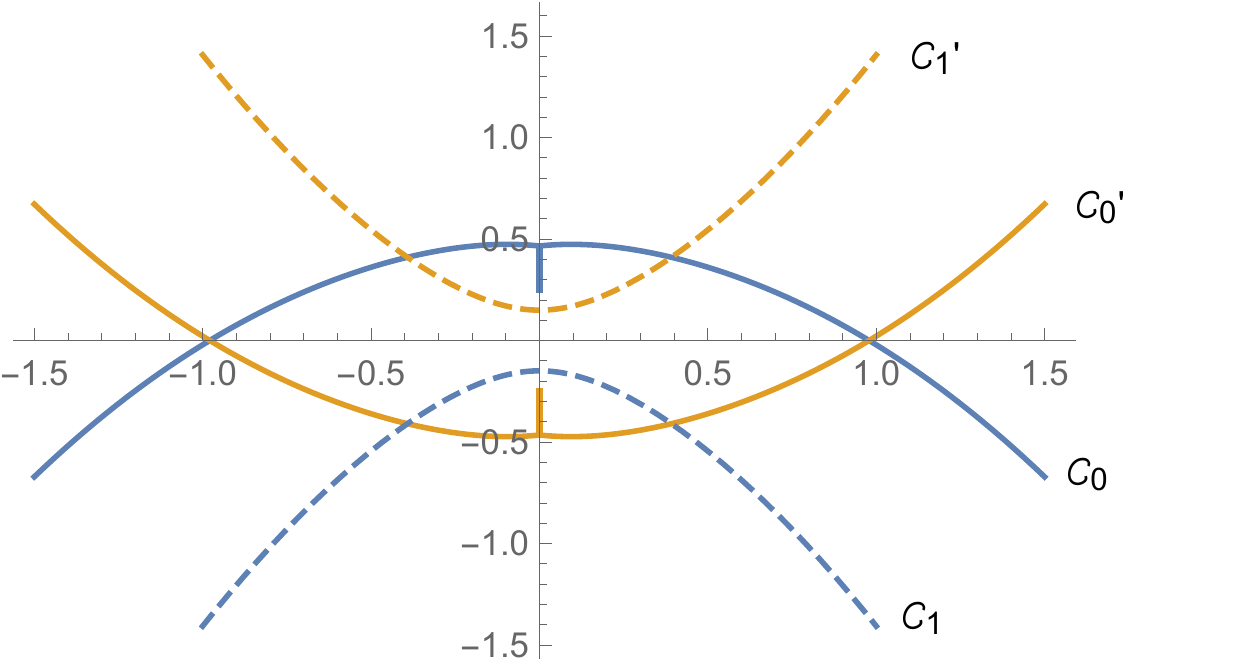}
\caption{$\alpha = -1.8$}
\end{subfigure}
\caption{Contours $\Cmainw$ (blue), $\Cmainz$ (orange), $\Cotherw$ (blue, dashed) and $\Cotherz$ (orange, dashed) for different values of $\alpha$. Different scales are used to show the notable features of each.}
\label{fig:C1C2}
\end{figure}

\begin{remark}
While it may appear from these formulas that there is a clear boundary at $\alpha=-1/\sqrt{2}$ that appears to separate a ordered-like region and a disordered-like region, this is in fact not the case. Once the $\psi(\alpha,\e1,\e2)$ term is taken account of, we see that the first derivative $K_a^{-1}$ at $\alpha = -1/\sqrt{2}$ is in fact continuous. Figure \ref{fig:orderhterms} shows $-I_0(\alpha,\e1,\e2)/(4\pi) - \psi(\alpha, \e1,\e2)$ as a function of $\alpha$ for different values of $(\e1,\e2)$, where $I_0(\alpha,\e1,\e2)$ is taken to be 0 for $-1/\sqrt{2} < \alpha < 1/\sqrt{2}$. 
\end{remark}

\begin{figure}[htbp]
\centering
\includegraphics[width = 0.8\textwidth]{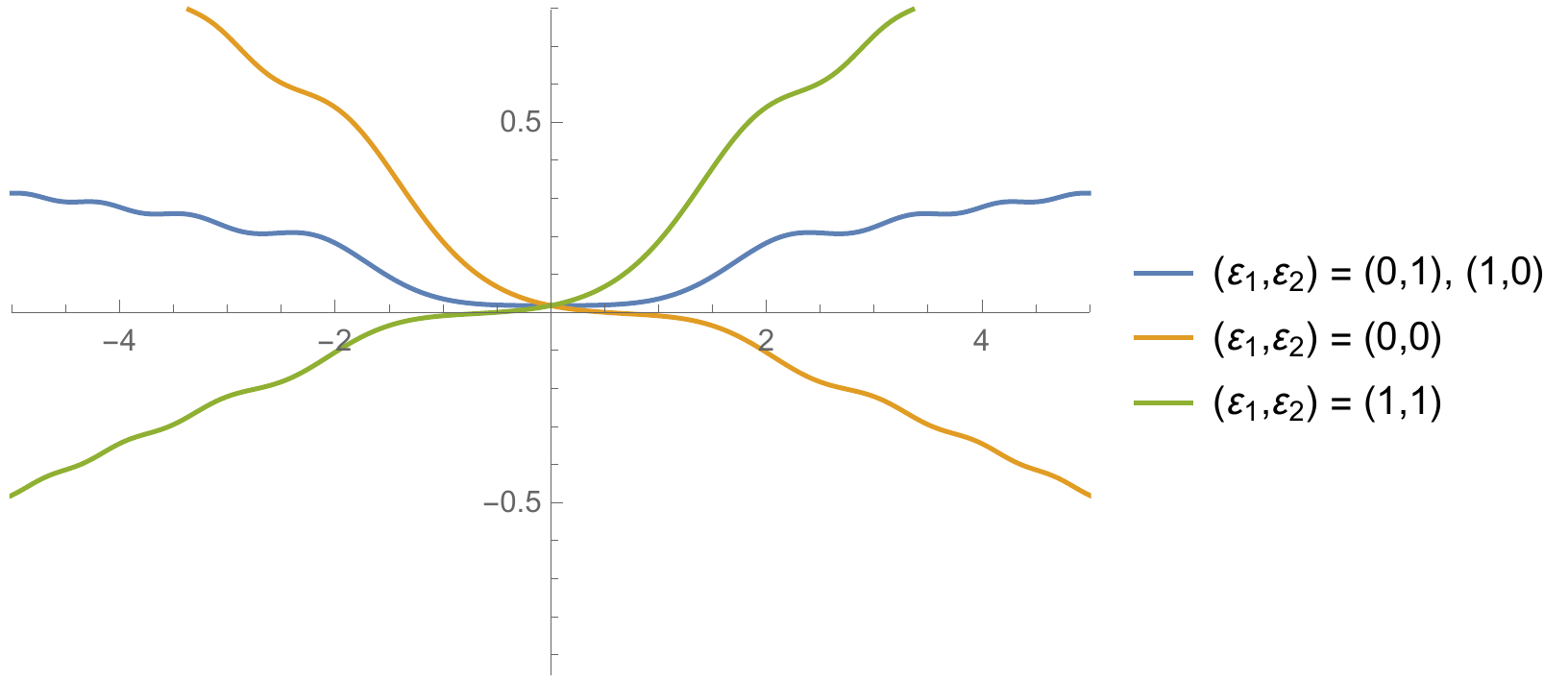}
\caption{$\frac{I_0(\alpha,\e1,\e2)}{4\pi} + \psi(\alpha, \e1,\e2)$ as a function of $\alpha$ for different values of $(\e1,\e2)$, where $I_0(\alpha,\e1,\e2)$ is taken to be 0 for $-1/\sqrt{2} < \alpha < 1/\sqrt{2}$. Up to a sign $\zeta(x,y)$, these are the parts of the coefficients of $\bconst m^{-1/2}$ in Equations~\ref{eq:mainresult1}--\ref{eq:mainresult2} that are not constant as a function of $\alpha$.}
\label{fig:orderhterms}
\end{figure}

From Theorem \ref{thm:mainresult} we can find the asymptotics of the one-point correlation functions $\rho(x,y)$.

\begin{corollary}\label{cor:onepointcor}
Let $(x,y)$ be an edge of $\Gamma$. For $-1/\sqrt{2} \leq \alpha < 0$, if $x = (x_1,x_2) \in \mathtt{W}_\e1$ and $y= (y_1,y_2) \in \mathtt{B}_\e2$ with $\e1,\e2\in \{0,1\}$, we have
\begin{equation}\label{eq:onepointresult1}
\rho(x,y) =  \frac{1}{4} +\zeta (x,y)\bconst m^{-1/2}\bigg( \frac{\log (\bconst m^{-1/2})- 2\log 2}{2\pi}
+ \psi(\alpha,\e1,\e2)\bigg)+ O(m^{-1}\log m )
\end{equation} and for $\alpha < -1/\sqrt{2}$ we have \begin{multline}\label{eq:onepointresult2}
\rho(x,y) =  \frac{1}{4} +\zeta (x,y)\bconst m^{-1/2}\bigg(\frac{\log (\bconst m^{-1/2})- 2\log 2}{2\pi}  \\
+ \frac{I_0(\alpha, \e1, \e2)}{4\pi}+ \psi(\alpha,\e1,\e2)\bigg)+ O(m^{-1}\log m )
\end{multline}
\end{corollary}
The proof is given in Section~\ref{sec:proofofcor}.

\begin{remark}
For $(x,y)$ an edge, we can show that the definition of $\zeta(x,y)$ in Equation \ref{eq:zeta} is equivalent to
\[
\zeta(x,y) = \begin{cases}
1 & \text{if the edge } (x,y) \text{ has weight }a \\
-1 & \text{if the edge } (x,y) \text{ has weight }1 \end{cases}\]
The types of dominos with faces centered at $x$ and $y$ where $\zeta(x,y) = 1$ are shown in Figure~\ref{fig:dominosa} and the dominos where $\zeta(x,y) = -1$ are shown in Figure~\ref{fig:dominos1}.
\end{remark}

\begin{figure}[htbp]
\centering
\vspace{0.1cm}
\begin{subfigure}[t]{0.9\textwidth}
\centering
\includegraphics[width = \textwidth]{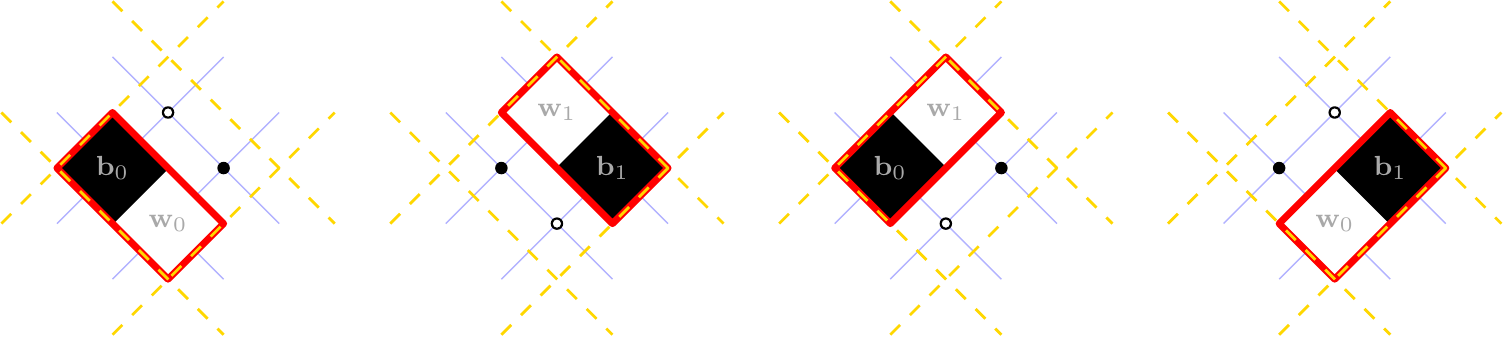}
\caption{Weight $a$ dominos, with $(\e1,\e2) = (0,0)$, $(1,1)$, $(1,0)$ and $(0,1)$ respectively.}\label{fig:dominosa}
\end{subfigure}

\vspace{0.5cm}

\begin{subfigure}[t]{0.9\textwidth}
\centering
\includegraphics[width = \textwidth]{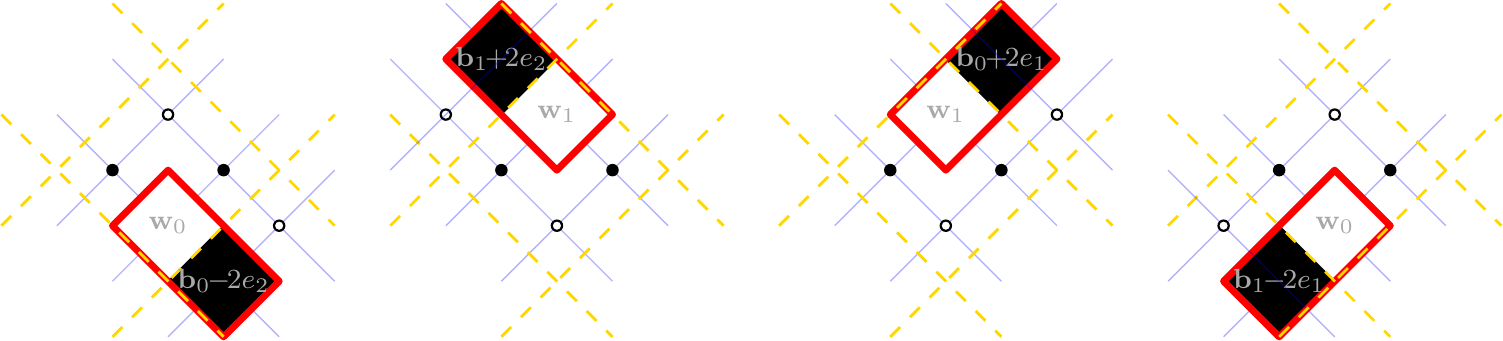}
\caption{Weight 1 dominos, with $(\e1,\e2) = (0,0)$, $(1,1)$, $(1,0)$ and $(0,1)$ respectively.}\label{fig:dominos1}
\end{subfigure}
\caption{Illustration of the eight different types of dominos. The boundaries of the fundamental domains are marked by dashed yellow lines.}
\label{fig:dominotypes}
\end{figure}

The proof of Theorem \ref{thm:mainresult} starts from the double integral formula for the inverse Kasteleyn matrix in \cite{chhita2016domino}. We then find the asymptotics of these integrals, using saddle point analysis when required.

\subsection{Overview of paper}
In Section~\ref{sec:asymptoticresultsandmainthmproof} we first state the formula for the inverse Kastelyn matrix that will form the basis of our asymptotic analysis. We then state the main asymptotic results that we prove in this paper, and use these to prove Theorem~\ref{thm:mainresult} and Corollary~\ref{cor:onepointcor}. In Sections~\ref{sec:Iasymptotics} and \ref{sec:Jasymptotics} we give full technical proofs of these asymptotic results. In Appendix~\ref{sec:moreproofs} we give some proofs that were omitted from the previous sections.

\section{Main asymptotic results and proof of main theorem}\label{sec:asymptoticresultsandmainthmproof}
\subsection{Definitions and inverse Kasteleyn matrix formula from \cite{chhita2016domino}}\label{sec:defs}
We will start our asymptotic analysis with a slight modification of the formula given in \cite{chhita2016domino}. First we make a few definitions. Let
\[c = \frac{1}{a + a^{-1}}.\] Since we are assuming $a \in (0,1)$, we have $c \in (0, 1/2)$. For $\omega \in \mathbb{C}\setminus i[-\sqrt{2c},\sqrt{2c}]$ we define \begin{equation}
\sqrt{\omega^2 + 2c} = i\sqrt{-i(\omega + i\sqrt{2c})}\sqrt{-i(\omega - i \sqrt{2c})}\label{eq:sqrtbranchcut}\end{equation} where the square roots on the right hand side are the principal branch of the square root. These have branch cuts for $\omega = it$ with $t < -\sqrt{2c}$ and $t < \sqrt{2c}$ respectively, but for $t < -\sqrt{2c}$ the branch cuts cancel. Note that when we are using asymptotic variables as in Section~\ref{sec:overviewofresults}, we use the principal branch of the square root, not this branch cut. Define \begin{equation}
G(\omega) = \frac{1}{\sqrt{2c}}(\omega - \sqrt{\omega^2 + 2c}).\label{eq:G}\end{equation} For even $x_1, x_2$ with $0 < x_1,x_2 < 2n$ define \begin{equation}
\widetilde{H}_{x_1,x_2}(\omega) = \frac{\omega^{2m}(-iG(\omega))^{2m - x_1/2}}{(iG(\omega^{-1}))^{2m-x_2/2}}.\label{eq:Htilde}\end{equation} This is slightly different to the $H_{x_1,x_2}$ defined in \cite{chhita2016domino} but will be more convenient for our asymptotic analysis, since $-iG(\omega)$ and $iG(\omega^{-1})$ will be close to 1 near the saddle point. Note that $\widetilde{H}_{x_1,x_2}(\omega)$ depends on $m$ both directly and through the dependence of $G$ on $c$. For $j,k,\e1,\e2 \in \{0,1\}$, define
\begin{equation}
V^{j,k}_{\e1, \e2} (\w1, \w2) = \frac{1}{2}\sum_{\gamma_1,\gamma_2=0}^1 (-1)^{\gamma_2 j + \gamma_1 k}(Q_{\gamma_1,\gamma_2}^{\e1,\e2}(\omega_1,\omega_2) + (-1)^{\e2+1}Q_{\gamma_1,\gamma_2}^{\e1,\e2}(\omega_1,-\omega_2))\label{eq:Voriginal}
\end{equation} where the functions $Q_{\gamma_1,\gamma_2}^{\e1,\e2}(\omega_1,\omega_2)$ are defined as follows. Let \begin{equation}
f_{a,b}(u, v) = (2a^2 uv + 2b^2 uv -ab(-1+u^2)(-1+v^2))
 (2a^2 uv + 2b^2 uv +ab(-1+u^2)(-1+v^2)). \label{eq:fa}
\end{equation} Now we define the following rational functions. We temporarily consider weights $a$ and $b$ where $b$ is not necessarily 1. Let
\begin{align}
\begin{split}
\mathrm{y}_{0,0}^{0,0}(a,b,u,v) =& \frac{1}{4(a^2+b^2)^2 f_{a,b}(u, v)}(2a^7 u^2 v^2 - a^5 b^2(1 + u^4 + u^2v^2 - u^4v^2 + v^4 - u^2v^4) \\
&- a^3 b^4(1+3u^2 + 3v^2 + 2u^2v^2 + u^4v^2 + u^2v^4 - u^4v^4) \\
&- a b^6(1+v^2+u^2+3u^2v^2)) \\
\mathrm{y}_{0,1}^{0,0}(a,b,u,v) =& \frac{a}{4(a^2 + b^2)f_{a,b}(u, v)}(b^2 + a^2u^2)(2a^2v^2 + b^2(1 + v^2 - u^2 + u^2v^2)) \\
\mathrm{y}_{1,0}^{0,0}(a,b,u,v) =& \frac{a}{4(a^2+b^2)f_{a,b}(u, v)}(b^2 + a^2v^2)(2a^2u^2 + b^2(1 - u^2 + v^2 + u^2v^2)) \\
\mathrm{y}_{1,1}^{0,0}(a,b,u,v) =& \frac{a}{4f_{a,b}(u, v)}(2a^2u^2v^2 +b^2(- 1 + v^2 + u^2 + u^2v^2)).\end{split}\label{eq:y00}
\end{align} For $\gamma_1,\gamma_2 \in \{0,1\}$ we define  \begin{align}
\begin{split}
\mathrm{y}_{\gamma_1,\gamma_2}^{0,1}(a,b,u,v) &= \frac{\mathrm{y}_{\gamma_1,\gamma_2}^{0,0}(b,a, u,v^{-1})}{v^2} \\
\mathrm{y}_{\gamma_1,\gamma_2}^{1,0}(a,b,u,v) &= \frac{\mathrm{y}_{\gamma_1,\gamma_2}^{0,0}(b,a,u^{-1},v)}{u^2} \\
\mathrm{y}_{\gamma_1,\gamma_2}^{1,1}(a,b,u,v) &= \frac{\mathrm{y}_{\gamma_1,\gamma_2}^{0,0}(a,b,u^{-1},v^{-1})}{v^2}.
\end{split}\label{eq:otherys}
\end{align} When $b=1$, we write $\mathrm{y}^{\e1,\e2}_{\gamma_1,\gamma_2}(u,v)  = \mathrm{y}^{\e1,\e2}_{\gamma_1,\gamma_2}(a,1,u,v) $. Then define \begin{equation}\mathrm{x}_{\g1, \g2}^{\e1, \e2}(\w1, \w2) = \frac{G(\w1) G(\w2)}{\prod_{i=1}^2 \sqrt{\omega_i^2 +2c} \sqrt{\omega_i^{-2} + 2c}} \textrm{y}_{\g1, \g2}^{\e1, \e2} (G(\w1), G(\w2))(1 - \w1^2\w2^2).\label{eq:xg1g2}\end{equation} and \begin{multline}
Q_{\g1,\g2}^{\e1,\e2}(\w1,\w2) = (-1)^{\e1+\e2+\e1\e2 + \g1(1+\e2) + \g2(1+\e1)}\\
\times t(\w1)^{\g1}t(\w2^{-1})^{\g2} G(\w1)^\e1 G(\w2^{-1})^\e2 \mathrm{x}_{\g1, \g2}^{\e1, \e2}(\w1, \w2^{-1})\label{eq:Q}\end{multline} where $t(\omega)$ is defined by \begin{equation}
t(\omega) = \omega\sqrt{\omega^{-2} + 2c}.\label{eq:t}\end{equation}

Now we define the following functions, which will be the exponential parts of our integrands. For  $x = (x_1,x_2) \in \mathtt{W}_\e1,$ and $y = (y_1,y_2) \in \mathtt{B}_\e2$ with $\e1,\e2\in \{0,1\}$, define \begin{align}
\begin{split}
h_{0,0}(\w1,\w2) &=  \frac{\widetilde{H}_{x_1 + 1, x_2}(\w1)}{\widetilde{H}_{y_1, y_2 + 1}(\w2)}\\
h_{1,0}(\w1,\w2) &= \frac{\widetilde{H}_{x_1 + 1, x_2}(\w1)}{\widetilde{H}_{2n-y_1, y_2 + 1}(\w2)} \\
h_{0,1}(\w1,\w2) &= \frac{\widetilde{H}_{x_1 + 1, 2n - x_2}(\w1)}{\widetilde{H}_{y_1, y_2 + 1}(\w2)} \\
h_{1,1}(\w1,\w2) &=  \frac{\widetilde{H}_{x_1 + 1, 2n - x_2}(\w1)}{\widetilde{H}_{2n-y_1, y_2 + 1}(\w2)}
\end{split}\label{eq:hij}
\end{align} where the terms $\widetilde{H}_{x_1 + 1, x_2}(\w1)$, $\widetilde{H}_{y_1, y_2 + 1}(\w2)$, $\widetilde{H}_{2n-y_1, y_2 + 1}(\w2)$ and $\widetilde{H}_{x_1 + 1, 2n - x_2}(\w1)$ are as defined in Equation~\ref{eq:Htilde}.

Now we can define the main integrals that will appear in the formula from which we start our asymptotic analysis. 

Let $\circlecontour_r$ denote a positively-oriented contour of radius $r$ centered at the origin. For $a < 1$, $\sqrt{2c} < r < 1$ and $x = (x_1,x_2) \in \mathtt{W}_\e1,\, y = (y_1,y_2) \in \mathtt{B}_\e2$ with $\e1,\e2\in \{0,1\}$ define \begin{align}\label{eq:mathcalI}
\mathcal{I}^{j,k}_{\e1,\e2} (a, x_1, x_2, y_1, y_2) &= \frac{i^{y_1-x_1}}{(2\pi i)^2} \int_{\circlecontour_r} \frac{d\w1}{\w1} \int_{\circlecontour_{1/r}}d\w2 \frac{V_{\e1, \e2}^{j,k} (\w1, \w2)}{\w2 - \w1} h_{j,k}(\w1,\w2)
\end{align} where $V_{\e1, \e2}^{j,k} (\w1, \w2)$ is defined in Equation~\ref{eq:Voriginal} and $h_{j,k}(\w1,\w2)$ is defined in Equation~\ref{eq:hij}. 

Finally, we state a formula for the whole plane inverse Kasteleyn matrix in the presence of a magnetic field \cite{kenyon_okounkov_sheffield_2006}. First, let \[
\mathcal{K}_a(z,w) = \begin{pmatrix}
i(a + w^{-1}) & a+z \\
a + z^{-1} & i(a+w)
\end{pmatrix}.
\] This is the ``magnetically altered'' Kasteleyn matrix \cite{kenyon_okounkov_sheffield_2006} for the fundamental domain with weights $a$ and 1 where $1/z$ is the multiplicative factor when crossing to a fundamental domain in the direction $e_1$ and $1/w$ is the multiplicative factor when crossing to a fundamental domain in the direction $e_2$ \cite{chhita2016domino}. See Figure~\ref{fig:aztec diamond diagram} for a diagram showing the fundamental domains. The inverse $(\mathcal{K}_a(z,w))^{-1}$ appears in the formula for the whole plane inverse Kasteleyn matrix $\mathbb{K}_{a,H,V}^{-1}(x,y)$ (Theorem~\ref{theorem:planeinverseoriginal} below), which in turn appears in the formula for the Kasteleyn matrix for the two-periodic weighted Aztec diamond $K_a^{-1}(x,y)$. The characteristic polynomial \cite{kenyon_okounkov_sheffield_2006} $P_a(z,w)$ is the determinant of $\mathcal{K}_a(z,w)$, explicitly given by
\begin{equation}\label{eq:Pazw}
P_a(z,w) = -2 -2a^2 - aw^{-1} - aw - az^{-1} - az.
\end{equation}

Let $\mathbb{K}_{a,H,V}^{-1}(x,y)$ denote the whole plane inverse Kasteleyn matrix for the entries $x$ and $y$ with magnetic field $(H,V)$ and weights $a<1$ and 1. Then, following from \cite{kenyon_okounkov_sheffield_2006}, we have the following theorem.
\begin{theorem}[\cite{kenyon_okounkov_sheffield_2006}]\label{thm:translationinvariantoriginal} Let $\circlecontour_r$ denote a contour of radius $r$ centered at the origin, traversed in a counter-clockwise direction. Then for $\mathbf{w} = (w_1,w_2) \in \mathtt{W}_\e1$ and $\mathbf{b} = (b_1,b_2) \in \mathtt{B}_\e2$ in the same fundamental domain, where $\e1,\e2\in \{0,1\}$, and $u, v \in \mathbb{Z}$ we have \begin{equation}
\label{eq:translationinvariantoriginal}
\mathbb{K}_{a,H,V}^{-1}(\mathbf{w}, \mathbf{b} + 2ue_1 + 2ve_2) = \frac{1}{(2\pi i)^2} \int_{\circlecontour_{e^H}}\frac{dz}{z}\int_{\circlecontour_{e^V}}\frac{dw}{w} (\mathcal{K}_a(z, w)^{-1})_{\e1\e2} z^u w^v 
\end{equation}
where for convenience rows and columns of the $2\times 2$ matrix $\mathcal{K}_a(z, w)^{-1}$ are indexed by 0 and 1.
\label{theorem:planeinverseoriginal}\end{theorem} Here we have \begin{equation}
\label{eq:Kazwinverse}
\mathcal{K}_a(z,w)^{-1} = \frac{1}{P_a(z,w)}\begin{pmatrix}
i(a+w) & -(a+z) \\
-(a+z^{-1}) & i(a+w^{-1}).
\end{pmatrix}
 \end{equation}

Now we can state the formula that forms the starting point for our asymptotic analysis. This is a slight modification of the formula from \cite{chhita2016domino}, which is in turn a simplification of the formula from \cite{chhita2014}.

\begin{theorem}\label{theorem:Ka1} For $n = 4m$ and $0 <a < 1$, take $x = (x_1,x_2) \in \mathtt{W}_\e1,\, y = (y_1,y_2) \in \mathtt{B}_\e2$ with $\e1,\e2 \in \{0,1\}$. Then the entries of the inverse Kasteleyn matrix $K_a^{-1}$ are given by
\begin{multline}
K_a^{-1}((x_1,x_2),(y_1,y_2)) = \mathbb{K}_{a,0,0}^{-1}((x_1,x_2),(y_1,y_2)) - \Big(\mathcal{I}^{0,0}_{\e1,\e2} (a, x_1, x_2, y_1, y_2) \\
-\mathcal{I}^{1,0}_{\e1,\e2}(a, x_1, x_2, y_1, y_2) - \mathcal{I}^{0,1}_{\e1,\e2} (a, x_1, x_2, y_1, y_2) + \mathcal{I}^{1,1}_{\e1,\e2} (a, x_1, x_2, y_1, y_2)\Big)\label{eq:Ka1inverse}
\end{multline}
where $\mathbb{K}_{a,0,0}^{-1}((x_1,x_2),(y_1,y_2))$ is defined in Equation~\ref{eq:translationinvariantoriginal} and $\mathcal{I}^{j,k}_{\e1,\e2} (a, x_1, x_2, y_1, y_2)$ is defined in Equation~\ref{eq:mathcalI}.
\end{theorem} 
The proof of Theorem~\ref{theorem:Ka1} is given in Appendix~\ref{sec:formuladerivation}.

\subsection{Statement of asymptotic results}
First we state the asymptotic expansion for $\mathcal{I}^{j,k}_{\e1,\e2} (a, x_1, x_2, y_1, y_2)$. Recall the functions \[f^\pm(w) = \sqrt{1/2-2iw} \pm \sqrt{1/2+2iw}\] defined in Equation~\ref{eq:fpm}. For fixed $\alpha_x,\alpha_y < 0$, define \begin{align}
\begin{split}
g_{0,0}(w,z) &= \bconst^2(-2i(w-z) + \alpha_x f^-(w) - \alpha_y f^-(z)) \\
g_{1,0}(w,z) &= \bconst^2(-2i(w-z) + \alpha_x f^-(w) + \alpha_y f^+(z)) \\
g_{0,1}(w,z) &= \bconst^2(-2i(w-z) + \alpha_x f^+(w) - \alpha_y f^-(z)) \\
g_{1,1}(w,z) &= \bconst^2(-2i(w-z) + \alpha_x f^+(w) + \alpha_y f^+(z))
\end{split}\label{eq:gij}
\end{align} We consider the case where $\alpha_x = \alpha_y = \alpha$ is the asymptotic coordinate from Equation~\ref{eq:asymptoticcoord}.
Our first important asymptotic result is
\begin{theorem}\label{thm:Iasymptotics}
 Recall the integrals $\mathcal{I}^{j,k}_{\e1,\e2} (a, x_1, x_2, y_1, y_2)$ defined in Equation~\ref{eq:mathcalI}. Let the contours $\Cmainw$, $\Cotherw$, $\Cmainz$ and $\Cotherz$ be as defined in Section~\ref{sec:overviewofresults}. Recall the functions $g_{j,k}(w,z)$ defined in Equation~\ref{eq:gij} and the functions $A^{j,k}_{\e1,\e2}(w,z)$ defined in Equation~\ref{eq:Aall}. Then for $-1/\sqrt{2} \leq \alpha < 0$,
  \begin{equation}
 \mathcal{I}^{0,0}_{\e1,\e2} (a, x_1, x_2, y_1, y_2) =  \frac{\zeta(x,y) \bconst m^{-1/2}}{8(2\pi i)^2\Sigma(x,y)}\int_{\Cmainw}dw \int_{\Cmainz}dz \frac{A^{0,0}_{\e1,\e2}(w,z)}{i(z-w)} e^{g_{0,0}(w,z)} + O(m^{-1}),
 \end{equation} and for $\alpha < -1/\sqrt{2}$,
 \begin{multline}
 \mathcal{I}^{0,0}_{\e1,\e2} (a, x_1, x_2, y_1, y_2) =  \frac{\zeta(x,y) \bconst m^{-1/2}}{8(2\pi i)^2\Sigma(x,y)}\\
\times\Bigg( \int_{\Cmainw}dw \int_{\Cmainz}dz \frac{A^{0,0}_{\e1,\e2}(w,z)}{i(z-w)} e^{g_{0,0}(w,z)} - 2\pi  \int_{-\eta}^{\eta} A^{0,0}_{\e1,\e2}(w,w) dw\Bigg) + O(m^{-1}).
 \end{multline} For $(j,k) \neq (0,0)$ for any $\alpha < 0$, 
 \begin{equation}
 \mathcal{I}^{j,k}_{\e1,\e2} (a, x_1, x_2, y_1, y_2) =  \frac{\zeta(x,y) \bconst m^{-1/2}}{8(2\pi i)^2\Sigma(x,y)} \int_{\mathcal{C}_j}dw \int_{\mathcal{C}_k'}dz \frac{A^{j,k}_{\e1,\e2}(w,z)}{i(z-w)} e^{g_{j,k}(w,z)} + O(m^{-1})
 \end{equation}
\end{theorem}
The proof can be found in Section~\ref{sec:integralformulas}.

Now we state the asymptotic expansion for $\mathbb{K}_{a,0,0}^{-1}(x,y)$. For $-1 \leq z \leq -3+2\sqrt{2}$, let \[
 \theta(z) = \frac{1}{2z} \left(i \sqrt{4z^2 - (4z + z^2 + 1)^2} - (4z + z^2+1)\right),
 \] and define \[
 k^{(v)}(z) = \frac{\theta(z)^v + \theta(z)^{-v}}{2}.
 \] Also define \begin{equation}b(z) = -\frac{1}{2}\log\left(\frac{4\sqrt{2}(1+z)}{1-z + \sqrt{-1 -6z - z^2}}\right).\end{equation}   We define coefficients $c_i(pe_1 + qe_2)$ for $p,q\in \mathbb{Z}$. When $p$ is odd and $q$ is even, define 
\begin{align}c_0(pe_1 + qe_2) &= \frac{1}{\pi}\int_{-1}^{-3 + 2\sqrt{2}} \frac{z^{(|p|-1)/2}k^{(q/2)}(z)}{\sqrt{-1-6z-z^2}}dz\label{eq:c0}\\
c_1(pe_1 + qe_2) &= \frac{1}{2\pi}\label{eq:c1}\\
c_2(pe_1 + qe_2) &= \frac{(-1)^{q/2}}{\pi}\Bigg(\int_{-1}^{-3 + 2\sqrt{2}} \left(\frac{1}{2}(-z)^{(|p|-1)/2} + \sum_{i=0}^{(|p|-3)/2}(-z)^{i}\right) \frac{k^{(q/2)}(z)}{\sqrt{-1-6z-z^2}}dz\label{eq:c2} \\
 &\hspace{1cm}-\log 2  - \int_{-1}^{-3 + 2\sqrt{2}} \frac{dk^{(q/2)}}{dz}(z) b(z)dz  \Bigg).\nonumber\end{align}When $p$ is even and $q$ is odd, define \[c_i(pe_1 + qe_2) = c_i(qe_1 + pe_2)\] for $i=0,1,2$.
Our second important asymptotic result is
\begin{theorem}\label{thm:Kgasinversefullasymptotics}
Take $x \in \mathtt{W}$ and $y\in \mathtt{B}$. Let $\zeta(x,y)$ and $\Sigma(x,y)$ be as defined in Equations~\ref{eq:zeta} and \ref{eq:Sigmaxy} respectively, and $c_i$ as defined above in Equations~\ref{eq:c0}--\ref{eq:c2}. Then \begin{multline*}
\mathbb{K}_{a,0,0}^{-1}(x,y) 
= \frac{1}{\Sigma(x,y)}\Big(c_0(y-x) + c_0(y-x) \bconst m^{-1/2}/2 \\
+ \zeta(x,y)(c_1(y-x) \bconst m^{-1/2} \log \bconst m^{-1/2}  + c_2(y-x) \bconst m^{-1/2} )\Big) + O(m^{-1}\log m).
 \end{multline*}
\end{theorem}

The proof can be found in Section~\ref{sec:Kgasinverseasym}. This, together with Theorem~\ref{thm:Iasymptotics} are all that we need to derive the main result.

\subsection{Proof of Theorem \ref{thm:mainresult}}\label{sec:proofofmainresult}
Here we prove the main result.
\begin{proof}
By Theorem~\ref{theorem:Ka1} we have \[ K_a^{-1}(x,y) = \mathbb{K}_{a,0,0}^{-1}(x,y) - \sum_{j,k=0}^1 (-1)^{j+k}\mathcal{I}^{j,k}_{\e1,\e2} (a, x_1, x_2, y_1, y_2).\] The asymptotic expansion of the first term is given in Theorem \ref{thm:Kgasinversefullasymptotics}, and the asymptotic expansions of the remaining terms are given in Theorem~\ref{thm:Iasymptotics}. Comparing the formulas in Theorems~\ref{thm:Iasymptotics} with the definitions of $I_k(\alpha,\e1,\e2)$ in Equations~\ref{eq:I1to4}--\ref{eq:I0}, and with the formula for $A_{\e1,\e2}^{0,0}(w,w)$ in Equation~\ref{eq:Aww} we see that \begin{multline*}
\mathcal{I}^{0,0}_{\e1,\e2} (a, x_1, x_2, y_1, y_2) \\
= \begin{dcases} -\frac{\bconst m^{-1/2}}{ 32 \pi^2}\frac{\zeta(x,y)}{K_1(y,x)} I_1(\alpha,\e1,\e2) + O(m^{-1}) & \text{ if } -1/\sqrt{2} \leq \alpha < 0 \\
 -\frac{\bconst m^{-1/2}}{ 32 \pi^2}\frac{\zeta(x,y)}{K_1(y,x)} (I_1(\alpha,\e1,\e2) + 8\pi I_0(\alpha, \e1,\e2))+ O(m^{-1}) & \text{ if } \alpha < -1/\sqrt{2}\end{dcases} 
\end{multline*} and \begin{align*}
\mathcal{I}^{1,0}_{\e1,\e2} (a, x_1, x_2, y_1, y_2) &= -\frac{\bconst m^{-1/2}}{ 32 \pi^2}\frac{\zeta(x,y)}{K_1(y,x)} I_2(\alpha,\e1,\e2)+ O(m^{-1}) \\
\mathcal{I}^{0,1}_{\e1,\e2} (a, x_1, x_2, y_1, y_2) &= -\frac{\bconst m^{-1/2}}{ 32 \pi^2}\frac{\zeta(x,y)}{K_1(y,x)} I_3(\alpha,\e1,\e2) + O(m^{-1})\\
\mathcal{I}^{1,1}_{\e1,\e2} (a, x_1, x_2, y_1, y_2) &= -\frac{\bconst m^{-1/2}}{ 32 \pi^2}\frac{\zeta(x,y)}{K_1(y,x)} I_4(\alpha,\e1,\e2)+ O(m^{-1})
\end{align*} So we see that for $-1/\sqrt{2} \leq \alpha < 0$, we have \begin{multline*}
 \sum_{j,k=0}^1 (-1)^{j+k}\mathcal{I}^{j,k}_{\e1,\e2} (a, x_1, x_2, y_1, y_2) = -\frac{\zeta(x,y)\bconst m^{-1/2}}{ 32 \pi^2\Sigma(x,y)}(I_1(\alpha,\e1,\e2)\\
 - I_2(\alpha,\e1,\e2) - I_3(\alpha,\e1,\e2) + I_4(\alpha,\e1,\e2)) + O(m^{-1})
\end{multline*} and for $\alpha < -1/\sqrt{2}$, we have \begin{multline*}
 \sum_{j,k=0}^1 (-1)^{j+k}\mathcal{I}^{j,k}_{\e1,\e2} (a, x_1, x_2, y_1, y_2) = -\frac{\zeta(x,y)\bconst m^{-1/2}}{ 32 \pi^2\Sigma(x,y)}(8\pi I_0(\alpha,\e1,\e2) \\
+ I_1(\alpha,\e1,\e2)  - I_2(\alpha,\e1,\e2) - I_3(\alpha,\e1,\e2) + I_4(\alpha,\e1,\e2)) + O(m^{-1})
\end{multline*} 
From Theorem~\ref{thm:Kgasinversefullasymptotics} we have
\begin{multline*}
\mathbb{K}_{a,0,0}^{-1}(x,y) 
= \frac{1}{\Sigma(x,y)}\Bigg(c_0(y-x) + c_0(y-x) \bconst m^{-1/2}/2 \\
+ \zeta(x,y)\left(\frac{\bconst m^{-1/2} \log \bconst m^{-1/2}}{2\pi}  + c_2(y-x) \bconst m^{-1/2} \right)\Bigg) + O(m^{-1}\log m).
 \end{multline*} Putting together these two formulas and comparing with the definition of $\psi(\alpha, \e1, \e2)$ in Equation~\ref{eq:psixy} gives the result.

\end{proof}

\subsection{Proof of Corollary \ref{cor:onepointcor}}\label{sec:proofofcor}

We now give a proof of Corollary~\ref{cor:onepointcor} on one-point correlation functions.

\begin{proof}[Proof of Corollary~\ref{cor:onepointcor}]
We start from the formula in Theorem~\ref{thm:mainresult} and compute the coefficients $c_0(y-x)$ and $c_2(y-x)$ when $(x,y)$ is an edge. We have $y-x = \pm e_1$ or $\pm e_2$. It is clear from the definitions that $c_i(e_1) = c_i(-e_1) = c_i(e_2) = c_i(-e_2)$. So we compute $c_i(e_1)$, i.e. we set $p=1$ and $q=0$ in Equations~\ref{eq:c0} and \ref{eq:c2}. We have
\[c_0(e_1 ) = \frac{1}{\pi}\int_{-1}^{-3 + 2\sqrt{2}} \frac{1}{\sqrt{-1-6z-z^2}}dz =\frac{1}{4}\] and
\[
c_2(e_1) = \frac{1}{\pi}\Bigg(\int_{-1}^{-3 + 2\sqrt{2}}  \frac{1}{2\sqrt{-1-6z-z^2}}dz -\log 2  \Bigg) = \frac{1}{8} - \frac{\log 2}{\pi}.\] Next we note that $\rho(x,y) = K_a(y,x)K_a^{-1}(x,y)$, where $K_a(y,x) = w(x,y)\epsilon(x,y)$ with $w(x,y)$ the weight of $(x,y)$ and $\epsilon(x,y)$ the Kasteleyn-Percus orientation. We have $\epsilon(x,y) = \Sigma(x,y)$, and \[
w(x,y) = \begin{cases} a &\text{if } \zeta(x,y) = 1 \\
 1 &\text{if } \zeta(x,y) = -1 \end{cases}.\] So we can write \[K_a(y,x) = \Sigma(x,y)\left(1 - \frac{\bconst m^{-1/2}}{2}(1+\zeta(x,y)) \right)+O(m^{-1})\]
So for $(x,y)$ an edge of $\Gamma$, we have
\begin{multline}
\rho(x,y) = \frac{1}{4}\left(1 + \frac{\bconst m^{-1/2}}{2} -  \frac{\bconst m^{-1/2}}{2}(1+\zeta(x,y))  \right) \\
+\zeta (x,y)\bconst m^{-1/2}\left(\frac{ \log (\bconst m^{-1/2})}{2\pi} + \frac{1}{8} - \frac{\log 2}{\pi}
+ \psi(\alpha,\e1,\e2)\right) + O(m^{-1}\log m )
\end{multline} $-1/\sqrt{2} \leq \alpha < 0$ and 
\begin{multline}
\rho(x,y) = \frac{1}{4}\left(1 + \frac{\bconst m^{-1/2}}{2} -  \frac{\bconst m^{-1/2}}{2}(1+\zeta(x,y))\right) \\
+\zeta (x,y)\bconst m^{-1/2}\bigg(\frac{\log (\bconst m^{-1/2})}{2\pi} + \frac{1}{8} - \frac{\log 2}{\pi} + \frac{I_0(\alpha, \e1, \e2)}{4\pi}
+ \psi(\alpha,\e1,\e2)\bigg)+ O(m^{-1}\log m )
\end{multline} for $\alpha < -1/\sqrt{2}$.
We simplify to get the result.
\end{proof}

\section{Asymptotics of $\mathcal{I}^{j,k}_{\e1,\e2} (a, x_1, x_2, y_1, y_2)$}\label{sec:Iasymptotics}
In this section, we compute the asymptotics of  $\mathcal{I}^{j,k}_{\e1,\e2} (a, x_1, x_2, y_1, y_2)$ in the limit $m \rightarrow \infty$, with weights 1 and $a = 1 - \bconst m^{-1/2}$. We finish with a proof of Theorem~\ref{thm:Iasymptotics}. First, we must identify the main asymptotic term in $m$ in the integrand of each integral. The only term in each integral that depends directly on $m$ is $h_{j,k}(\w1,\w2)$ but we also have $V_{\e1, \e2}^{j,k} (\w1, \w2)$ depending on $m$ through $a$. However, this dependence is not exponential, so this term is not relevant for our saddle point analysis. The dependence of $h_{j,k}(\w1,\w2)$ on $m$ is a bit more complicated, and we will look at it more carefully.

\subsection{Asymptotic co-ordinates}
For $x\in \mathtt{W}$ and $y\in \mathtt{B}$ we write
\begin{align}
\begin{split}
x_1 &= [4m + 2m^{1/2}\alpha_x \bconst] + \overline{x_1} \\ 
x_2 &= [4m + 2m^{1/2}\alpha_x \bconst] + \overline{x_2} \\ 
y_1 &= [4m + 2m^{1/2}\alpha_y \bconst] + \overline{y_1} \\ 
y_2 &= [4m + 2m^{1/2}\alpha_y \bconst] + \overline{y_2},
\end{split}\label{eq:asymptotic_coords_xy}
\end{align}
where $\alpha_x,\alpha_y < 0$ and the integer parts $\overline{x_1}, \overline{x_2}, \overline{y_1}$ and $\overline{y_2}$ are order 1. Recall that $\bconst > 0$ is a constant and $a = 1 - \bconst m^{-1/2}$. For the moment the asymptotic coordinates $\alpha_x$ and $\alpha_y$ are not necessarily equal, though in the final result we only consider $\alpha_x = \alpha_y = \alpha$ as in Equation~\ref{eq:asymptoticcoord}. In a future paper, we will look at the two-point correlations where $\alpha_x \neq \alpha_y$. 

\subsection{Approximate location of saddle points}

Let $x =(x_1, x_2) \in \mathtt{W}$ and $y = (y_1,y_2) \in \mathtt{B}$ be as in Equation~\ref{eq:asymptotic_coords_xy}. To find the contours of steepest ascent and descent for the four integrals $\mathcal{I}^{j,k}_{\e1,\e2} (a, x_1, x_2, y_1, y_2)$ with $j,k \in \{0,1\}$ we need to analyze the functions $\widetilde{H}_{x_1 + 1, x_2}(\w1)$, $\widetilde{H}_{y_1, y_2 + 1}(\w2)$, $\widetilde{H}_{2n-y_1, y_2 + 1}(\w2)$ and $\widetilde{H}_{x_1 + 1, 2n - x_2}(\w1)$. 

First we look at the function $\widetilde{H}_{x_1 + 1, x_2}(\omega)$. 

 \begin{lemma}\label{lemma:saddlepointlocation}
 Let $(x_1,x_2) \in \mathtt{W}$ be such that $x_i = [4m + 2m^{1/2}\alpha \bconst] + \overline{x_i}$ with $\alpha < 0$, and with the integer parts $\overline{x_i}$ of order 1. The saddle points of $\log \widetilde{H}_{x_1+1, x_2}(\omega)$ that are bounded as $m \rightarrow \infty$ occur at $\omega = \pm i + O(m^{-1})$.
 \end{lemma} 
 \begin{proof}We have
 \[
\widetilde{H}_{x_1+1, x_2}(\omega) = \frac{\omega^{2m} (-iG(\omega))^{-\alpha \bconst m^{1/2}+O(1)}}{(i G(\omega^{-1}))^{-\alpha \bconst m^{1/2}+O(1)}}.\] So \[
\log \widetilde{H}_{x_1+1, x_2}(\omega) = m\big(2\log\omega + m^{-1/2}(-\alpha \bconst \log (-iG(\omega)) + \alpha \bconst \log (iG(\omega^{-1}))) + O(m^{-1})\big).\]Note that $G(\omega)$ and $G(\omega^{-1})$ also depend on $m$ through $c = 1/(a + a^{-1})$. We want to find the saddle points of this function that are bounded at $m \rightarrow \infty$. Let \[
g_{\alpha}(\omega) = 2\log\omega + m^{-1/2}(-\alpha \bconst \log (-iG(\omega)) + \alpha \bconst \log (iG(\omega^{-1})))),
\]
 so $\log \widetilde{H}_{x_1+1, x_2}(\omega) = m(g_{\alpha}(\omega) + O(m^{-1}))$. Then $\log \widetilde{H}_{x_1+1, x_2}(\omega)$ has a saddle point when \begin{equation} g_{\alpha}'(\omega) = \frac{1}{\omega}\left(2 + \alpha \bconst  m^{-1/2}\left( \frac{\omega}{\sqrt{\omega^2 + 2c}} + \frac{\omega^{-1}}{ \sqrt{\omega^{-2} + 2c}}\right)\right) = O(m^{-1})\label{eq:saddlederivative}\end{equation} recalling that the square roots are defined in Equation~\ref{eq:sqrtbranchcut}.

We are looking for saddle points $\omega_c$ that approach a finite limit as $m \rightarrow 0$. So we need solutions to \[
  2 + \alpha \bconst m^{-1/2} \left( \frac{\omega}{\sqrt{\omega^2 + 2c}} + \frac{\omega^{-1}}{ \sqrt{\omega^{-2} + 2c}}\right) = O(m^{-1}).
 \] By exchanging $\omega$ and $\omega^{-1}$, it is clear that if $\omega$ is a solution then so is $\omega^{-1}$. We also see that we must have either $\omega^2 + 2c = O(m^{-1})$ or $\omega^{-2} + 2c = O(m^{-1})$. We compute \begin{align}\begin{split}
 2c &= \frac{2}{a + a^{-1}} \\
 &= \frac{1}{1 + \bconst^2 m^{-1}/2 + O(m^{-3/2})} \\
 &= 1 - \bconst^2 m^{-1}/2 + O(m^{-3/2}).
 \end{split}
 \end{align}
Therefore the saddle points occur when $\omega^2 = -1 \pm O(m^{-1})$ (as this coincides with $\omega^{-2} = -1 \pm O(m^{-1})$), so $\omega = \pm i + O(m^{-1})$. 
 \end{proof}
 
We have a similar lemma for $\widetilde{H}_{2n-y_1, y_2 + 1}(\omega)$.
  
  \begin{lemma}\label{lemma:saddlepointlocation1}
 Let $(y_1,y_2) \in \mathtt{B}$ be such that $y_i = [4m + 2m^{1/2}\alpha \bconst] + \overline{y_i}$ with $\alpha < 0$ and the integer parts $\overline{y_i}$ order 1. The saddle points of $\log \widetilde{H}_{2n-y_1, y_2 + 1}(\omega)$ that are bounded as $m \rightarrow \infty$ occur at $\omega = \pm i + O(m^{-1})$.
 \end{lemma}
 \begin{proof}We have
 \[
\widetilde{H}_{2n-y_1, y_2 + 1}(\omega) = \frac{\omega^{2m} (-iG(\omega))^{\alpha \bconst m^{1/2}+O(1)}}{(i G(\omega^{-1}))^{-\alpha \bconst m^{1/2}+O(1)}}.\] So \[
\log \widetilde{H}_{2n-y_1, y_2 + 1}(\omega) = m\big(2\log\omega + m^{-1/2}(\alpha \bconst \log (-iG(\omega)) + \alpha \bconst \log (iG(\omega^{-1}))) + O(m^{-1})\big).\]  Let \[
g_{\alpha}^{(1)}(\omega) = 2\log\omega + m^{-1/2}(\alpha \bconst \log (-iG(\omega)) + \alpha \bconst \log (iG(\omega^{-1}))))
\] so $\log \widetilde{H}_{2n-y_1, y_2 + 1}(\omega) = m(g_{\alpha}^{(1)}(\omega) + O(m^{-1}))$.
Then $\log \widetilde{H}_{2n-y_1, y_2 + 1}(\omega)$ has a saddle point when \begin{equation*} (g_{\alpha}^{(1)})'(\omega) = \frac{1}{\omega}\left(2 + \alpha \bconst  m^{-1/2}\left( -\frac{\omega}{\sqrt{\omega^2 + 2c}} + \frac{\omega^{-1}}{ \sqrt{\omega^{-2} + 2c}}\right)\right) = O(m^{-1}).\end{equation*}
The proof is finished in exactly the same way as Lemma~\ref{lemma:saddlepointlocation}. The only difference is the sign of the first term.
\end{proof} 

Finally we prove a similar lemma for $\widetilde{H}_{x_1 + 1, 2n - x_2}(\omega)$. 

\begin{lemma}\label{lemma:saddlepointlocation2}
 Let $(x_1,x_2) \in \mathtt{W}$ be such that $x_i = [4m + 2m^{1/2}\alpha \bconst] + \overline{x_i}$ with $\alpha < 0$ and the integer parts $\overline{x_i}$ order 1. The saddle points of $\log \widetilde{H}_{x_1+1, 2n-x_2}(\omega)$ that are bounded as $m \rightarrow \infty$ occur at $\omega = \pm i + O(m^{-1})$.
 \end{lemma} 
 \begin{proof}
 \[
\widetilde{H}_{x_1+1, 2n-x_2}(\omega) = \frac{\omega^{2m} (-iG(\omega))^{-\alpha \bconst m^{1/2}+O(1)}}{(i G(\omega^{-1}))^{\alpha \bconst m^{1/2}+O(1)}}.\] So \[
\log \widetilde{H}_{x_1+1, 2n - x_2}(\omega) = m\big(2\log\omega 
+ m^{-1/2}(-\alpha \bconst \log (-iG(\omega)) - \alpha \bconst \log (iG(\omega^{-1}))) + O(m^{-1})\big).\] We want to find the saddle points of this function that are bounded at $m \rightarrow \infty$. Let \[
g_{\alpha}^{(2)}(\omega) = 2\log\omega + m^{-1/2}(-\alpha \bconst \log (-iG(\omega)) - \alpha \bconst \log (iG(\omega^{-1})))),
\]
 so $\log \widetilde{H}_{x_1+1, 2n - x_2}(\omega) = m(g_{\alpha}^{(2)}(\omega) + O(m^{-1}))$. Then $\log \widetilde{H}_{x_1+1,2n- x_2}(\omega)$ has a saddle point when \begin{equation*} (g_{\alpha}^{(2)})'(\omega) = \frac{1}{\omega}\left(2 + \alpha \bconst  m^{-1/2}\left( \frac{\omega}{\sqrt{\omega^2 + 2c}} - \frac{\omega^{-1}}{ \sqrt{\omega^{-2} + 2c}}\right)\right) = O(m^{-1})\end{equation*} The proof is finished in exactly the same way as Lemma~\ref{lemma:saddlepointlocation}. The only difference is the sign of the second term.
 \end{proof}

\subsection{Basic asymptotic expansions}

We have found that the saddle points are at a distance of order $m^{-1}$ from $\pm i$. So we will move the contours to pass through the appropriate saddle points in asymptotic coordinates and locally follow a path of steepest descent. Hence we will need to find asymptotic expansions of our integrands near $\pm i$. First we show that $V^{j,k}_{\e1, \e2} (\w1, \w2) h_{j,k}(\w1,\w2)$ is an even function in each variable, so we only need to look at the asymptotics for $\w1$ and $\w2$ both in a neighborhood of $i$.

First we note that by choice of branch cut we have \begin{equation}\label{eq:sqrtminus}
\sqrt{(-\omega)^2 + 2c} = -\sqrt{\omega^2 + 2c}
\end{equation} and so \begin{equation}\label{eq:Gtevenodd}
G(-\omega) = -G(\omega) \hspace{0.2cm}\text{ and }\hspace{0.2cm} t(-\omega) = t(\omega)\end{equation} Now we can prove the following lemmas.
\begin{lemma} Let $V^{j,k}_{\e1, \e2} (\w1, \w2)$ be as defined in Equation~\ref{eq:Voriginal}. Then 
\begin{align*}V_{\e1,\e2}^{j,k}(-\w1,\w2) &= (-1)^{1+\e1} V_{\e1,\e2}^{j,k}(\w1,\w2)\\V_{\e1,\e2}^{j,k}(\w1,-\w2) &= (-1)^{1+\e2} V_{\e1,\e2}^{j,k}(\w1,\w2)\end{align*}
\end{lemma}
\begin{proof}
The second equality is clear from Equation~\ref{eq:Voriginal}. For the first equality, note that $\mathrm{y}_{\gamma_1,\gamma_2}^{\e1,\e2}(u,v)$ is a even function in each variable. Then from Equations~\ref{eq:sqrtminus}--\ref{eq:Gtevenodd} we see that $\mathrm{x}_{\gamma_1,\gamma_2}^{\e1,\e2}(-\w1,\w2) = -\mathrm{x}_{\gamma_1,\gamma_2}^{\e1,\e2}(\w1,\w2)$, and using Equation~\ref{eq:Gtevenodd} again we obtain $Q_{\gamma_1,\gamma_2}^{\e1,\e2}(-\w1,\w2) = (-1)^{1+\e1} Q_{\gamma_1,\gamma_2}^{\e1,\e2}(\w1,\w2)$, from which the result follows.
\end{proof} 
\begin{lemma} Take $x = (x_1,x_2) \in \mathtt{W}_\e1,$ and $y = (y_1,y_2) \in \mathtt{B}_\e2$ with $\e1,\e2\in \{0,1\}$ and let $h_{j,k}(\w1,\w2)$ be as defined in Equation~\ref{eq:hij}.
Then we have \begin{align*}
h_{j,k}(-\w1,\w2) &= (-1)^{1+\e1} h_{j,k}(\w1,\w2) \\
h_{j,k}(\w1,-\w2) &= (-1)^{1+\e2} h_{j,k}(\w1,\w2) 
\end{align*} 
\end{lemma}
\begin{proof}
First we look at $h_{0,0}(\w1,\w2) = \widetilde{H}_{x_1+1,x_2}(\w1)/\widetilde{H}_{y_1,y_2+1}(\w2)$. We have 
\begin{align*}
\widetilde{H}_{x_1+1,x_2}(-\w1) &= \frac{\w1^{2m}(-iG(-\w1))^{2m - (x_1+1)/2}}{(iG(-\w1^{-1}))^{2m-x_2/2}}\\
&= (-1)^{-(x_1+1)/2+x_2/2}\widetilde{H}_{x_1+1,x_2}(\w1)\\
&= (-1)^{(x_1+1)/2+x_2/2}\widetilde{H}_{x_1+1,x_2}(\w1)\\
&=(-1)^{\e1+1}\widetilde{H}_{x_1+1,x_2}(\w1)
\end{align*} where we use the fact that $x_1$ is odd since $x$ is a white vertex, and $x_1 + x_2 \equiv 2\e1 + 1 \mod 4$.  Similarly \begin{align*}
\widetilde{H}_{y_1,y_2+1}(-\w2) &= (-1)^{-y_1/2+(y_2+1)/2}\widetilde{H}_{y_1,y_2+1}(\w2)\\
&= (-1)^{y_1/2+(y_2+1)/2}\widetilde{H}_{y_1,y_2+1}(\w2)\\
&=(-1)^{\e2+1}\widetilde{H}_{y_1,y_2+1}(\w2)
\end{align*}
This proves the result for $i=0,j=0$. For the other cases, note that $(-1)^{(2n - x_2)/2} = (-1)^{x_2/2}$ and $(-1)^{(2n - y_1)/2} = (-1)^{y_1/2}$. So we can prove these in exactly the same way.
\end{proof} Hence we have the following theorem.
\begin{theorem}\label{thm:integrandeven}
Take $x = (x_1,x_2) \in \mathtt{W}_\e1,$ and $y = (y_1,y_2) \in \mathtt{B}_\e2$ with $\e1,\e2\in \{0,1\}$ and let $h_{j,k}(\w1,\w2)$ be as defined in Equation~\ref{eq:hij}. Let $V^{j,k}_{\e1, \e2} (\w1, \w2)$ be as defined in Equation~\ref{eq:Voriginal}. Then \begin{align*}V_{\e1,\e2}^{j,k}(-\w1,\w2)h_{j,k}(-\w1,\w2) &=  V_{\e1,\e2}^{j,k}(\w1,\w2)h_{j,k}(\w1,\w2) \\
V_{\e1,\e2}^{j,k}(\w1,-\w2) h_{j,k}(\w1,-\w2) &= V_{\e1,\e2}^{j,k}(\w1,\w2)h_{j,k}(\w1,\w2) \end{align*}
\end{theorem}

Now we will compute some asymptotic expansions that we will need later. Let \begin{equation}\label{eq:omegaexpansion}
\omega = i + \bconst^2 m^{-1} w
\end{equation} for $|w| < m^{\delta}$ for some $0 < \delta < 1/2$. First, we have the following lemma.
\begin{lemma}\label{lemma:sqrtexpansions}
Let the $ \sqrt{\omega^2 + 2c}$ and $\sqrt{\omega^{-2} + 2c}$ be as defined in Equation~\ref{eq:sqrtbranchcut} and $\omega$ as in Equation~\ref{eq:omegaexpansion}. Then 
\begin{align}
 \sqrt{\omega^2 + 2c} &=i\bconst m^{-1/2}\sqrt{1/2-2iw} + O(m^{-1}w), \\
 \sqrt{\omega^{-2} + 2c} &= -i\bconst m^{-1/2}\sqrt{1/2 + 2iw} + O(m^{-1}w).
\end{align} 
\end{lemma}
\begin{proof}
We compute
 \begin{align*}
 \sqrt{\omega^2 + 2c} &= i\sqrt{-i(\omega + i\sqrt{2c})}\sqrt{-i(\omega - i \sqrt{2c})} \\
 &= i\sqrt{-i(i + \bconst^2 m^{-1} w+ i(1- \bconst^2 m^{-1}/4) + O(m^{-3/2}))}\\
 &\hspace{1cm}\times\sqrt{-i(i + \bconst^2 m^{-1} w - i(1- \bconst^2 m^{-1}/4) + O(m^{-3/2}))}\\
 &= i(\sqrt{2}+O(m^{-1}w))\sqrt{1 -i\bconst^2 m^{-1} w - (1- \bconst^2 m^{-1}/4) + O(m^{-3/2})} \\
 &= i(\sqrt{2}+O(m^{-1}w))\sqrt{-i\bconst^2 m^{-1} w + \bconst^2 m^{-1}/4 + O(m^{-3/2})} \\
 &= i\bconst m^{-1/2}(1
+O(m^{-1}w))\sqrt{1/2-2iw  + O(m^{-1/2}))} \\
 &= i\bconst m^{-1/2}\sqrt{1/2-2iw} + O(m^{-1}w^{-1/2}, m^{-3/2}w^{3/2})\\
  &= i\bconst m^{-1/2}\sqrt{1/2-2iw} + O(m^{-1}w).
\end{align*}
For the second equation, we have $\omega^{-1} = -i + \bconst^2m^{-1}w + O(m^{-2})$. By Equation~\ref{eq:sqrtminus}, we have \begin{align*}
\sqrt{\omega^{-2} + 2c} &= \sqrt{(-i + \bconst^2m^{-1}w + O(m^{-2}))^2 + 2c} \\
&= -\sqrt{(i - \bconst^2m^{-1}w + O(m^{-2}))^2 + 2c}
\end{align*} from which the result follows by a similar calculation.
\end{proof} Then we have 
\begin{align}\label{eq:Gasymptotics}
\begin{split}
-iG(\omega) &= \frac{1}{\sqrt{2c}}(1  - \bconst m^{-1/2}\sqrt{1/2-2iw}  + O(m^{-1}w)) \\
iG(\omega^{-1}) &= \frac{1}{\sqrt{2c}}(1  - \bconst m^{-1/2}\sqrt{1/2+2iw}  + O(m^{-1}w)).\end{split}
\end{align} 
Also, 
\begin{align}
\begin{split}
\log \omega &= \log (i + \bconst^2 m^{-1} w) \\
&= \log i + \log(1 - i\bconst^2 m^{-1} w) \\
&= \frac{\pi i}{2} -i\bconst^2 m^{-1} w + O(m^{-2}w^2),
\end{split}\label{eq:logomegaexpansion}
\end{align}
\begin{align}
\begin{split}
\log(-iG(\omega)) &= \log \frac{1}{\sqrt{2c}} + \log (1  - \bconst m^{-1/2}\sqrt{1/2 - 2i w} + O(m^{-1}w)) \\
&=   - \bconst m^{-1/2}\sqrt{1/2 - 2i w} +  O(m^{-1}w)
\end{split}\label{eq:logGomegaexpansion}
\end{align} and \begin{align}
\log (iG(\omega^{-1})) &=  -  \bconst m^{-1/2}\sqrt{ 1/2+2i w} +  O(m^{-1}w)\label{eq:logGomegainvexpansion}\end{align}

\subsection{Asymptotic expansion of exponential part of integrands}
First we will find asymptotic expansions for the terms $h_{j,k}(\w1,\w2)$ in the integrands. Let $\alpha_x,\,\alpha_y < 0$, $x = (x_1,x_2) \in \mathtt{W}_\e1$ and $y = (y_1, y_2) \in \mathtt{B}_\e2$ with $\e1,\e2\in \{0,1\}$ be as in Equation~\ref{eq:asymptotic_coords_xy}.
We define the asymptotic variables $w$ and $z$ near $\w1 = i$ and $\w2 = i$ respectively.
\begin{definition}\label{def:asymptoticsvars}
In a neighborhood of $i$, let \begin{equation}\label{eq:omega12expansion}
\w1 = i + \bconst^2 m^{-1} w \hspace{0.5cm} \text{and} \hspace{0.5cm} \w2 = i + \bconst^2 m^{-1} z 
\end{equation} for $|w|, |z| < m^{\delta}$ for some $0 < \delta < 1/2$.
\end{definition}
Then we can prove the following lemma.
\begin{lemma}
For $\alpha_x,\,\alpha_y < 0$, $x = (x_1,x_2) \in \mathtt{W}_\e1$ and $y = (y_1, y_2) \in \mathtt{B}_\e2$ as in Equation~\ref{eq:asymptotic_coords_xy}, and local coordinates $w,z$ as in Defintion~\ref{def:asymptoticsvars}, we have
\begin{align}
\widetilde{H}_{x_1 + 1, x_2}(\w1) &= (-1)^m e^{\bconst^2 (-2iw + \alpha_x(\sqrt{1/2 - 2i w} - \sqrt{1/2+2i w})) + O(m^{-1/2}w)} \label{eq:Hx1x2asymp}\\
\widetilde{H}_{y_1, y_2 + 1}(\w2) &= (-1)^m e^{\bconst^2 (-2iz + \alpha_y(\sqrt{1/2 - 2i z} - \sqrt{1/2+2i z})) + O(m^{-1/2}z)}\label{eq:Hy1y2asymp}\\
\widetilde{H}_{2n-y_1, y_2 + 1}(\w2) &= (-1)^m e^{\bconst^2(-2iz - \alpha_y( \sqrt{1/2 - 2i z} + \sqrt{1/2+2i z}))+ O(m^{-1/2}z)}\label{eq:H2ny1y2asymp}\\
\widetilde{H}_{x_1 + 1, 2n - x_2}(\w1) &= (-1)^m e^{\bconst^2(-2iw + \alpha_x (\sqrt{1/2-2iw} + \sqrt{1/2+2iw}))+ O(m^{-1/2}w)}\label{eq:Hx12nx2asymp}
\end{align}
\end{lemma}
\begin{proof}
We will look at these expressions one at a time. Firstly we have \[\widetilde{H}_{x_1 + 1, x_2}(\w1) = \frac{\w1^{2m} (-iG(\w1))^{-m^{1/2}\alpha_x \bconst +O(1)}}{(iG(\w1^{-1}))^{-m^{1/2}\alpha_x \bconst + O(1)}}\] so \[
\log \widetilde{H}_{x_1 + 1, x_2}(\w1) = 2m \w1 - (m^{1/2}\alpha_x \bconst  + O(1))(\log(-iG(\w1)) - \log(iG(\w1^{-1}))).
\] From Equations \ref{eq:logomegaexpansion}, \ref{eq:logGomegaexpansion} and \ref{eq:logGomegainvexpansion} we have \begin{multline*}
\log \widetilde{H}_{x_1 + 1, x_2}(\w1) = 2m\left(\frac{\pi i}{2} -i\bconst^2 m^{-1} w + O(m^{-2}w^2)\right)  \\
- (m^{1/2}\alpha_x \bconst  + O(1))(- \bconst m^{-1/2}\sqrt{1/2 - 2i w} +  \bconst m^{-1/2}\sqrt{ 1/2+2i w}+  O(m^{-1}w) ) \\
= m\pi i - 2 i \bconst^2 w + \alpha_x \bconst^2 (\sqrt{1/2 - 2i w} -\sqrt{ 1/2+2i w})  + O(m^{-1/2}w)
\end{multline*} where for the error term, we note that because $|w| \leq m^\delta$ with $0 < \delta < 1/2$, we can consolidate the $O(m^{-1}w^2),\, O(m^{-1/2}w^{1/2})$ and $O(m^{-1}w)$ error terms into one error term of order $m^{-1/2}w$. Equation~\ref{eq:Hx1x2asymp} follows. Equation~\ref{eq:Hy1y2asymp} also follows by replacing $\alpha_x$ with $\alpha_y$ and $w$ with $z$.

Next we look at \[\widetilde{H}_{2n - y_1, y_2+1}(\w2) = \frac{\w2^{2m} (-iG(\w2))^{m^{1/2}\alpha_y \bconst +O(1)}}{(iG(\w2^{-1}))^{-m^{1/2}\alpha_y \bconst + O(1)}}.\] Again, using Equations \ref{eq:logomegaexpansion}, \ref{eq:logGomegaexpansion} and \ref{eq:logGomegainvexpansion} we see that \begin{align*}
\log \widetilde{H}_{2n - y_1, y_2+1}(\w2) &= 2m \w2 + (m^{1/2}\alpha_y \bconst  + O(1))(\log(-iG(\w2)) + \log(iG(\w2^{-1}))) \\
&= m\pi i - 2 i \bconst^2 z - \alpha_y \bconst^2 (\sqrt{1/2 - 2i z} +\sqrt{ 1/2+2i z})  + O(m^{-1/2}z)
\end{align*} from which Equation~\ref{eq:H2ny1y2asymp} follows. Similarly, \[\widetilde{H}_{x_1 + 1, 2n - x_2}(\w1) = \frac{\w1^{2m} (-iG(\w1))^{-m^{1/2}\alpha_x \bconst +O(1)}}{(iG(\w1^{-1}))^{m^{1/2}\alpha_x \bconst + O(1)}}\] and so \begin{align*}
\log \widetilde{H}_{x_1 + 1, 2n - x_2}(\w1) &=  2m \w1 - (m^{1/2}\alpha_x \bconst  + O(1))(\log(-iG(\w1)) + \log(iG(\w1^{-1}))) \\
&= m\pi i - 2 i \bconst^2 w + \alpha_x \bconst^2 (\sqrt{1/2 - 2i w} +\sqrt{ 1/2+2i w})  + O(m^{-1/2}w)
\end{align*} from which Equation~\ref{eq:Hx12nx2asymp} follows.
\end{proof}

Recall the functions \[f^\pm(w) = \sqrt{1/2-2iw} \pm \sqrt{1/2+2iw}\] defined in Equation~\ref{eq:fpm} and $g_{j,k}(w,z)$ defined in Equation~\ref{eq:gij}.  The following theorem follows immediately.
\begin{theorem}\label{thm:Hw1w2expansion}
Let $\alpha_x,\,\alpha_y < 0$, $x = (x_1,x_2) \in \mathtt{W}_\e1$ and $y = (y_1, y_2) \in \mathtt{B}_\e2$ be as in Equation~\ref{eq:asymptotic_coords_xy}. Let $w$ and $z$ be local coordinates for $\w1, \w2$ near $i$ respectively as defined in Definition~\ref{def:asymptoticsvars}. Let the functions $h_{j,k}(\w1,\w2)$ be as defined in Equation~\ref{eq:hij}. Let the functions $g_{j,k}(w,z)$ be as defined in \ref{eq:gij}. Then 
\begin{equation}
h_{j,k}(\w1,\w2) = e^{g_{j,k}(w,z) + O(m^{-1/2}w, m^{-1/2}z)}
\end{equation}
These are the exponential parts of the integrands of $\mathcal{I}^{j,k}_{\e1,\e2} (a, x_1, x_2, y_1, y_2)$ near $\w1 = i$ and $\w2 = i$.
\end{theorem} 

\subsection{Asymptotic expansion of pre-exponential part of integrands}
The integrand of each integral $\mathcal{I}^{j,k}_{\e1,\e2} (a, x_1, x_2, y_1, y_2)$, defined in Equation~\ref{eq:mathcalI}, also contains a pre-exponential term \[\frac{V^{j,k}_{\e1, \e2} (\w1, \w2)}{\w1(\w2-\w1)}
\] where $V^{j,k}_{\e1, \e2} (\w1, \w2)$ is defined in Equation~\ref{eq:Voriginal}. We now state the asymptotic expansion of this term near $i$.
\begin{theorem}\label{thm:Vexpansion}
Let $\w1 = i + \bconst^2 m^{-1} w$ and $\w2 = i + \bconst^2 m^{-1} z $ for $|w|, |z| < m^{\delta}$ for some $0 < \delta < 1/2$ as in Definition~\ref{def:asymptoticsvars}. Then we have
\begin{equation}
V^{j,k}_{\e1,\e2}(\omega_1,\omega_2) = m^{1/2} \frac{(-1)^{\e1\e2}i^{\e1-\e2}}{16\bconst}(A^{j,k}_{\e1,\e2}(w,z)  + O(m^{-1/2}))
\end{equation} as $m\rightarrow \infty$, where $A^{j,k}_{\e1,\e2}(w,z)$ is defined in Equation~\ref{eq:Aall}.

\end{theorem}
The proof is given in Appendix~\ref{sec:Vexpansionproof}.

\subsection{Location of saddle points in asymptotic coordinates and their properties}
We will move our contours to contours of steepest descent for the functions in Theorem~\ref{thm:Hw1w2expansion} in a neighborhood of $i$, and symmetric contours in a neighborhood of $-i$. We will show in Theorem~\ref{thm:arcerrors} that outside these neighborhoods the contribution to the integral is exponentially small. Let
\begin{align}
p_{\alpha_x}(w) &= -2iw + \alpha_x (\sqrt{1/2-2iw} - \sqrt{1/2+2iw} )\label{eq:pdef}\\
q_{\alpha_x}(w) &= -2iw + \alpha_x (\sqrt{1/2-2iw} + \sqrt{1/2+2iw} )\label{eq:qdef}. 
\end{align} We can write
\begin{equation*}
p_{\alpha_x}(w) = -2iw + \alpha_x f^-(w) \hspace{0.5cm}\text{and}\hspace{0.5cm} q_{\alpha_x}(w) = -2iw + \alpha_x f^+(w), 
\end{equation*} where $f^\pm(w)$ are defined in Equation~\ref{eq:fpm}, and note that since $f^-(w) = -f^-(-w)$ and $f^+(w) = f^+(-w)$, we have $2iz - \alpha_yf^-(z) = -p_{\alpha_y}(z) = p_{\alpha_y}(-z)$ and $2iz + \alpha_y f^+(z) = q_{\alpha_y}(-z)$. So from Theorem~\ref{thm:Hw1w2expansion} we have 
\begin{align}\begin{split}
h_{0,0}(\w1,\w2) &= \exp(\bconst^2(p_{\alpha_x}(w) + p_{\alpha_y}(-z)) + O(m^{-1/2}w, m^{-1/2}z))\\
h_{1,0}(\w1,\w2) &= \exp(\bconst^2(p_{\alpha_x}(w) + q_{\alpha_y}(-z)) + O(m^{-1/2}w, m^{-1/2}z)) \\
h_{0,1}(\w1,\w2) &= \exp(\bconst^2(q_{\alpha_x}(w) + p_{\alpha_y}(-z)) + O(m^{-1/2}w, m^{-1/2}z)) \\
h_{1,1}(\w1,\w2) &= \exp(\bconst^2(q_{\alpha_x}(w) + q_{\alpha_y}(-z)) + O(m^{-1/2}w, m^{-1/2}z))\end{split}\label{eq:expparts}
\end{align} So it is sufficient for us to find contours of steepest descent for $p_{\alpha}(w)$ and $q_{\alpha}(w)$, where $\alpha < 0$. We first find the saddle points. These occur when $p_{\alpha}'(w) = 0$ and $q_{\alpha}'(w) = 0$ respectively. We compute \begin{equation}\label{eq:pprime}
p_{\alpha}'(w) = i\left(-2 - \alpha \left(\frac{1}{\sqrt{1/2 - 2iw}} + \frac{1}{\sqrt{1/2 + 2iw}}\right)\right)
\end{equation} and \begin{equation}\label{eq:qprime}
q_{\alpha_x}'(w) = i\left(-2 - \alpha \left(\frac{1}{\sqrt{1/2 - 2iw}} - \frac{1}{\sqrt{1/2 + 2iw}}\right)\right).
\end{equation} To analyze the locations of the zeros of $p_{\alpha}'(w)$ and $q_{\alpha}'(w)$ we will first define a couple of functions and prove some lemmas about them. For $w \in \mathbb{C}\setminus i((-\infty, -1/4]\cup [1/4, \infty))$, let \begin{align}\begin{split}
\psi^+(w) &= \frac{1}{\sqrt{1/2 - 2 i w}} + \frac{1}{\sqrt{1/2+ 2i w}} \\
\psi^-(w) &= \frac{1}{\sqrt{1/2 - 2 i w}} - \frac{1}{\sqrt{1/2+ 2i w}}.\end{split}\label{eq:psi}\end{align} So we are looking for solutions to \begin{equation}
\psi^+(w) = -2/\alpha \hspace{0.5cm} \text{and} \hspace{0.5cm} \psi^-(w) = -2/\alpha\label{eq:psipm}
\end{equation}

\begin{lemma}\label{lemma:etaquadratic}
Take $\alpha < 0$. Solutions to the equations $\psi^+(w) = \pm 2/\alpha$ and $\psi^-(w) = \pm 2/\alpha$ all satisfy the quartic equation \begin{equation}256\,w^4 - 16(\alpha^4 + 2\alpha^2 - 2)\,w^2 + (-2\alpha^2 + 1) = 0.\label{eq:etaquadratic}\end{equation}
\end{lemma}
\begin{proof}
This is just a computation. Start from \[
\pm \frac{1}{\sqrt{1/2 - 2 i \eta}} + \frac{1}{\sqrt{1/2+ 2i \eta}} = \pm \frac{2}{\alpha}.\] Squaring both sides, we obtain\[
\frac{1}{1/2-2i\eta} + \frac{1}{1/2 + 2i\eta} \pm \frac{2}{\sqrt{1/2 - 2i\eta}\sqrt{1/2 + 2i\eta}} = \frac{4}{\alpha^2}.\] Then we multiply by $(1/2 - 2i\eta)(1/2 + 2i\eta)$ to obtain \[
1 + 2\sqrt{1/2 - 2i\eta}\sqrt{1/2 + 2i\eta} = \frac{4}{\alpha^2}\left(\frac{1}{4} + 4\eta^2\right)\] and rearranging and squaring again, we obtain \[
\left(1 - \frac{4}{\alpha^2}\left(\frac{1}{4} + 4\eta^2\right)\right)^2 = 1 + 16\eta^2.\] Multiplying everything out and rearranging, we obtain Equation~\ref{eq:etaquadratic}.
\end{proof}

\begin{lemma}\label{lemma:psiimageim} Consider the restriction of $\psi^\pm(w)$ to $w \in i(-1/4, 1/4)$. Then $\psi^+(w)$ takes all values in $[2\sqrt{2}, \infty)$, and $\psi^-(w)$ takes all values in $(-\infty, \infty)$.
\end{lemma}
\begin{proof}
Let $w = it$ for $t \in (-1/4, 1/4)$. Then \[\psi^\pm(it) = \frac{1}{\sqrt{1/2 + 2 t}} \pm \frac{1}{\sqrt{1/2- 2 t}}.\] These are continuous for $t \in (-1/4,1/4)$. We have \begin{align*}
\lim_{t\rightarrow -\frac{1}{4}}\psi^-(it) &= \infty \\
\lim_{t\rightarrow \frac{1}{4}}\psi^-(it) &= -\infty
\end{align*} hence $\psi^-(it)$ takes all values in $(-\infty, \infty)$.

For $\psi^+(it)$, note that for $t \in (-1/4, 1/4)$, by the AM-GM inequality we have $\psi^+(it) \geq 2\sqrt{2}$, with equality if and only if $t = 0$. Since $lim_{t\rightarrow 1/4}\psi^+(it) = $, we see that $\psi^+(it)$ takes all values in $(0, 2\sqrt{2}]$.
\end{proof}

\begin{lemma}\label{lemma:psiimagereal}
Consider the restriction of $\psi^+(w)$ to $w \in \mathbb{R}$. Then $\psi^+(w)$ takes all values in $(0, 2\sqrt{2}]$.
\end{lemma}
\begin{proof}
In the case that $w$ is real, we have \[\psi^+(w) = 2\, \mathrm{Re}\left(\frac{1}{\sqrt{1/2 - 2iw}}\right)\] We have $\psi^+(0) = 2\sqrt{2}$ and $\lim_{w\rightarrow \infty} \psi^+(w) = \infty$. Since $\psi^+(w)$ is continuous, it takes all values in the range $(0, 2\sqrt{2})$.
\end{proof}

\begin{lemma} \label{lemma:psipmsols}
Take $\alpha < 0$. Then \begin{enumerate}
\item There exists $\eta \in [0, \infty) \cup i(0, 1/4)$ such that $\psi^+(w) = -2/\alpha$ has exactly two solutions $w = \pm \eta$, unless $\alpha = -1/\sqrt{2}$ in which case the only solution is $w = 0$.
\item There exists $\eta' \in i(0, 1/4)$ such that Equation $\psi^-(w) = 2/\alpha$ has exactly one solution $w = \eta'$ and Equation $\psi^-(w) = -2/\alpha$ has exactly one solution $w = -\eta'$.
\end{enumerate}  
\end{lemma}
\begin{proof}

First we look at solutions to $\psi^-(w) = \pm 2/\alpha$. By Lemma~\ref{lemma:psiimageim}, we see that $\psi^-(w) = -2/\alpha$ has at least one solution in $i(-1/4,1/4)$. Moreover, $\psi^-(it) < 0$ for $t > 0$, $\psi^-(it) > 0$ for $t < 0$ and $\psi^-(it) = 0$ for $t = 0$, so since $-2/\alpha > 0$, any solution to $\psi^-(it) = 2/\alpha$ has $t \in (0, 1/4)$. So we can take $w=\eta'$ be a solution to $\psi^-(w) = 2/\alpha$ with $\eta' \in i(0, 1/4)$. Then it is clear that $w = -\eta'$ is a solution to $\psi^-(w) = -2/\alpha$. Note that since $-2/\alpha \neq 0$, $\eta'$ is non-zero.

Now we look at solutions to $\psi^+(w) = -2/\alpha$. By Lemmas \ref{lemma:psiimageim} and \ref{lemma:psiimagereal}, we see that $\psi^+(w) = -2/\alpha>0$ has at least one solution in $\mathbb{R}\cup i(0, 1/4)$. Moreover, it is clear that the negative of a solution is a solution, so we can take $w=\eta$ to be a solution to $\psi^+(w) = -2/\alpha$ in $[0, \infty) \cup i(0, 1/4)$. Then $w = -\eta$ is also a solution of $\psi^+(w) = -2/\alpha$.

So if $\eta \neq 0$, then we have already found four roots of Equation~\ref{eq:etaquadratic} (namely, $\pm \eta$, $\pm \eta'$), so there are no more. If $\eta_0 = 0$, so $\alpha = -1/\sqrt{2}$, then $\eta = 0$ is a repeated root of Equation~\ref{eq:etaquadratic} and we have still found all of the roots. By Lemma~\ref{lemma:etaquadratic}, we have found all solutions to $\psi^+(w) = \pm 2/\alpha$ and $\psi^-(w) = \pm 2/\alpha$. Note that the equation $\psi^+(w) = - 2/\alpha$ has two solutions while $\psi^+(w) = 2/\alpha$ has no solutions.
\end{proof}

Recall that we are looking for solutions to $p_\alpha'(w) = 0$ and $q_\alpha'(w) = 0$ where $p_\alpha'(w)$ and $q_\alpha'(w)$ are given by Equations \ref{eq:pprime} and \ref{eq:qprime} respectively. We can write \begin{align*}
p_\alpha'(w) &= i(-2 -\alpha \psi^+(w)) \\
q_\alpha'(w) &= i(-2 -\alpha \psi^-(w)).
\end{align*} We have the following lemma.

\begin{lemma}\label{lemma:eta0sols}
For $-1/\sqrt{2} < \alpha < 0$, $p_\alpha'(w) = 0$ has two solutions, which are in the interval $i(-1/4,1/4)$. For $\alpha = -1/\sqrt{2}$, $p_\alpha'(w) = 0$ has only one solution, $w = 0$. For $\alpha < -1/\sqrt{2}$, $p_\alpha'(w) = 0$ has two solutions, which are real.
In all cases if $w=\eta$ is a solution of $p_\alpha'(w) = 0$ then so is $w=-\eta$.
\end{lemma}
\begin{proof}
We are looking for solutions to $\psi^+(w) = -2/\alpha$. For $-1/\sqrt{2} < \alpha < 0$, we have $-2/\alpha \in (2\sqrt{2}, \infty)$. For $\alpha = -1/\sqrt{2}$, we have $-2/\alpha = 2\sqrt{2}$. For $\alpha < -1/\sqrt{2}$, we have $-2/\alpha \in (0, 2\sqrt{2})$. The lemma follows from Lemma~\ref{lemma:psipmsols}, Lemma~\ref{lemma:psiimageim} and  Lemma~\ref{lemma:psiimagereal}.
 \end{proof}
 
We have a somewhat simpler lemma for solutions to $\psi^-(w) = -2/\alpha$.
 
 \begin{lemma}\label{lemma:eta01sols}
For any $\alpha < 0$, there exists $\eta_0'$ in $i(0,1/4)$ such that Equation $\psi^+(w) = -2/\alpha$ has exactly one solution $w = -\eta'$.
\end{lemma} 
\begin{proof}
We are looking for solutions to $\psi^-(w) = -2/\alpha$. The result follows from Lemma~\ref{lemma:psipmsols}.  
\end{proof}

Now that we have found the solutions to $\psi^+(w) = -2/\alpha$ and $\psi^-(w) = -2/\alpha$ we can make the following definition.
\begin{definition}\label{def:etaeta'}
 Let $\eta$ be defined to be the unique complex number with non-negative real part and non-negative imaginary part that satisfies \begin{equation*}
\frac{1}{\sqrt{1/2 - 2 i \eta}} + \frac{1}{\sqrt{1/2+ 2i \eta}} = -2/\alpha.\end{equation*} and let $\eta'$ be the unique complex number that satisfies \begin{equation*}
\frac{1}{\sqrt{1/2 - 2 i \eta'}}- \frac{1}{\sqrt{1/2+ 2i \eta'}} = 2/\alpha.\end{equation*}
\end{definition} So $\eta$ satisfies $\psi^+(\eta) = -2/\alpha$ and so also satisfies $p_\alpha'(\eta) = 0$. Similarly $\eta'$ satisfies $\psi^-(\eta') = 2/\alpha$, hence $\psi^-(-\eta') = -2/\alpha$, so $q_\alpha'(-\eta') = 0$. This gives us the following theorem.

\begin{theorem}\label{thm:pqsaddlepoints}
Take $\alpha < 0$ and let $\eta$ and $\eta'$ be as defined in Definition~\ref{def:etaeta'}. Then \begin{enumerate}
\item The saddle points of $p_{\alpha}(w)$ are as follows. 
\begin{description}
\item[$-1/\sqrt{2} < \alpha < 0$.]\hfill\\
There are 2 distinct saddle points, at $w = \pm \eta$, where $\eta \in i (0,1/4)$. 
\item[$\alpha = -1/\sqrt{2}$.]\hfill\\
There is only one saddle point, which occurs at $w = 0$.
\item[$\alpha < -1/\sqrt{2}$.]\hfill\\
There are 2 distinct saddle points, at $w = \pm  \eta$, where $\eta \in (0, \infty)$.
\end{description}
\item The unique saddle point of $q_{\alpha}(w)$ is at $w = -  \eta'$ with $\eta'\in i(0,1/4)$.
\end{enumerate}
\end{theorem}
\begin{proof}
Saddle points of $p_{\alpha}(w)$ occur where $p_{\alpha}'(w) = 0$, and saddle points of $q_{\alpha}(w)$ occur where $q_{\alpha}'(w) = 0$. The result follows from Lemma~\ref{lemma:eta0sols} and Lemma~\ref{lemma:eta01sols}.
\end{proof}

Now we follow with a lemma about the values of $p_\alpha''(w)$ and $q_\alpha''(w)$ at the saddle points. We will use this to find the direction of the steepest ascent and descent contours at the saddle points. 

\begin{lemma}\label{lemma:saddlepoints2ndderiv}
Take $\alpha < 0$ and let $\eta$ and $\eta'$ be as defined in Definition~\ref{def:etaeta'}. Then \begin{enumerate}
\item For $-1/\sqrt{2} < \alpha < 0$, we have $p_\alpha''(\eta) > 0$ and $p_\alpha''(-\eta) < 0$. \\
For $\alpha = -1/\sqrt{2}$, we have $p_\alpha''(\eta) = 0$, so $\eta = 0$ is a double saddle point. \\
For $\alpha < -1/\sqrt{2}$ we have $p_\alpha''(\eta) \in i(0,\infty)$ and $p_\alpha''(-\eta) \in i(-\infty, 0)$.
\item For all $\alpha < 0$, we have $q_\alpha''(-\eta') < 0$.
\end{enumerate}
\end{lemma}
\begin{proof}
\begin{enumerate}
\item We compute \[
p_\alpha''(w) = \alpha\left(\frac{1}{(1/2 - 2iw)^{3/2}} - \frac{1}{(1/2 + 2iw)^{3/2}}\right).\] From Theorem~\ref{thm:pqsaddlepoints} we know that if $-1/\sqrt{2} < \alpha < 0$ then $\eta \in i(0, 1/4)$. So $1/2 - 2iw > 1/2 + 2iw$. Then since $\alpha < 0$, we have $p_\alpha''(\eta) > 0$ and $p_\alpha''(-\eta) < 0$. Also, if $\alpha = -1/\sqrt{2}$ then $\eta = 0$ so $p_\alpha''(\eta) = 0$. Finally, if $\alpha < -1/\sqrt{2}$ then $\eta \in (0, \infty)$ so \[
p_\alpha''(w) = 2\alpha i \,\mathrm{Im}\left(\frac{1}{(1/2 - 2iw)^{3/2}}\right).\] Then $\mathrm{Im}(1/2-2iw) < 0$ which means that $\mathrm{Im}(1/(1/2-2iw)^{3/2}) > 0$. So, $p_\alpha''(\eta) \in i(-\infty, 0)$ and $p_\alpha''(-\eta) \in i(0, \infty)$.

\item We compute \[
q_\alpha''(w) = \alpha\left(\frac{1}{(1/2 - 2iw)^{3/2}} + \frac{1}{(1/2 + 2iw)^{3/2}}\right).\] From Theorem~\ref{thm:pqsaddlepoints} we know that $\eta' \in i(0, 1/4)$.  So both terms are positive, and recalling that $\alpha < 0$, we see that $q_\alpha''(-\eta') < 0$.
\end{enumerate}
\end{proof}

Finally, since when $\alpha = -1/\sqrt{2}$ we have a double saddle point, we will need the third derivative of $p_{-1/\sqrt{2}}(w)$ at $w = 0$. We find that \begin{equation}\label{eq:doublethirdderiv}
p_{-1/\sqrt{2}}'''(0) = - 24 i.
\end{equation}

\subsection{Contours of steepest ascent and descent}\label{sec:contours}

Now we will characterize the contours of steepest ascent and descent of $p_{\alpha}(w)$ and $q_{\alpha}(w)$.  These are contours of constant $\mathrm{Im}(p_{\alpha}(w))$ and $\mathrm{Im}(q_{\alpha}(w))$. Contours of steepest ascent are those where $\mathrm{Re}(p_{\alpha}(w))$ and $\mathrm{Re}(q_{\alpha}(w))$ increases from the saddle point along the contour respectively, while contours of steepest descent are those where $\mathrm{Re}(p_{\alpha}(w))$ and $\mathrm{Re}(q_{\alpha}(w))$ decrease. Theorem~\ref{thm:pqsaddlepoints}, Lemma~\ref{lemma:saddlepoints2ndderiv} and Equation~\ref{eq:doublethirdderiv} contain all the information we need to find the directions of steepest ascent and descent contours at the saddle points.

\begin{theorem}\label{thm:contoursnearsp}
Take $\alpha < 0$ and let $\eta$ and $\eta'$ be as defined in Definition~\ref{def:etaeta'}. Then \begin{enumerate}
\item The contours of steepest ascent and descent near the saddle points $w = \pm \eta$ of $p_{\alpha}(w)$ are as follows. 
\begin{description}
\item[$-1/\sqrt{2} < \alpha < 0$.]\hfill\\
At the saddle point $w = \eta \in i(0, 1/4)$, the contour of steepest ascent passes through the saddle point parallel to the real axis, while the contour of steepest descent is parallel to the imaginary axis. At the saddle point $w = -\eta \in i(-1/4, 0)$, the contour of steepest ascent passes through the saddle point  parallel to the imaginary axis, while the contour of steepest descent is parallel to the real axis. 
\item[$\alpha = -1/\sqrt{2}$.]\hfill\\
There is a double saddle point at $w = 0$. There are contours of steepest ascent leaving the saddle point at angles $\pi/6$, $5\pi/6$ and $-\pi/2$. There are contours of steepest descent leaving the saddle point at angles $-\pi/6$, $-5\pi/6$ and $\pi/2$.
\item[$\alpha < -1/\sqrt{2}$.]\hfill\\
At the saddle point $w = \eta \in \mathbb{R}$, the contour of steepest ascent passes through the saddle point at an angle of $\pi/4$, while the contour of steepest descent passes through the saddle point at an angle of $-\pi/4$. At the saddle point $w = -\eta \in \mathbb{R}$, the contour of steepest ascent passes through the saddle point at an angle of $-\pi/4$, while the contour of steepest descent passes through the saddle point at an angle of $\pi/4$.

\end{description}
\item The contour of steepest ascent of $q_\alpha(w)$ passes through the saddle point $w = -\eta' \in i(-1/4, 0)$ parallel to the imaginary axis, while the contour of steepest descent is parallel to the real axis.
\end{enumerate}
\end{theorem}
\begin{proof}
\begin{enumerate}
\item Locally we have \[
p_\alpha(w) = p_\alpha(\eta) + \frac{(w - \eta)^2}{2}p_\alpha''(\eta) + O((w - \eta)^3).\] So contours of steepest ascent are in the directions where $(w - \eta)^2p_\alpha''(\eta) > 0$ and contours of steepest descent are in the directions where  $(w - \eta)^2p_\alpha''(\eta) < 0$ as $w \rightarrow 0$. Then the results for $\alpha \neq -1/\sqrt{2}$ follow immediately from Lemma~\ref{lemma:saddlepoints2ndderiv}. For $\alpha = -1/\sqrt{2}$, which has a double saddle point at 0 we have \[
p_{-1/\sqrt{2}}(w) = p_{-1/\sqrt{2}}(0) + \frac{w^3}{6}p_{-1/\sqrt{2}}'''(0) + O(w^4).\] Then the result follows from Equation~\ref{eq:doublethirdderiv}.

\item We follow the same steps as above.
\end{enumerate}
\end{proof}

Now we proceed to describe the contours of steepest ascent and descent of $p_\alpha(w)$ and $q_\alpha(w)$. Note that writing $w = w_1 + i w_2$, we have \[p_\alpha(w_1 - iw_2) = \overline{p_\alpha(w_1+ iw_2)}\] and \[q_\alpha(w_1 - iw_2) = \overline{q_\alpha(w_1+ iw_2)}\] where the bar denotes complex conjugate. Also note that $p_\alpha(-w) = -p_\alpha(w)$. So all contours of steepest ascent and descent are symmetric under reflection in the imaginary axis, and for  $p_\alpha(w)$, contours of steepest ascent are just reflections in the real axis of contours of steepest descent. 

\begin{figure}[htbp]
\centering
\begin{subfigure}[t]{0.41\textwidth}
\centering
\includegraphics[width = \textwidth]{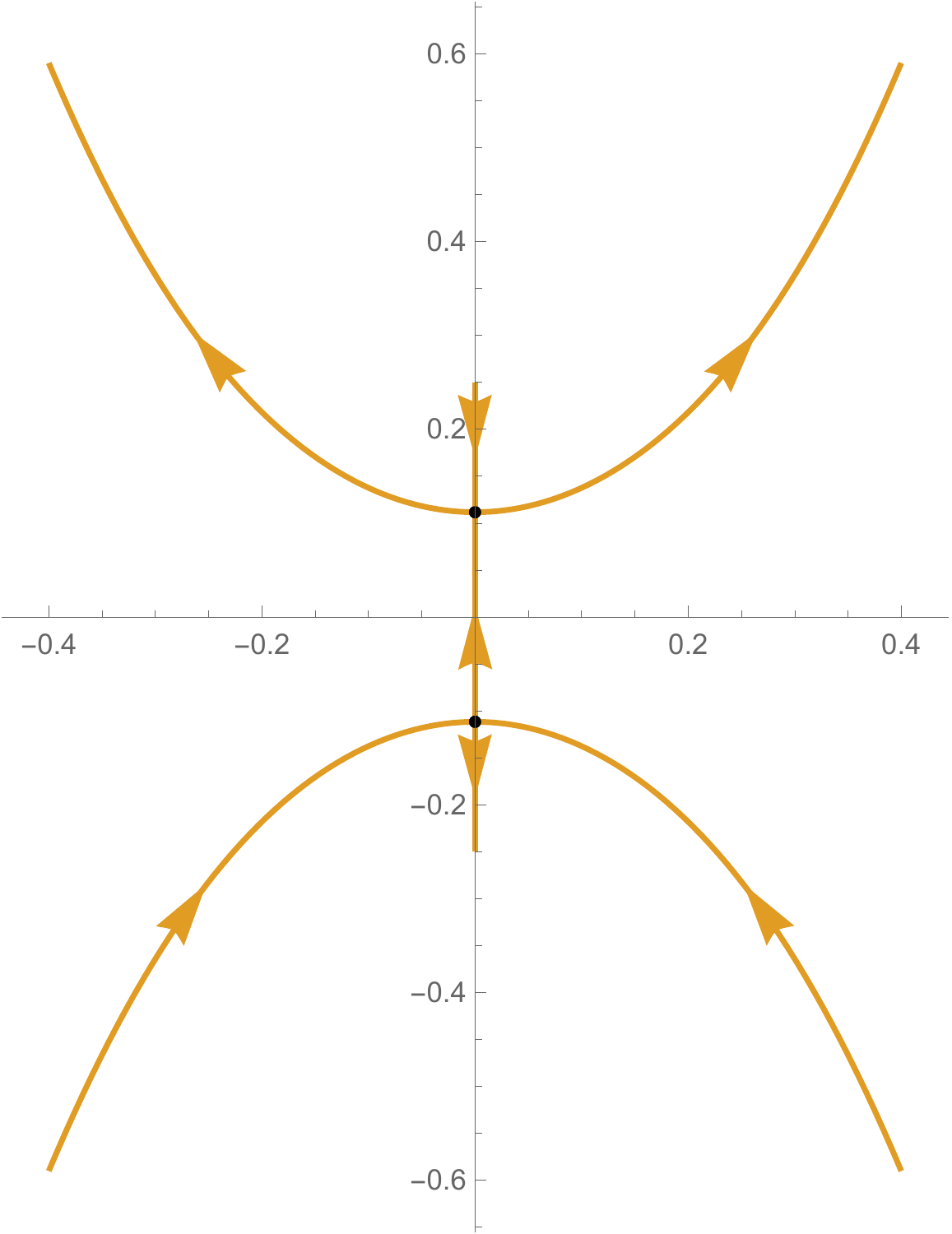}
\caption{$\alpha = -0.65$}
\label{fig:contoursarrowsgas}
\end{subfigure}\hspace{0.2cm}%
\begin{subfigure}[t]{0.41\textwidth}
\centering
\includegraphics[width = \textwidth]{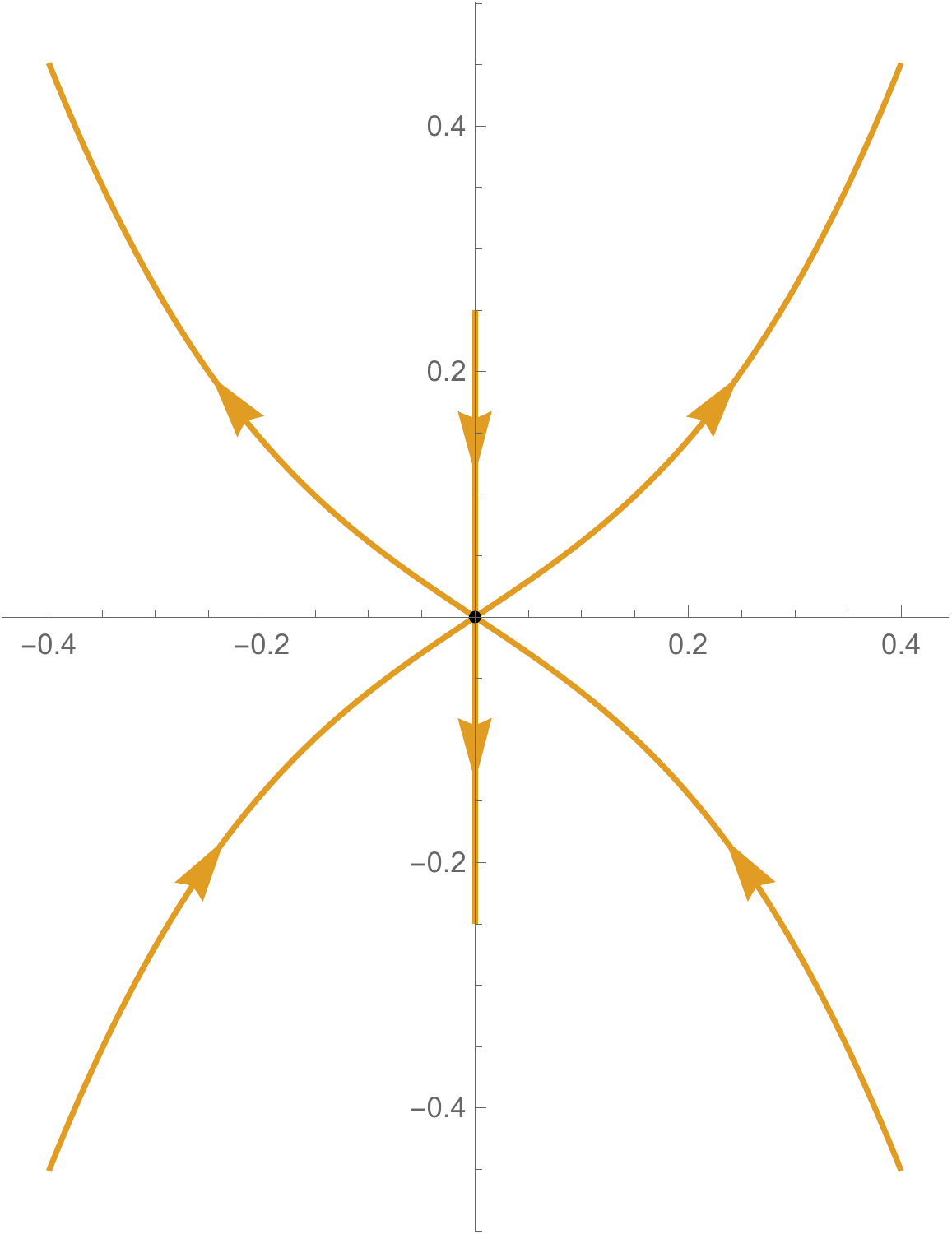}
\caption{$\alpha = -1/\sqrt{2}$}
\label{fig:contoursarrowsboundary}
\end{subfigure}\vspace{0.4cm}

\hspace{0.2cm}
\begin{subfigure}[t]{0.41\textwidth}
\centering
\includegraphics[width = \textwidth]{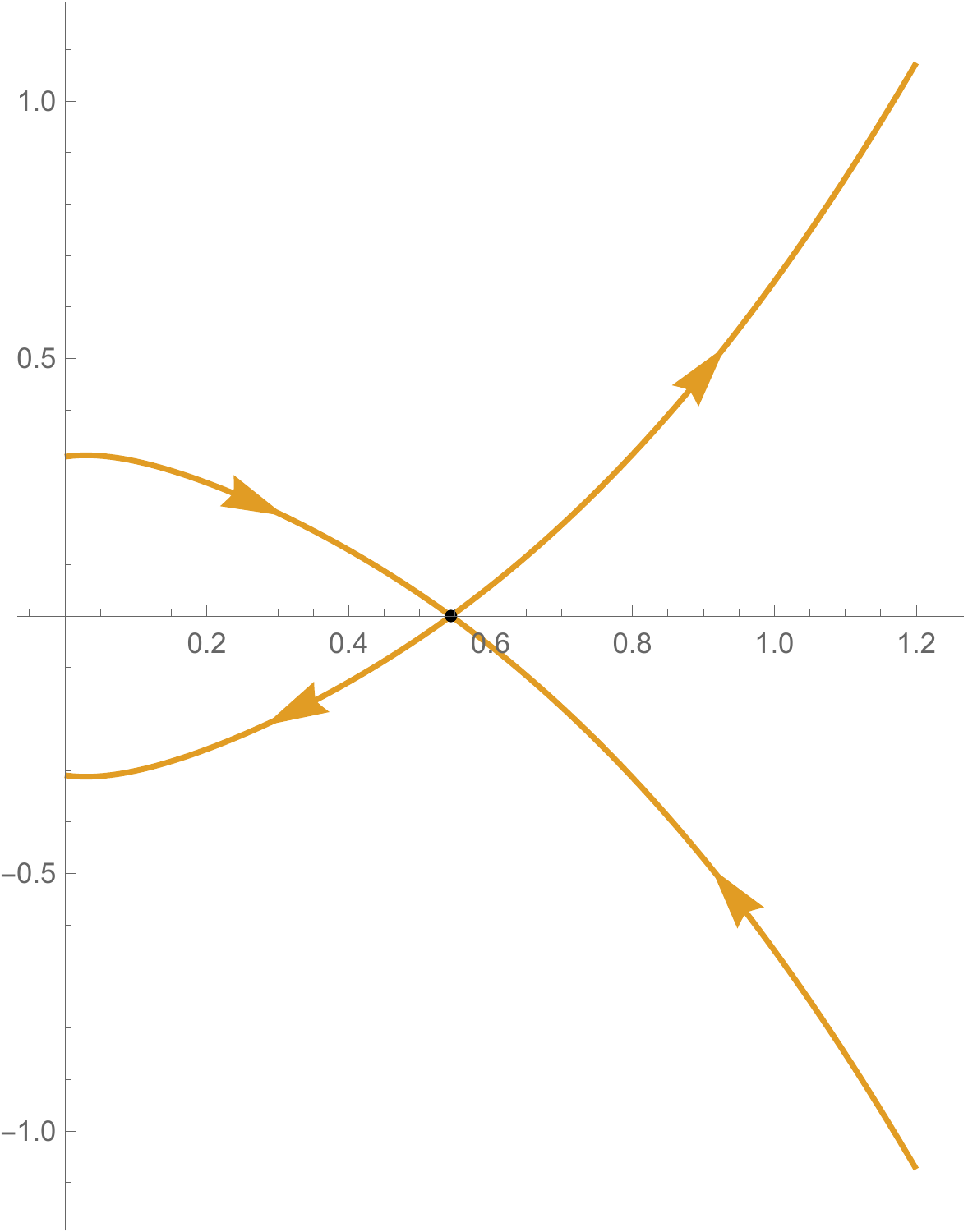}
\caption{$\alpha = -1.3$}
\label{fig:contoursarrowsliquid}
\end{subfigure}%
\begin{subfigure}[t]{0.41\textwidth}
\centering
\includegraphics[width = \textwidth]{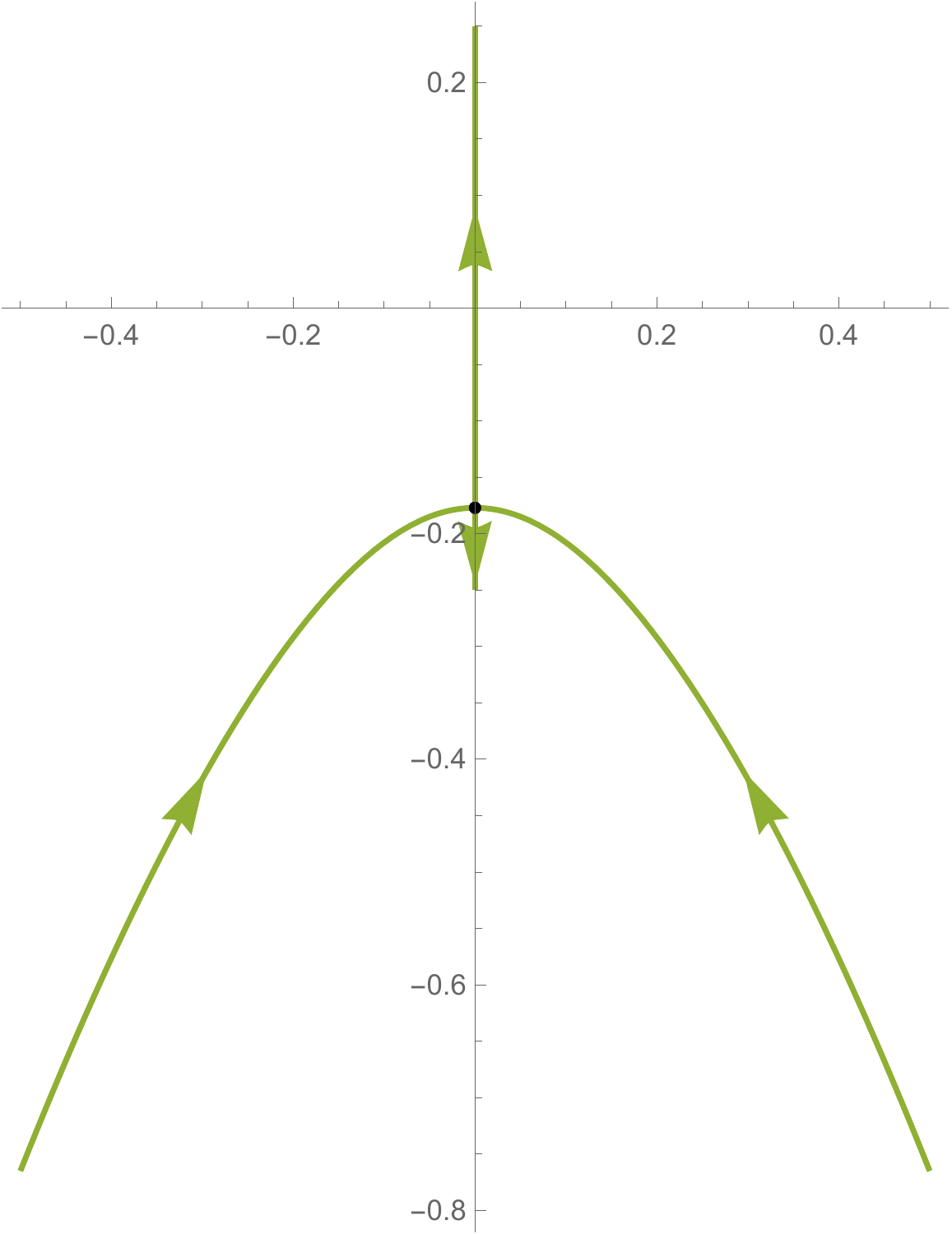}
\caption{$\alpha = -1.3$}
\label{fig:contoursarrowsq}
\end{subfigure}

\caption{Steepest ascent and descent contours for $p_\alpha(w)$ in orange and $q_\alpha(w)$ in green for different values of $\alpha$. The value of $\mathrm{Re}(p_\alpha(w))$ or $\mathrm{Re}(q_\alpha(w))$ increases in the direction of the arrows, i.e. ascent contours have arrows pointing away from the saddle point, and descent contours have arrows pointing towards the saddle point.}\label{fig:contoursarrows}
\end{figure}

First we look at $p_\alpha(w)$. Note that the contours of steepest ascent and descent can only cross at saddle points. Since the contours of steepest ascent are reflections in the real axis of contours of steepest descent, this means that neither can touch the real axis except where we have a saddle point. So the contours of steepest ascent and descent either go to the imaginary axis or to infinity. For $t\in \mathbb{R}$ with $it$ on the positive side of the branch cut, we have \[
\mathrm{Im}(p_\alpha(it)) = \begin{cases} 
0 & \text{if } |t| \leq 1/4 \\
-\alpha \sqrt{2|t| - 1/2} & \text{if } |t| > 1/4.\end{cases}\] We consider different values of $\alpha$. \begin{description}
\item[$-1/\sqrt{2} < \alpha < 0$]\hfill\\
 For $-1/\sqrt{2} < \alpha < 0$, both saddle points are on the imaginary axis between $-1/4$ and $1/4$, so we have $\mathrm{Im}(p_\alpha(w)) = 0$ on all contours of steepest ascent and descent. So from Theorem~\ref{thm:contoursnearsp} we must have a contour of steepest descent going from the saddle point at $w = \eta\in i(0,1/4)$ to the other saddle point at $w = -\eta\in i(-1/4,0)$ along the imaginary axis, a contour of steepest descent going from the saddle point at $w = \eta\in i(0,1/4)$ to the branch cut at $i/4$ along the imaginary axis, a contour of steepest ascent leaving the saddle point at $w = \eta$ perpendicularly into $\mathbb{H}^+$ and going to infinity in the positive half plane, a contour of steepest descent leaving the saddle point at $w = -\eta$ perpendicularly into $\mathbb{H}^+$ and going to infinity in the negative half plane, and a contour of steepest ascent going from the saddle point at $w = -\eta$ to the branch cut at $-i/4$ along the imaginary axis,. There are symmetric contours in the second and third quadrants. See Figure~\ref{fig:contoursarrowsgas}.

\item[$\alpha = -1/\sqrt{2}$]\hfill\\
For $\alpha = -1/\sqrt{2}$, we have $\mathrm{Im}(p_\alpha(\eta)) = 0$. The contours do not cross the real axis except at the origin. So we see from Theorem~\ref{thm:contoursnearsp} that we must have a descent contour from 0 to $i/4$ along the imaginary axis, an ascent contour from 0 to $-i/4$ along the imaginary axis, ascent contours going to infinity in the first and second quadrant, and descent contours going to infinity in the third and fourth quadrant. See Figure~\ref{fig:contoursarrowsboundary}.

\item[$\alpha < -1/\sqrt{2}$]\hfill\\
For $\alpha < -1/\sqrt{2}$, the steepest ascent and descent contours that pass through $w = \eta \in (0, \infty)$ have $\mathrm{Im}(p_\alpha(w)) = p_\alpha(\eta)/i = -2\eta - i\alpha(\sqrt{1/2 - 2i\eta} - \sqrt{1/2 + 2i\eta})$. It can be shown using Definition~\ref{def:etaeta'} that this expression is strictly positive for $\alpha < -1/\sqrt{2}$. So there is a point $w = it$ with $t > 1/4$ on the positive side branch cut such that $\mathrm{Im}(p_\alpha(it)) = \mathrm{Im}(p_\alpha(\eta))$ and  $\mathrm{Im}(p_\alpha(-it)) = \mathrm{Im}(p_\alpha(\eta))$. So for the ascent and descent contours from $w = \eta$ we have one going to the branch cut in the first quadrant, one going to the branch cut in the fourth quadrant, one going to infinity in the first quadrant and one going to infinity in the fourth quadrant. From Theorem~\ref{thm:contoursnearsp}, since the contours cannot cross, we see that the descent contours go to the branch cut in the first quadrant and infinity in the fourth quadrant, while the ascent contours go to the branch cut in the fourth quadrant and infinity in the first quadrant. We have symmetric contours in the second and third quadrants, but note that unlike in the $-1/\sqrt{2} < \alpha < 0$ case they do not join up on the imaginary axis, since they go to different sides of the branch cut. Indeed, the value of $\mathrm{Im}(p_\alpha(it))$ at the point that the contours hit the branch cut is non zero, and differs by a sign on either side of the branch cut. In practice, we will have to join up these contours by going around the branch cut. See Figure~\ref{fig:contoursarrowsliquid} for an example of the steepest ascent and descent contours in the first and fourth quadrants.

\end{description}

Now we look at $q_\alpha(w)$. We only have one saddle point, at $w = -\eta' \in i(-1/4,0)$. Here we have $\mathrm{Im}(q_\alpha(-\eta')) = 0$, so on the contours of steepest ascent and descent we have $\mathrm{Im}(q_\alpha(w)) = 0$. For $t\in \mathbb{R}$ with $it$ on the positive side of the branch cut, we have \[
\mathrm{Im}(q_\alpha(it)) = \begin{cases} 
-\alpha \sqrt{-2t - 1/2} & \text{if } t < -1/4 \\
0 & \text{if } -1/4 \leq t \leq 1/4 \\
\alpha \sqrt{2t - 1/2} & \text{if } t > 1/4.\end{cases}\] So the contours cannot go to the branch cut except at $\pm i/4$. For $w \in \mathbb{R}$, we have  $\mathrm{Im}(q_\alpha(w)) = -2w$, so the contours do not touch the real axis except at 0. Then from Theorem~\ref{thm:contoursnearsp} we see that there is a contour of steepest ascent from $-\eta'$ to $-i/4$ along the imaginary axis, a contour of steepest ascent from $-\eta'$ to $i/4$ along the imaginary axis, and contours of steepest descent that go to infinity in the third and fourth quadrants, and pass through the saddle point parallel to the real axis. See Figure~\ref{fig:contoursarrowsq}.

Now we look at the behavior of the contours as $|w| \rightarrow \infty$. We will only look in the third quadrant, since the contours in the other quadrants can be found by symmetry. 
\begin{theorem}\label{thm:thetalargew}
Take $\alpha < 0$. For $w$ in the third quadrant with $|w| \gg 1$, the contours of steepest descent of both $p_\alpha(w)$ and $q_\alpha(w)$ that go to infinity can be parametrized as $w = re^{i\theta}$ where \begin{equation}
\theta =  -\frac{\pi}{2} -\frac{1}{\sqrt{2}}\alpha r^{-1/2}  + O(r^{-1}).\label{eq:contourslargew}
\end{equation} for $r \rightarrow \infty$.
\end{theorem}
\begin{proof}
For large $|w|$, we have \[
p_{\alpha}(w) =  -2iw + \alpha  (\sqrt{ - 2iw } -\sqrt{ 2iw}) + O(w^{-1/2}).\] and  \[
q_{\alpha}(w) =  -2iw + \alpha  (\sqrt{ - 2iw } + \sqrt{ 2iw}) + O(w^{-1/2}).\] Let $w = re^{i\theta}$, with $\theta \in (-\pi/2, 0)$. Then \begin{align*}
\sqrt{-2iw} &= \sqrt{2r} e^{i\left(\frac{\theta}{2} - \frac{\pi}{4}\right)},\\
\sqrt{2iw} &= \sqrt{2r} e^{i\left(\frac{\theta}{2} + \frac{\pi}{4}\right)}.\end{align*} Let $\phi = \frac{\theta}{2} + \frac{\pi}{4}$, so $\sqrt{-2iw} = \sqrt{2r}(-i\cos\phi + \sin\phi)$ and $\sqrt{2iw} = \sqrt{2r}(\cos\phi + i\sin\phi)$. Then 
\begin{equation}
\mathrm{Im}(p_{\alpha}(w)) =  -2r\cos\theta + \alpha  \sqrt{2r}(-\cos\phi -\sin\phi) + O(w^{-1/2})\label{eq:implarge}
\end{equation} and \begin{equation}\label{eq:imqlarge}
\mathrm{Im}(q_{\alpha}(w)) =  -2r\cos\theta + \alpha  \sqrt{2r}(-\cos\phi + \sin\phi) + O(w^{-1/2}). 
\end{equation} First we look at contours of constant $\mathrm{Im}(p_{\alpha}(w))$. Clearly we must have $\cos\theta \rightarrow 0$ as $r \rightarrow \infty$, so $\theta \rightarrow -\pi/2$ and $\phi \rightarrow 0$. We have $\cos\theta = \sin 2\phi$. We must have $\phi = O(r^{-1/2})$. Then the contour of steepest descent of $p_\alpha(w)$ in the third contour satisfies \begin{align*}
O(1) &= -2r(2\phi + O(\phi^3)) + \alpha\sqrt{2}\sqrt{r}(-1-\phi^2 -\phi + O(\phi^3)) \\
&= -4r\phi - \alpha \sqrt{2}\sqrt{r} + O(1)
\end{align*} So \begin{equation}
\phi = -\frac{1}{2\sqrt{2}}\alpha r^{-1/2} + O(r^{-1})\label{eq:contourslargephi}\end{equation} and from this we can deduce Equation~\ref{eq:contourslargew}. Note that since $\alpha < 0$ we have $\theta \in (-\pi/2, 0)$ as required. The contour of steepest descent of $q_\alpha(w)$ in the third contour also satisfies Equation~\ref{eq:contourslargephi} and therefore Equation~\ref{eq:contourslargew} since the $\sqrt{2r}\sin\phi$ terms in Equations \ref{eq:implarge} and \ref{eq:imqlarge} are $O(1)$. 
\end{proof}
From this we can prove the following corollary.
\begin{corollary}\label{cor:expboundcontour}
We have the following bounds on the exponential parts of the integrands as stated in Equation~\ref{eq:expparts}. There exists $A_0, c_0 > 0$ such that for $|w| > m^\delta$ on a contour of steepest descent that goes to infinity,
\begin{align*}
|\exp (\bconst^2 p_{\alpha}(w))| &< A_0 e^{-c_0 |w|} \\
|\exp (\bconst^2 q_{\alpha}(w))| &< A_0 e^{-c_0 |w|}
\end{align*}
\end{corollary}
\begin{proof}
Parametrize $w$ as $w = r e^{i\theta}$ so $|w| = r$. Then by Theorem~\ref{thm:thetalargew} we have \[\theta =  -\frac{\pi}{2} -\frac{1}{\sqrt{2}}\alpha r^{-1/2}  + O(r^{-1}).\] Let $\phi = \frac{\theta}{2} + \frac{\pi}{4}$ as in the proof of the theorem. Then we have \begin{align*}
\mathrm{Re}(p_\alpha(w) &= 2r\sin \theta + \alpha (\sqrt{2r}\sin \phi - \sqrt{2r}\cos\phi)\\
&=2r(-1 + O(r^{-1})) + \alpha(\sqrt{2r}O(r^{-1/2}) - \sqrt{2r}(1 + O(r^{-1})) \\
&= -2r -\alpha \sqrt{2r} + O(1)
\end{align*}
So \[
|e^{\bconst^2 p_\alpha(w)}| = e^{\bconst^2(-2r-\alpha\sqrt{2r})} + O(1) \] Similarly, we find \[
|e^{\bconst^2 q_\alpha(w)}| = e^{\bconst^2(-2r+\alpha\sqrt{2r})} + O(1) \] So there exists $A_0, c_0>0$ such that \[|\exp (\bconst^2 p_{\alpha}(w))| < A_0 e^{-c_0 r} \text{ and }
|\exp (\bconst^2 q_{\alpha}(w))| < A_0 e^{-c_0 r}\] as required.
\end{proof}

\subsection{Contours of integration}

Recall that we need to evaluate the four integrals $\mathcal{I}^{j,k}_{\e1,\e2} (a, x_1, x_2, y_1, y_2)$, for $j,k\in \{0,1\}$ defined in Equation~\ref{eq:mathcalI}. These are double integrals in $\w1$ and $\w2$ over the contours $\circlecontour_r$ and $\circlecontour_{1/r}$, which are circles centered at the origin of radius $r$ and $1/r$ respectively, for $\sqrt{2c}<r<1$. The only singularities in the integrand are at $0$ and at $w=z$, and there are branch cuts on the imaginary axis at $i(-\infty, -1/\sqrt{2c})\cup(-\sqrt{2c},\sqrt{2c})\cup (\sqrt{2c},\infty)$. We can move the contours as long as we don't cross the branch cut (which includes the origin). If we cross the contours over each other we will pick up a single integral over the residues.
 
We want to move the contours of integration so that in a neighborhood of $i$ they are steepest descent contours for the exponents in Equation~\ref{eq:expparts}, and away from this neighborhood the contribution to the integral is negligible.
 
We define the following contours of integration, based on the contours described in detail in Section~\ref{sec:contours}. They are shown for different values of $\alpha$ in Figure~\ref{fig:C1C2}. First we repeat the definition of the contours $\Cmainw, \Cmainz, \Cotherw$ and $\Cotherz$ from Section~\ref{sec:overviewofresults}.

\begin{definition}\label{def:C1C2}
Let  $p_\alpha(w)$ and $q_\alpha(w)$ be as defined in Equation~\ref{eq:pdef} and Equation~\ref{eq:qdef} respectively. Let $\eta$, $\eta'$ be as defined in Definition~\ref{def:etaeta'}.

For $-1/\sqrt{2}<\alpha < 0$, let $\Cmainw$ be the steepest descent contour for $p_{\alpha}(w)$ that is contained in the negative half plane and passes through the saddle point at $w = -\eta$. For $\alpha = -1/\sqrt{2}$ let $\Cmainw$ be the steepest descent contour for $p_{\alpha}(w)$ that passes through the saddle point at $w = 0$ and enters the negative half plane at angles of $-\pi/6$ and $-5\pi/6$. For $\alpha < -1/\sqrt{2}$, let $\Cmainw$ consist of the steepest descent contour for $p_{\alpha}(w)$ starts from the branch cut $i(1/4,\infty)$, passes through the saddle point at $w = -\eta$ and goes to infinity in the third quadrant; the reflection in the imaginary axis of this contour; and a contour that goes around the branch cut at $i/4$. This contour is shown in detail in Figure~\ref{fig:C1}.

Let $\Cmainz$ be the reflection of $\Cmainw$ in the real axis.

Let $\Cotherw$ be the steepest descent contour for $q_{\alpha}(w)$. This passes through $w = -\eta'$ and goes to infinity in the negative half plane.

Let $\Cotherz$ be the reflection of $\Cotherw$ in the real axis.
\end{definition}

Now we define finite restrictions of these contours. Fix $\delta$ with $0 < \delta < 1/2$.
\begin{definition}\label{def:tildeC1C2}
Let $\Cmainwres$, $\Cotherwres$, $\Cmainzres$ and $\Cotherzres$ be the restrictions in the $w-$plane of $\Cmainw$, $\Cotherw$, $\Cmainz$ and $\Cotherz$ (defined in Definition~\ref{def:C1C2} above) to the region $|w| \leq m^\delta$.
\end{definition}

Finally we define contours $\Cmainwm$, $\Cotherwm$, $\Cmainzm$ and $\Cotherzm$ in the $\omega$-plane that the contours $\circlecontour_r$ and $\circlecontour_{1/r}$ will be deformed to. 
\begin{definition}\label{def:C1mC2m}
First we define $\Cmainwm$ and $\Cotherwm$. Near $\omega=i$, specifically for $|\omega - i| \leq \bconst^2 m^{\delta-1} w$, we can write $\omega = i + \bconst^2 m^{-1} w$ for $|w| \leq m^\delta$. Let the $\omega$-contour  $\Cmainwm$ in this region agree with the $w-$contour  $\widetilde{\mathcal{C}_{1}}$, and let the $\omega$-contour  $\Cotherwm$ in this region agree with the $w$-contour  $\widetilde{\mathcal{C}_{1}}'$. Near $\omega = -i$, define the contours so they are symmetric in the real axis. Outside of the regions defined by $|\omega \pm i| \leq \bconst^2 m^{\delta-1} $, join up the contours by arcs of constant radius $R_1$. Note that by Theorem~\ref{thm:thetalargew}, when $|w| = m^\delta$ is on $\widetilde{\mathcal{C}_{1}}$ or $\widetilde{\mathcal{C}_{1}}'$, we have $w = -m^\delta i \pm \alpha m^{\delta/2}/\sqrt{2} + O(1)$, so $\omega = i(1 - \bconst^2 m^{\delta-1}) + O(m^{\delta/2-1})$, and hence $R_1 = 1 - \bconst^2 m^{\delta-1} + O(m^{\delta/2-1})$.

We define $\Cmainzm$ and $\Cotherzm$ similarly, so that in a neighborhood of $i$ they agree with $\widetilde{\mathcal{C}_{2}}$ and $\widetilde{\mathcal{C}_{2}}'$ respectively, are symmetric in the real axis, and are joined by arcs of radius $R_2$ outside of the regions defined by $|\omega \pm i| \leq \bconst^2 m^{\delta-1}$. Here $R_2 = 1 + \bconst^2 m^{\delta-1}+ O(m^{\delta/2-1})$.
\end{definition}

Figure~\ref{fig:C1mC2m} shows the contours $\Cmainwm$ in blue and $\Cmainzm$ in orange for $\bconst = 1$, $\alpha = -3$, $m = 64$, $\delta = 9/20$.

\begin{figure}[htbp]
\centering
\includegraphics[width = 0.6\textwidth]{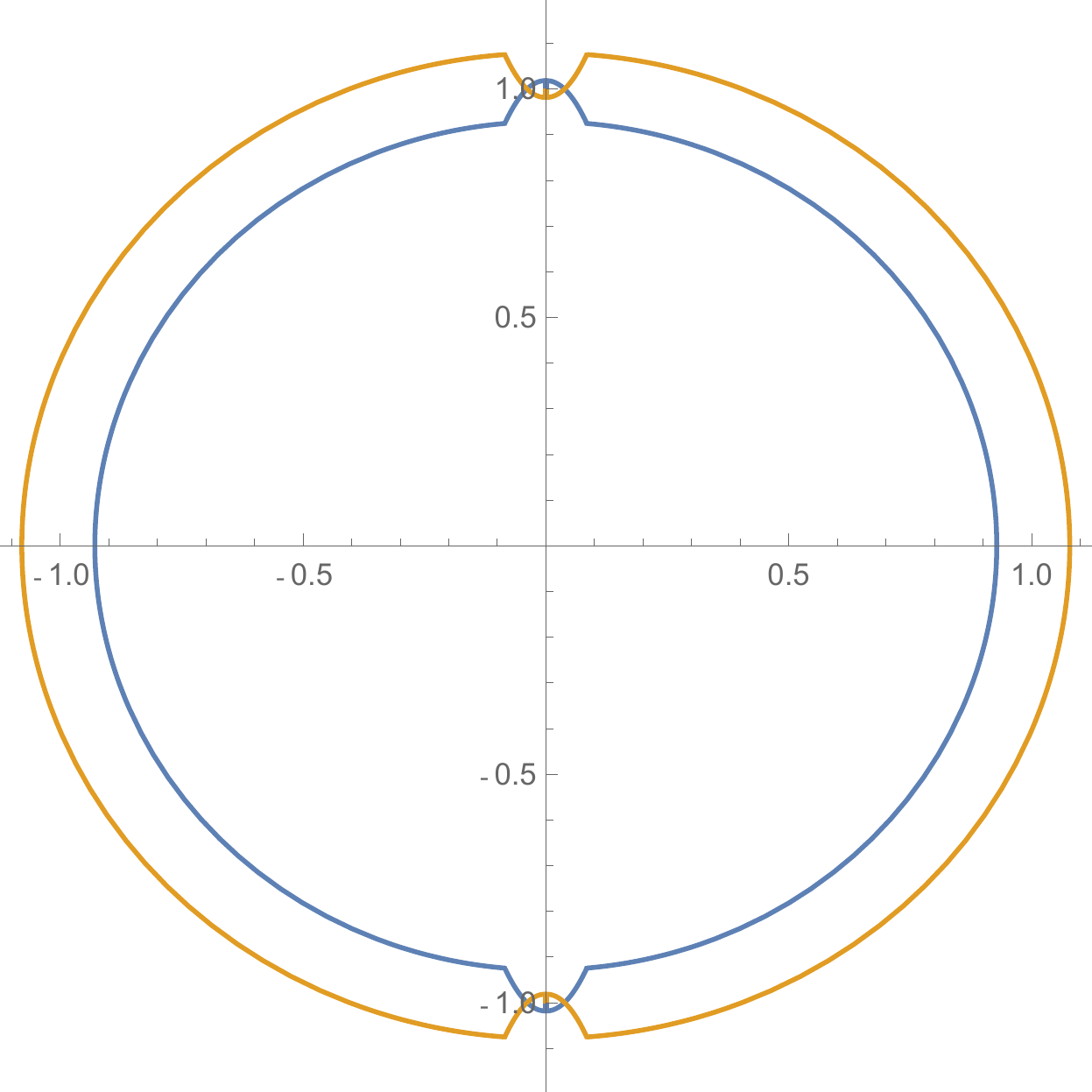}
\caption{The deformed contours for the integral $\mathcal{B}_{\e1, \e2} (a, x_1, x_2, y_1, y_2)$ with $\bconst = 1$, $\alpha = -3$, $m = 64$.}
\label{fig:C1mC2m}
\end{figure}

\subsection{Integral formulas}\label{sec:integralformulas}

For $\e1,\e2\in \{0,1\}$, let $x = (x_1,x_2) \in \mathtt{W}_\e1$, $y = (y_1,y_2) \in \mathtt{B}_\e2$, $\alpha_x < 0$ and $\alpha_y < 0$ be as in Equation~\ref{eq:asymptotic_coords_xy}. From now on we will assume that $\alpha_x = \alpha_y$. The case where $\alpha_x \neq \alpha_y$ will be dealt with in a later paper. We will deform the contours to the contours $\Cmainwm$, $\Cotherwm$, $\Cmainzm$ and $\Cotherzm$ in Definition~\ref{def:C1mC2m} as appropriate. We will show that on the arc segment away from the saddle points, the contribution is negligible. Then we are left with integrals $\Cmainwres$, $\Cotherwres$, $\Cmainzres$ and $\Cotherzres$. We will show that the difference between these and the integrals $\Cmainwm$, $\Cotherwm$, $\Cmainzm$ and $\Cotherzm$ is also negligible. In some cases, we need to cross the contours, and in these cases we will pick up a single integral from the residues. First we make some definitions.
\begin{definition}\label{def:Dintegrals}
Let $\Cmainwm$, $\Cotherwm$, $\Cmainzm$ and $\Cotherzm$ be as defined in Definition~\ref{def:C1mC2m}. Define \begin{equation*}
\mathcal{D}^{j,k}_{\e1,\e2}(a, x_1, x_2, y_1, y_2) = \frac{i^{y_1-x_1}}{(2\pi i)^2} \int_{\mathcal{C}_{j,m}} \frac{d\w1}{\w1} \int_{\mathcal{C}_{k,m}'}d\w2 \frac{V_{\e1, \e2}^{j,k} (\w1, \w2)}{\w2 - \w1} h_{j,k}(\w1,\w2)
\end{equation*}
\end{definition}
Note that the integrand of $\mathcal{D}^{j,k}_{\e1,\e2}(a, x_1, x_2, y_1, y_2)$ is the same as $\mathcal{I}_{j,k}(a, x_1, x_2, y_1, y_2)$, but the contours are different.
\begin{definition}\label{def:Eintegrals}
Let $\Cmainwm$, $\Cotherwm$, $\Cmainzm$ and $\Cotherzm$ be as defined in Definition~\ref{def:C1mC2m}. Let $\eta$ be as in Definition~\ref{def:etaeta'}, so the contours $\Cmainw$ and $\Cmainz$ cross at $\omega = i \pm \bconst m^{-1} \eta$. Let $\gamma_{00}$ be any contour in the $\omega$-plane from $\omega = i + \bconst m^{-1} \eta$ to $\omega = i - \bconst m^{-1} \eta$ that crosses the imaginary axis in the interval $i(\sqrt{2c},1/\sqrt{2c})$. Let $z_0 = x_0 + iy_0$ be the point in the first quadrant where the contours $\Cmainw$ and $\Cotherz$ cross. Let $\gamma_{10}$ be any contour in the $\omega$-plane between $i + \bconst m^{-1}(x_0 + iy_0)$ and $i + \bconst m^{-1}(x_0 - iy_0)$ that crosses the imaginary axis in the interval $i(\sqrt{2c},1/\sqrt{2c})$. Let $\gamma_{01}$ be any contour in the $\omega$-plane between $i + \bconst m^{-1}(-x_0 - iy_0)$ and $i + \bconst m^{-1}(-x_0 + iy_0)$ that crosses the imaginary axis in the interval $i(\sqrt{2c},1/\sqrt{2c})$. Define \begin{align}
 \mathcal{E}^{0,0}_{\e1,\e2}(a, x_1, x_2, y_1, y_2) &= \frac{i^{y_1-x_1}}{\pi i} \int_{\gamma_{00}} V_{\e1, \e2}^{0,0} (\omega, \omega) h_{0,0}(\omega, \omega) \frac{d\omega}{\omega}\label{eq:C00def}\\
\mathcal{E}^{1,0}_{\e1,\e2}(a, x_1, x_2, y_1, y_2) &= \frac{i^{y_1-x_1}}{\pi i} \int_{\gamma_{10}} V_{\e1, \e2}^{1,0} (\omega, \omega) h_{1,0}(\omega, \omega)\frac{d\omega}{\omega}\label{eq:C10def} \\
\mathcal{E}^{0,1}_{\e1,\e2}(a, x_1, x_2, y_1, y_2) &= \frac{i^{y_1-x_1}}{\pi i} \int_{\gamma_{01}} V_{\e1, \e2}^{0,1} (\omega, \omega) h_{0,1}(\omega,\omega)\frac{d\omega}{\omega}.\label{eq:C01def}
\end{align}
\end{definition}

We now prove the following theorem.

\begin{theorem}\label{thm:contourdeformations}
Let $\Cmainwm$, $\Cotherwm$, $\Cmainzm$ and $\Cotherzm$ be as defined in Definition~\ref{def:C1mC2m}. Recall the double integrals $\mathcal{I}^{j,k}_{\e1,\e2} (a, x_1, x_2, y_1, y_2)$ defined in Equation~\ref{eq:mathcalI}. Let $\mathcal{D}^{j,k}_{\e1,\e2}(a, x_1, x_2, y_1, y_2)$ for $j,k \in \{0,1\}$ be as in Definition~\ref{def:Dintegrals}. Then for $-1/\sqrt{2} \leq \alpha < 0$ we have
\begin{equation}
\mathcal{I}^{j,k}_{\e1,\e2} (a, x_1, x_2, y_1, y_2) = \mathcal{D}^{j,k}_{\e1,\e2}(a, x_1, x_2, y_1, y_2)
\label{eq:Igasdeformation}
\end{equation} i.e. we can deform the contours of integration to these contours without changing the value of the integral. However, for $\alpha < -1/\sqrt{2}$, in all but the last case we pick up an extra term when deforming the contours. Let $\mathcal{E}^{j,k}_{\e1,\e2}(a, x_1, x_2, y_1, y_2)$ for $(j,k) = (0,0)$, $(0,1)$ and $(1,0)$ be as in Definition~\ref{def:Eintegrals}. Then for $\alpha < -1/\sqrt{2}$ we have
\begin{align}\label{eq:I00liquiddeformation}
\mathcal{I}^{0,0}_{\e1,\e2} (a, x_1, x_2, y_1, y_2) &=  \mathcal{D}^{0,0}_{\e1,\e2}(a, x_1, x_2, y_1, y_2) + \mathcal{E}^{0,0}_{\e1,\e2}(a, x_1, x_2, y_1, y_2)\\\label{eq:I10liquiddeformation}
\mathcal{I}^{1,0}_{\e1,\e2} (a, x_1, x_2, y_1, y_2) &=\mathcal{D}^{1,0}_{\e1,\e2}(a, x_1, x_2, y_1, y_2) + \mathcal{E}^{1,0}_{\e1,\e2}(a, x_1, x_2, y_1, y_2)\\\label{eq:I01liquiddeformation}
\mathcal{I}^{0,1}_{\e1,\e2} (a, x_1, x_2, y_1, y_2) &= \mathcal{D}^{0,1}_{\e1,\e2}(a, x_1, x_2, y_1, y_2) + \mathcal{E}^{0,1}_{\e1,\e2}(a, x_1, x_2, y_1, y_2) \\
\mathcal{I}^{1,1}_{\e1,\e2} (a, x_1, x_2, y_1, y_2) &= \mathcal{D}^{1,1}_{\e1,\e2}(a, x_1, x_2, y_1, y_2) \label{eq:I11liquiddeformation}
\end{align}
\end{theorem} 
\begin{proof}
First deform $\circlecontour_{1/r}$ to $\mathcal{C}_{k,m}'$, and deform $\circlecontour_r$ to a contour $\mathcal{C}$ lying entirely inside this contour, that agrees with $\widetilde{\mathcal{C}_j}$ as appropriate on the circular arc section, and is symmetric in the real axis. Under this deformation the contours do not cross the branch cut or each other. Now we consider different cases corresponding to different values of $\alpha$. 

If $-1/\sqrt{2} < \alpha < 0$, then $\mathcal{C}_{k,m}'$ lies entirely outside the unit circle, and $\mathcal{C}_{j,m}$ lies entirely inside the unit circle. So we can deform $\mathcal{C}$ to $\mathcal{C}_{j,m}$ without crossing the contours. In this process we also do not cross any branch cuts. This proves Equation~\ref{eq:Igasdeformation}.

 If $\alpha < -1/\sqrt{2}$, then the contours $\Cmainwm$ and $\Cmainzm$ cross the unit circle, but the $\Cotherwm$ and $\Cotherzm$ do not. For $\mathcal{I}^{1,1}_{\e1,\e2} (a, x_1, x_2, y_1, y_2)$, we can move the $\w1-$contour $\mathcal{C}$ to $\Cotherwm$ without crossing the contours, which proves Equation~\ref{eq:I11liquiddeformation}.  In the other cases the contours will cross. 
 
 First we look at the integral $\mathcal{I}^{0,0}_{\e1,\e2} (a, x_1, x_2, y_1, y_2)$. Note that we have    \[
\mathcal{I}^{0,0}_{\e1,\e2} (a, x_1, x_2, y_1, y_2) = \frac{i^{y_1-x_1}}{(2\pi i)^2}  \int_{\Cmainzm}d\w2 \int_{\mathcal{C}} \frac{d\w1}{\w1} \frac{V_{\e1, \e2}^{0,0} (\w1, \w2)}{\w2 - \w1} h_{0,0}(\w1,\w2)\] The contour $\Cmainwm - \mathcal{C}$ has two components, near $\omega = \pm i$ respectively. Let $\Lambda$ be the component that is near $i$, so $\Cmainwm = \mathcal{C} + \Lambda + (-\Lambda)$ and hence  \[
  \int_{\Cmainwm} = \int_\mathcal{C} + \int_\Lambda + \int_{-\Lambda}.\] So writing $\mathcal{I}^{0,0}_{\e1,\e2} (a, x_1, x_2, y_1, y_2)$ as in Equation~\ref{eq:I00liquiddeformation} we have \begin{equation}\label{eq:C00lambda}
   \mathcal{E}^{0,0}_{\e1,\e2}(a, x_1, x_2, y_1, y_2) = -\frac{i^{y_1-x_1}}{(2\pi i)^2}  \int_{\Cmainzm}d\w2 \int_{\Lambda + (-\Lambda)} \frac{d\w1}{\w1} \frac{V_{\e1, \e2}^{0,0} (\w1, \w2)}{\w2 - \w1} h_{0,0}(\w1,\w2)
  \end{equation} \begin{figure}[htb]
  \centering
  \includegraphics[width=0.5\textwidth]{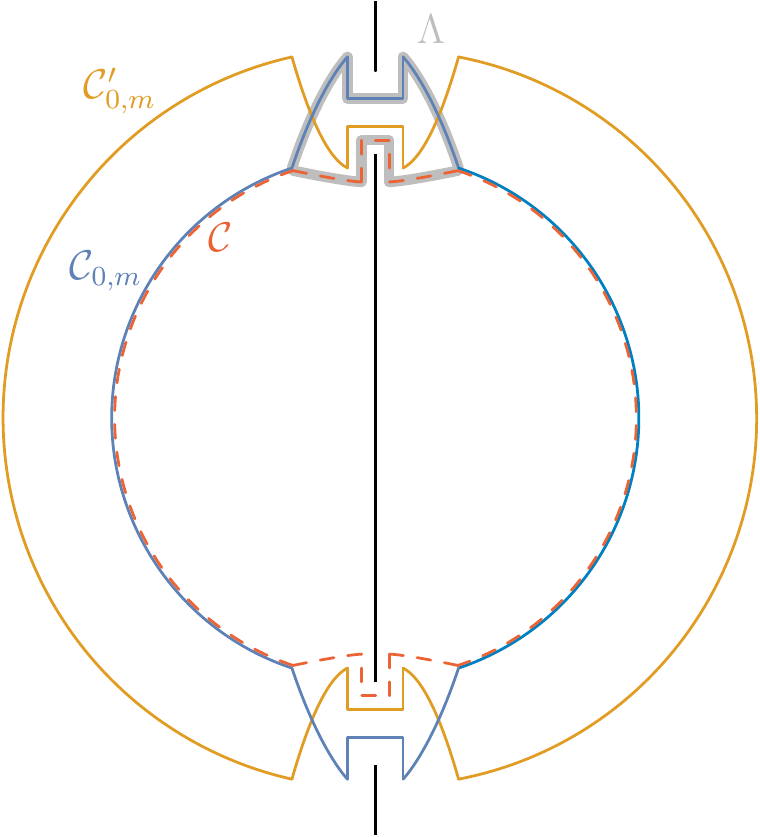}
  \caption{The contours $\Cmainwm$, $\Cmainzm$, $\mathcal{C}$ and $\Lambda$. The branch cut is shown in black.}
  \label{fig:contourdeformations}
\end{figure}    Let $\gamma_{00}$ be the segment of $\Cmainzm$ that is inside the contour $\Lambda$, which is exactly the segment of $\Cmainzm$ near $i$ that is inside $\Cmainwm$. These contours are shown in Figure~\ref{fig:contourdeformations}. Let $\omega_2 \in \gamma_{00}$, so $\omega_2$ is inside the contour $\Lambda$. By the residue theorem we have
 \[
  \int_{\Lambda} \frac{d\w1}{\w1}  \frac{V_{\e1, \e2}^{0,0} (\w1, \w2)}{\w2 - \w1} h_{0,0}(\w1,\w2) = -2\pi i  \frac{V_{\e1, \e2}^{0,0} (\w2, \w2)}{\w2}h_{0,0}(\w2,\w2) \] and similarly  \begin{align*}
  \int_{-\Lambda} \frac{d\w1}{\w1}  \frac{V_{\e1, \e2}^{0,0} (\w1, -\w2)}{(-\w2) - \w1} h_{0,0}(\w1,-\w2) &= -2\pi i  \frac{V_{\e1, \e2}^{0,0} (-\w2, -\w2)}{-\w2}h_{0,0}(-\w2,-\w2)\\
  &= 2\pi i  \frac{V_{\e1, \e2}^{0,0} (\w2, \w2)}{\w2} h_{0,0}(\w2,\w2)\end{align*} where we use Theorem~\ref{thm:integrandeven} in the last line. Then from Equation~\ref{eq:C00lambda} we have \begin{align*}
   \mathcal{E}^{0,0}_{\e1,\e2}(a, x_1, x_2, y_1, y_2) &= -\frac{i^{y_1-x_1}}{(2\pi i)^2}  \int_{\Cmainzm}d\w2 \int_{\Lambda + (-\Lambda)} \frac{d\w1}{\w1} \frac{V_{\e1, \e2}^{0,0} (\w1, \w2)}{\w2 - \w1}h_{0,0}(\w1,\w2)  \\
    &=-\frac{i^{y_1-x_1}}{(2\pi i)^2} \int_{\gamma_{00} + (-\gamma_{00})}d\w2 \int_{\Lambda} \frac{d\w1}{\w1} \frac{V_{\e1, \e2}^{0,0} (\w1, \w2)}{\w2 - \w1}h_{0,0}(\w1,\w2)  \\
   &= \frac{i^{y_1-x_1}}{\pi i} \int_{\gamma_{00}}d\w2   \frac{V_{\e1, \e2}^{0,0} (\w2, \w2)}{\w2}h_{0,0}(\w2,\w2)
   \end{align*} Equations \ref{eq:C10def} and \ref{eq:C01def} are obtained in a similar fashion.
 
\end{proof}

First we will find the asymptotics for the double integrals, and then we will look at the integrals $\mathcal{E}^{j,k}_{\e1,\e2}(a, x_1, x_2, y_1, y_2)$. First we prove the following theorem that bounds the integrand away from the saddle points.
\begin{theorem}\label{thm:arcerrors}
Let $\widetilde{H}_{x_1,x_2}(\omega)$ be as defined in Equation~\ref{eq:Htilde}.  Let $\circlecontour_{(m)}$ denote the part of the contour $\mathcal{C}_{k,m}$ that has constant radius $R_1$ as defined in Definition~\ref{def:C1mC2m}.  Let $\circlecontour_{(m)}'$ denote the part of the contour $\mathcal{C}_{k,m}'$ that has constant radius $R_2$ as defined in Definition~\ref{def:C1mC2m}. Then there exists $d > 0$ such that when $m$ is sufficiently large, for all $\w1\in \circlecontour_{(m)}$ we have \begin{align*}
|\widetilde{H}_{x_1 + 1, x_2}(\w1)| &< e^{-d m^\delta}, \\
|\widetilde{H}_{x_1 + 1, 2n - x_2}(\w1)| &< e^{-d m^\delta}
\end{align*} and for all $\w2 \in \circlecontour_{(m)}'$ we have
\begin{align*}
|\widetilde{H}_{y_1, y_2 + 1}(\w2)| > e^{d m^\delta},\\
|\widetilde{H}_{2n-y_1, y_2 + 1}(\w2)| > e^{d m^\delta}.
\end{align*}
\end{theorem}

The proof can be found in Appendix~\ref{sec:arcerrors}.

\begin{theorem}
Let $V_{\e1, \e2}^{j,k} (\w1, \w2)$ be as defined in Equation~\ref{eq:Voriginal} and $h_{j,k}(\w1,\w2)$ as defined in Equation~\ref{eq:hij}. Let $\widetilde{\mathcal{C}_j}$ and $\widetilde{\mathcal{C}_k}'$ be as defined in Definition~\ref{def:C1mC2m}. Let $\circlecontour_{(m)}$ denote the part of the contour $\mathcal{C}_{k,m}$ that has constant radius $R_1$ as defined in Definition~\ref{def:C1mC2m}.  Let $\circlecontour_{(m)}'$ denote the part of the contour $\mathcal{C}_{k,m}'$ that has constant radius $R_2$ as defined in Definition~\ref{def:C1mC2m}. Then there exists $A_1 > 0, c_1 > 0$ such that
\begin{equation*}
\frac{i^{y_1-x_1}}{(2\pi i)^2} \int_{\circlecontour_{(m)}} \frac{d\w1}{\w1} \int_{\mathcal{C}_{k,m}'}d\w2 \frac{V_{\e1, \e2}^{j,k} (\w1, \w2)}{\w2 - \w1} h_{j,k}(\w1,\w2) < A_1 e^{-c_1 m^\delta}
\end{equation*}
and 
\begin{equation*}
\frac{i^{y_1-x_1}}{(2\pi i)^2} \int_{\mathcal{C}_{j,m}} \frac{d\w1}{\w1} \int_{\circlecontour_{(m)}'}d\w2 \frac{V_{\e1, \e2}^{j,k} (\w1, \w2)}{\w2 - \w1} h_{j,k}(\w1,\w2) < A_1 e^{-c_1 m^\delta}
\end{equation*}
\end{theorem}
\begin{proof}
This follows from Theorem~\ref{thm:arcerrors} since the dependence of $V_{\e1, \e2}^{j,k} (\w1, \w2)$ on $m$ is not exponential, and $|\w2 - \w1|$ is bounded below by $\bconst^2 m^{\delta-1} + O(m^{\delta/2-1})$. 
\end{proof}

Then we are left with integrals over the parts of the contour corresponding to $\Cmainwres$, $\Cotherwres$, $\Cmainzres$ and $\Cotherzres$. 

Next we do a change of variables as in Definition~\ref{def:asymptoticsvars} and substitute the asymptotic expansions in Theorem~\ref{thm:Hw1w2expansion} and Theorem~\ref{thm:Vexpansion}. We will need the following lemma to bound the error terms in the exponent after substituting the asymptotic expansions.
\begin{lemma}\label{lemma:experrorbound}
Let $A^{j,k}_{\e1,\e2}(w,z)$ be as defined in Equation~\ref{eq:Aall} and  $g_{j,k}(w,z)$ be as defined in Equation~\ref{eq:gij}. Let the contours $\widetilde{\mathcal{C}_j}$, $\widetilde{\mathcal{C}_k}'$ be as defined in Definition~\ref{def:tildeC1C2}. Then
\[\int_{\widetilde{\mathcal{C}_j}}dw \int_{\widetilde{\mathcal{C}_k}'}dz \frac{A^{j,k}_{\e1,\e2}(w,z)}{z-w} e^{g_{j,k}(w,z) + O(m^{-1/2}w, m^{-1/2}z)} 
= \int_{\widetilde{\mathcal{C}_j}}dw \int_{\widetilde{\mathcal{C}_k}'}dz \frac{A^{j,k}_{\e1,\e2}(w,z)}{z-w} e^{g_{j,k}(w,z)} + O(m^{-1/2})\]
\end{lemma}
\begin{proof}
The proof can be found in Appendix~\ref{sec:experrorbound}.
\end{proof}

Now we move to asymptotic coordinates to approximate our integrals.
\begin{lemma}\label{lemma:DtildeC}
Let $\mathcal{D}^{j,k}_{\e1,\e2}(a, x_1, x_2, y_1, y_2)$ be as defined in Definition~\ref{def:Dintegrals}. Let the contours $\Cmainwres$, $\Cotherwres$, $\Cmainzres$ and $\Cotherzres$ be as defined in Definition~\ref{def:tildeC1C2}. Recall the functions $g_{j,k}(w,z)$ defined in Equation~\ref{eq:gij}. Then \begin{multline}
\mathcal{D}^{j,k}_{\e1,\e2}(a, x_1, x_2, y_1, y_2) =  \bconst m^{-1/2}\frac{i^{y_1-x_1-1}(-1)^{\e1\e2}i^{\e1-\e2}}{8(2\pi i)^2}\\
\times\int_{\widetilde{\mathcal{C}_j}}dw \int_{\widetilde{\mathcal{C}_k}'}dz \frac{A^{j,k}_{\e1,\e2}(w,z)}{z-w} e^{g_{j,k}(w,z)} + O(m^{-1})
\end{multline}
\end{lemma}
\begin{proof}
Recall from Definition~\ref{def:asymptoticsvars} the change of variables for $\w1,\w2$ near $i$ \begin{equation*}
\w1 = i + \bconst^2 m^{-1} w \hspace{0.5cm} \text{and} \hspace{0.5cm} \w2 = i + \bconst^2 m^{-1} z
\end{equation*} for $|w|, |z| < m^{\delta}$ for some $0 < \delta < 1/2$. Near $-i$, we use variables $-\w1,-\w2$. By Theorem~\ref{thm:integrandeven}, and since our contours are symmetric in the real axis, we can just integrate over $\w2$ in the upper half plane and double the result. Note that $d\w1 = \bconst^2 m^{-1} dw$ and $d\w2 = \bconst^2 m^{-1} dz$. Then by Theorem~\ref{thm:arcerrors} we have
\begin{multline*}
\mathcal{D}^{j,k}_{\e1,\e2}(a, x_1, x_2, y_1, y_2) = 2\frac{i^{y_1-x_1}}{(2\pi i)^2} \int_{\widetilde{\mathcal{C}_j}} \bconst^2 m^{-1}\frac{dw}{\w1} \int_{\widetilde{\mathcal{C}_k}'}\bconst^2 m^{-1}dz \frac{V_{\e1, \e2}^{j,k} (\w1, \w2)}{\w2 - \w1}h_{j,k}(\w1,\w2)\\
+ 2\frac{i^{y_1-x_1}}{(2\pi i)^2} \int_{\widetilde{\mathcal{C}_j}} \bconst^2 m^{-1}\frac{-dw}{-\w1} \int_{\widetilde{\mathcal{C}_k}'}\bconst^2 m^{-1}dz \frac{V_{\e1, \e2}^{j,k} (-\w1, \w2)}{\w2 + \w1} h_{j,k}(-\w1,\w2)+O(A_1 e^{-c_1 m^\delta})
\end{multline*}
Call these two terms $S_1$ and $S_2$ respectively, so $\mathcal{D}^{j,k}_{\e1,\e2}(a, x_1, x_2, y_1, y_2) = S_1 + S_2+O(A_1 e^{-c_1 m^\delta})$. Substituting the expansions in Theorem~\ref{thm:Hw1w2expansion} and Theorem~\ref{thm:Vexpansion}, the first term becomes \begin{multline*}
S_1 = 2\frac{i^{y_1-x_1}(-1)^{\e1\e2}i^{\e1-\e2}}{16(2\pi i)^2} \bconst m^{-1/2}\\
\times  \int_{\widetilde{\mathcal{C}_j}}\frac{dw}{i + \bconst^2 m^{-1} w}\int_{\widetilde{\mathcal{C}_k}'}dz\Bigg( \frac{A^{j,k}_{\e1,\e2}(w,z) + O(m^{-1/2})}{z-w} e^{g_{j,k}(w,z) + O(m^{-1/2}w, m^{-1/2}z)}\Bigg)
\end{multline*} where $A^{j,k}_{\e1,\e2}(w,z)$ is defined in Equation~\ref{eq:Aall}. Since $|w| < m^\delta$ with $0<\delta<1/2$, we have $1/(i + \bconst^2 m^{-1} w) = -i + O(m^{\delta - 1}) = -i + O(m^{-1/2})$. So we can write \begin{multline*}
S_1 = 2\frac{i^{y_1-x_1-1}(-1)^{\e1\e2}i^{\e1-\e2}}{16(2\pi i)^2} \bconst m^{-1/2}(1  + O(m^{-1/2}))\\
\times \int_{\widetilde{\mathcal{C}_j}}dw \int_{\widetilde{\mathcal{C}_k}'}dz \frac{A^{j,k}_{\e1,\e2}(w,z)}{z-w} e^{g_{j,k}(w,z) + O(m^{-1/2}w, m^{-1/2}z)}
\end{multline*} By Lemma~\ref{lemma:experrorbound}, the error term in the exponent gives rise to a global $O(m^{-1/2})$ error, so we have \[
S_1 = 2\frac{i^{y_1-x_1-1}(-1)^{\e1\e2}i^{\e1-\e2}}{16(2\pi i)^2} \bconst m^{-1/2}\int_{\widetilde{\mathcal{C}_j}}dw \int_{\widetilde{\mathcal{C}_k}'}dz \frac{A^{j,k}_{\e1,\e2}(w,z)}{z-w} e^{g_{j,k}(w,z)} + O(m^{-1})
\] Using Theorem~\ref{thm:integrandeven}, we can write the $S_2$ term as \[S_2 = 2\frac{i^{y_1-x_1}}{(2\pi i)^2} \int_{\widetilde{\mathcal{C}_j}} \bconst^2 m^{-1}\frac{dw}{\w1} \int_{\widetilde{\mathcal{C}_k}'}\bconst^2 m^{-1}dz \frac{V_{\e1, \e2}^{j,k} (\w1, \w2)}{\w1 + \w2} h_{j,k}(\w1,\w2)\] so using the same substitutions as above, but noting that $1/(\w1+\w2) = -i/2 + O(m^{-1})$, we see that $S_2 = O(m^{-3/2})$. So this term is negligible and the result follows. 
\end{proof}
Now we prove that we can replace the contours $\Cmainwres$, $\Cotherwres$, $\Cmainzres$ and $\Cotherzres$ by $\Cmainw$, $\Cotherw$, $\Cmainz$ and $\Cotherz$ with exponentially small error.
\begin{lemma}\label{lemma:Ccontourslargew}
Let the contours $\Cmainwres$, $\Cotherwres$, $\Cmainzres$ and $\Cotherzres$ be as defined in Definition~\ref{def:tildeC1C2} and the contours $\Cmainw$, $\Cotherw$, $\Cmainz$ and $\Cotherz$ be as defined in Definition~\ref{def:C1C2}. Then there exists $c_0 > 0$ such that,
\[\int_{\mathcal{C}_j}dw \int_{\mathcal{C}_k'}dz \frac{A^{j,k}_{\e1,\e2}(w,z)}{z-w} e^{g_{j,k}(w,z)} = \int_{\widetilde{\mathcal{C}_j}}dw \int_{\widetilde{\mathcal{C}_k}'}dz \frac{A^{j,k}_{\e1,\e2}(w,z)}{z-w} e^{g_{j,k}(w,z)} + O( e^{-c_0 m^\delta}) \]
\end{lemma}
\begin{proof}
Recall from Equation~\ref{eq:pdef}, Equation~\ref{eq:qdef} and Equation~\ref{eq:gij} that we can write \begin{align*}
g_{0,0}(\w1,\w2) &= \bconst^2(p_{\alpha}(w) + p_{\alpha}(-z))\\
g_{1,0}(\w1,\w2) &= \bconst^2(p_{\alpha}(w) + q_{\alpha}(-z)) \\
g_{0,1}(\w1,\w2) &= \bconst^2(q_{\alpha}(w) + p_{\alpha}(-z))\\
g_{1,1}(\w1,\w2) &= \bconst^2(q_{\alpha}(w) + q_{\alpha}(-z))
\end{align*}  By Corollary \ref{cor:expboundcontour}, there exists $A_0, c_0 > 0$ for $w \in \mathcal{C}_j\setminus \widetilde{\mathcal{C}_j}$, we have \[|\exp (\bconst^2 p_{\alpha}(w))| < A_0 e^{-c_0 |w|} \text{ and } |\exp (\bconst^2 q_{\alpha}(w))| < A_0 e^{-c_0 |w|} \] and for $z \in \mathcal{C}_k'\setminus \widetilde{\mathcal{C}_k}'$, we have \[|\exp (\bconst^2 p_{\alpha}(-z))| < A_0 e^{-c_0 |z|} \text{ and } |\exp (\bconst^2 q_{\alpha}(-z))| < A_0 e^{-c_0 |z|} \] First we show that  \[
\int_{\mathcal{C}_j\setminus\widetilde{\mathcal{C}_j}}dw \int_{\mathcal{C}_k'}dz \frac{A^{j,k}_{\e1,\e2}(w,z)}{z-w} e^{g_{j,k}(w,z)} =  O( e^{-c_2m^\delta})
\] We have \[
\left|\int_{\mathcal{C}_j\setminus\widetilde{\mathcal{C}_j}}dw \int_{\mathcal{C}_k'}dz \frac{A^{j,k}_{\e1,\e2}(w,z)}{z-w} e^{g_{j,k}(w,z)}\right| 
\leq A_0 \int_{\mathcal{C}_j\setminus\widetilde{\mathcal{C}_j}}e^{-c_0 |w|}|dw|\int_{\mathcal{C}_k'}\left|\frac{A^{j,k}_{\e1,\e2}(w,z)}{z-w} e^{\bconst^2(2iz + \alpha f^\pm(-z))}dz\right|
\] Note that the integrand has no singularities on these contours. By the bounds above, the inner integral converges. By Theorem~\ref{thm:thetalargew} we can write \[
w = -i r \pm \frac{1}{\sqrt{2}} \alpha r^{1/2} + O(1)\] where $r = |w|$. Then \[|dw| = \left|-idr \pm \frac{1}{2\sqrt{2}}\alpha r^{-1/2}dr\right| = (1 + O(r^{-1/2})) dr\] so we have \begin{multline*}
\left|\int_{\mathcal{C}_j\setminus\widetilde{\mathcal{C}_j}}dw \int_{\mathcal{C}_k'}dz \frac{A^{j,k}_{\e1,\e2}(w,z)}{z-w} e^{g_{j,k}(w,z)}\right| \\
\leq \Bigg( 2A_0 \int_{m^\delta}^\infty e^{-c_0 r}(1 + O(r^{-1/2})) dr \Bigg) \sup_{w\in \mathcal{C}_j\setminus\widetilde{\mathcal{C}_j}} \int_{\mathcal{C}_k'}\left|\frac{A^{j,k}_{\e1,\e2}(w,z)}{z-w} e^{\bconst^2 (2iz + \alpha f^\pm(-z))}dz\right|.
\end{multline*} This is $O(e^{-c_0 m^\delta})$. Similarly, we have \[
\left|\int_{\mathcal{C}_j}dw \int_{\mathcal{C}_k'\setminus\widetilde{\mathcal{C}_k}'}dz \frac{A^{j,k}_{\e1,\e2}(w,z)}{z-w} e^{g_{j,k}(w,z)}\right| = O(e^{-c_0 m^\delta}).
\] The result follows.
\end{proof}
Putting Lemma~\ref{lemma:DtildeC} and Lemma~\ref{lemma:Ccontourslargew} together we obtain the following theorem.
\begin{theorem}\label{thm:DasymptoticintegraloverC}
Let $\mathcal{D}^{j,k}_{\e1,\e2}(a, x_1, x_2, y_1, y_2)$ be as defined in Definition~\ref{def:Dintegrals}. Let the contours $\Cmainw$, $\Cotherw$, $\Cmainz$ and $\Cotherz$ be as defined in Definition~\ref{def:C1C2}. Recall the functions $g_{j,k}(w,z)$ defined in Equation~\ref{eq:gij} and the functions $A^{j,k}_{\e1,\e2}(w,z)$ defined in Equation~\ref{eq:Aall}. Then \begin{multline}
\mathcal{D}^{j,k}_{\e1,\e2}(a, x_1, x_2, y_1, y_2) =  \bconst m^{-1/2}\frac{i^{y_1-x_1-1}(-1)^{\e1\e2}i^{\e1-\e2}}{8(2\pi i)^2}\\
\times\int_{\mathcal{C}_j}dw \int_{\mathcal{C}_k'}dz \frac{A^{j,k}_{\e1,\e2}(w,z)}{z-w} e^{g_{j,k}(w,z)} + O(m^{-1})
\end{multline}
\end{theorem}
\begin{proof}
This follows directly from Lemma~\ref{lemma:DtildeC} and Lemma~\ref{lemma:Ccontourslargew}.
\end{proof}

Now we find the asymptotic behavior of the integrals $ \mathcal{E}^{j,k}_{\e1,\e2}(a, x_1, x_2, y_1, y_2)$.
\begin{theorem}\label{thm:Easymptoticintegral}
Let $\mathcal{E}^{j,k}_{\e1,\e2}(a, x_1, x_2, y_1, y_2)$ for $(j,k) = (0,0),(1,0)$ and $(0,1)$ be as defined in Definition~\ref{def:Eintegrals}. Then 
 \[\mathcal{E}^{0,0}_{\e1,\e2}(a, x_1, x_2, y_1, y_2) = -\bconst m^{-1/2} \frac{i^{y_1-x_1-1}(-1)^{\e1\e2}i^{\e1-\e2}}{16 \pi i} \int_{-\eta}^{\eta} A^{0,0}_{\e1,\e2}(w,w) dw + O(m^{-1})\] where 
 \begin{equation}\label{eq:Aww}
A^{0,0}_{\e1,\e2}(w,w) = \frac{-4(1 + (-1)^\e2\sqrt{1/2-2iw} + (-1)^\e1\sqrt{1/2+2iw})}{\sqrt{1/2-2iw}\sqrt{1/2+2iw}}.
 \end{equation} 
 Furthermore, we have
 \[
\mathcal{E}^{1,0}_{\e1,\e2}(a, x_1, x_2, y_1, y_2) = O(m^{-1}) \hspace{0.5cm}\text{and}\hspace{0.5cm} \mathcal{E}^{0,1}_{\e1,\e2}(a, x_1, x_2, y_1, y_2) = O(m^{-1}).
\]
\end{theorem}
\begin{proof}
From the definition of $\mathcal{E}^{j,k}_{\e1,\e2}(a, x_1, x_2, y_1, y_2)$ we can choose the contours $\gamma_{jk}$ so that they are at most a distance of $O(m^{-1})$ from $i$. Then we can apply the change of variables $\omega = i + \bconst^2 m^{-1} w$ from Equation~\ref{eq:omegaexpansion}. Let $\Gamma_{jk}$ be the image of $\gamma_{jk}$ in the $w$-plane. Deform $\Gamma_{jk}$ so that it consits of vertical segments from the endpoints to the real axis, and a segment along the real axis. This does not depend on $m$. Then using Theorem~\ref{thm:Hw1w2expansion} and Theorem~\ref{thm:Vexpansion} we have \begin{align*}
 \mathcal{E}^{j,k}_{\e1,\e2}(a, x_1, x_2, y_1, y_2) &= \frac{i^{y_1-x_1-1}}{\pi i} \int_{\gamma_{jk}} V_{\e1, \e2}^{j,k} (\omega, \omega) h_{j,k}(\omega, \omega) \frac{d\omega}{\omega} \\
 &= \bconst m^{-1/2}\frac{i^{y_1-x_1}(-1)^{\e1\e2}i^{\e1-\e2} }{16\pi i} \\
 &\hspace{1cm} \int_{\Gamma_{jk}}  (A^{j,k}_{\e1,\e2}(w,w)  + O(m^{-1/2})) e^{g_{j,k}(w,w) + O(m^{-1/2} w)} dw
\end{align*} where $A^{j,k}_{\e1,\e2}(w,z)$ is defined in Equation~\ref{eq:Aall} and $g_{j,k}(w,z)$ is defined in Equation~\ref{eq:gij}. We compute 
\[
A^{0,0}_{\e1,\e2}(w,w) = -\frac{4(1 + (-1)^\e2\sqrt{1/2-2iw} + (-1)^\e1\sqrt{1/2+2iw})}{\sqrt{1/2-2iw}\sqrt{1/2+2iw}},
\]
\[
A^{1,0}_{\e1,\e2}(w,w) = 0 \hspace{0.2cm}\text{ and }\hspace{0.2cm} A^{0,1}_{\e1,\e2}(w,w) = 0. 
\] Since $\Gamma_{jk}$ does not depend on $m$, it is clear that \[
\mathcal{E}^{1,0}_{\e1,\e2}(a, x_1, x_2, y_1, y_2) = O(m^{-1})\hspace{0.2cm}\text{ and }\hspace{0.2cm} \mathcal{E}^{0,1}_{\e1,\e2}(a, x_1, x_2, y_1, y_2) = O(m^{-1}).\] Moreover, the contour $\Gamma_{00}$ is just a straight line along the real axis from $\eta$ to $-\eta$, and we have $g_{0,0}(w,w) = 0$, hence  \[\mathcal{E}^{0,0}_{\e1,\e2}(a, x_1, x_2, y_1, y_2) = \bconst m^{-1/2} \frac{i^{y_1-x_1-1}(-1)^{\e1\e2}i^{\e1-\e2}}{16 \pi i} \int_{\eta}^{-\eta} A^{0,0}_{\e1,\e2}(w,w) dw + O(m^{-1}).\] We reverse the orientation of the contour to obtain the result.
\end{proof}

Now we prove a lemma to make sense of the quantity $i^{y_1-x_1}(-1)^{\e1\e2}i^{\e1-\e2}$ that appears in many of our formulas.

\begin{lemma}\label{lemma:uglycoefficient}
For $\e1,\e2\in \{0,1\}$, let $x = (x_1, x_2) \in \mathtt{W}_\e1 $, $y = (y_1, y_2)\in \mathtt{B}_\e2$ be vertices that are joined by an edge. Let $\zeta (x,y)$ be as defined in Equation~\ref{eq:zeta}. Then
\begin{equation}
\frac{\zeta(x,y)}{\Sigma(y,x)} = (-1)^{(y_1-x_1 - 1)/2}(-1)^{\e1\e2}i^{\e1-\e2 + 1}
\end{equation}
\end{lemma}
\begin{proof}
Firstly, if $\e1=\e2$ then $\Sigma(x,y) = i$, and if $\e1\neq\e2$ then $\Sigma(x,y) = 1$. Recall that $\zeta(x,y)=(-1)^{(y_2-x_1)/2}$. We have \begin{align*}
(-1)^{(y_1-x_1 - 1)/2}(-1)^{\e1\e2}i^{\e1-\e2 + 1}\zeta(x,y)^{-1} &= (-1)^{(-y_2+x_2 - 1)/2}(-1)^{\e1\e2}i^{\e1-\e2 + 1} \\
&= (-1)^{\e1\e2 - \e2 + 1}i^{\e1-\e2 + 1} \\
&= \begin{cases} -i &\text{ if } \e1=\e2 \\
1 & \text{ if } \e1 \neq \e2 \end{cases}
\end{align*}
\end{proof}

Now we are ready to prove Theorem~\ref{thm:Iasymptotics}.
\begin{proof}[Proof of Theorem~\ref{thm:Iasymptotics}]
Combining Theorem~\ref{thm:contourdeformations}, Theorem~\ref{thm:DasymptoticintegraloverC} and Theorem~\ref{thm:Easymptoticintegral}, we have the following formulas. Then for $-1/\sqrt{2} \leq \alpha < 0$,
  \begin{multline}
 \mathcal{I}^{0,0}_{\e1,\e2} (a, x_1, x_2, y_1, y_2) =  \bconst m^{-1/2}\frac{i^{y_1-x_1-1}(-1)^{\e1\e2}i^{\e1-\e2}}{8(2\pi i)^2}\\
\times\int_{\Cmainw}dw \int_{\Cmainz}dz \frac{A^{0,0}_{\e1,\e2}(w,z)}{z-w} e^{g_{0,0}(w,z)} + O(m^{-1}),
 \end{multline} and for $\alpha < -1/\sqrt{2}$,
 \begin{multline}
 \mathcal{I}^{0,0}_{\e1,\e2} (a, x_1, x_2, y_1, y_2) =  \bconst m^{-1/2}\frac{i^{y_1-x_1-1}(-1)^{\e1\e2}i^{\e1-\e2}}{8(2\pi i)^2}\\
\times\Bigg( \int_{\Cmainw}dw \int_{\Cmainz}dz \frac{A^{0,0}_{\e1,\e2}(w,z)}{z-w} e^{g_{0,0}(w,z)} - 2\pi i \int_{-\eta}^{\eta} A^{0,0}_{\e1,\e2}(w,w) dw\Bigg) + O(m^{-1}).
 \end{multline} For $(j,k) \neq (0,0)$ for any $\alpha < 0$, 
 \begin{multline}
 \mathcal{I}^{j,k}_{\e1,\e2} (a, x_1, x_2, y_1, y_2) =  \bconst m^{-1/2}\frac{i^{y_1-x_1-1}(-1)^{\e1\e2}i^{\e1-\e2}}{8(2\pi i)^2}\\
\times \int_{\mathcal{C}_j}dw \int_{\mathcal{C}_k'}dz \frac{A^{j,k}_{\e1,\e2}(w,z)}{z-w} e^{g_{j,k}(w,z)} + O(m^{-1})
 \end{multline}
Then we use Lemma~\ref{lemma:uglycoefficient} to finish the proof.
\end{proof}

This concludes the asymptotic analysis of $\mathcal{I}^{j,k}_{\e1,\e2} (a, x_1, x_2, y_1, y_2)$.

\section{Asymptotics of $ \mathbb{K}_{a,0,0}^{-1}((x_1,x_2),(y_1,y_2))$}\label{sec:Jasymptotics}

\subsection{Derivation of real integral formula}

In this section we look at the asymptotics of $\mathbb{K}_{a,0,0}^{-1}((x_1,x_2),(y_1,y_2))$, which appears in Equation~\ref{eq:Ka1inverse}. We finish with a proof of Theorem~\ref{thm:Kgasinversefullasymptotics}.

For $\e1,\e2\in \{0,1\}$, take $\mathbf{w} = (w_1,w_2) \in \mathtt{W}_\e1$ and $\mathbf{b} = (b_1,b_2) \in \mathtt{B}_\e2$ in the same fundamental domain, and $u, v \in \mathbb{Z}$. Recall that $e_1 = (1,1)$ and $e_2 = (-1,1)$. Let $\circlecontour_1$ denote a contour of unit radius centered at the origin, traversed in a counter-clockwise direction. Let $\mathcal{K}_a(z, w)^{-1}$ be as defined in Equation~\ref{eq:Kazwinverse}. From Equation~\ref{eq:translationinvariantoriginal} we have
\begin{equation}\label{eq:gastranslationinvariantoriginal}
\mathbb{K}_{a,0,0}^{-1}(\mathbf{w}, \mathbf{b} + 2ue_1 + 2ve_2) = \frac{1}{(2\pi i)^2} \int_{\circlecontour_1}\frac{dz}{z}\int_{\circlecontour_1}\frac{dw}{w} (\mathcal{K}_a(z, w)^{-1})_{\e1\e2} z^u w^v 
\end{equation}
where for convenience rows and columns of the $2\times 2$ matrix $\mathcal{K}_a(z, w)^{-1}$ are indexed by 0 and 1.

Recall from Equation~\ref{eq:Pazw} the characteristic polynomial
 \[P_a(z,w) = -2 -2a^2 - aw^{-1} - aw - az^{-1} - az,\] which is the determinant of $\mathcal{K}_a(z, w)$ and appears in the denominator of the integrand.
 
First we will prove some symmetry relations.
  \begin{lemma}\label{lemma:K11invequalities}
  Let  $\mathbf{w}_0\in \mathtt{W}_0$, $\mathbf{w}_1\in \mathtt{W}_1$, $\mathbf{b}_0\in \mathtt{B}_0$ and $\mathbf{b}_1\in \mathtt{B}_1$, be vertices in one fundamental domain, and take $u,v\in \mathbb{Z}$. Then \begin{multline}
  \label{eq:K11invequalitiesall} \mathbb{K}_{a,0,0}^{-1}(\mathbf{w}_0, \mathbf{b}_1 + 2ue_1 + 2ve_2) = \mathbb{K}_{a,0,0}^{-1}(\mathbf{w}_0, \mathbf{b}_1 + 2ue_1 - 2ve_2) \\= \mathbb{K}_{a,0,0}^{-1}(\mathbf{w}_1, \mathbf{b}_0 - 2ue_1 + 2ve_2)  = \mathbb{K}_{a,0,0}^{-1}(\mathbf{w}_1, \mathbf{b}_0 - 2ue_1 - 2ve_2) \\ =  i\mathbb{K}_{a,0,0}^{-1}(\mathbf{w}_0, \mathbf{b}_0 + 2ve_1 + 2ue_2) =  i\mathbb{K}_{a,0,0}^{-1}(\mathbf{w}_0, \mathbf{b}_0 - 2ve_1 + 2ue_2)\\ = i\mathbb{K}_{a,0,0}^{-1}(\mathbf{w}_1, \mathbf{b}_1 + 2ve_1 - 2ue_2)  = i\mathbb{K}_{a,0,0}^{-1}(\mathbf{w}_1, \mathbf{b}_1 - 2ve_1 - 2ue_2)
  \end{multline}
  \end{lemma}
  \begin{proof}
The first, third and sixth equalities follows from applying the change of variables $w \rightarrow w^{-1}$ to the integral in Equation~\ref{eq:gastranslationinvariantoriginal}. The second, fifth and seventh equalities follow from applying the change of variables $z \rightarrow z^{-1}$. The fourth equality follows from exchanging the variables $z$ and $w$.
  \end{proof}
  Hence it suffices to compute $\mathbb{K}_{a,0,0}^{-1}(\mathbf{w}_0, \mathbf{b}_1 + 2ue_1 + 2ve_2)$ for $u,v\in\mathbb{Z}$ with $v\geq 0$. 
  
\begin{lemma}\label{lemma:K11singlecontourintegral}
 Let $\mathbf{w}_0\in \mathtt{W}_0$ and $\mathbf{b}_1\in \mathtt{B}_1$ be vertices in one fundamental domain. Take $u,v\in\mathbb{Z}$ with $v\geq 0$. Then we can write
\begin{equation}\label{eq:K11singlecontourintegral}
\mathbb{K}_{a,0,0}^{-1}(\mathbf{w}_0, \mathbf{b}_1 + 2ue_1 + 2ve_2)= \frac{1}{2\pi i} \int_{\circlecontour_1} \frac{dz}{z} \frac{(a+z)\,z^u\,w_0(z)^v}{a\sqrt{(2(a + a^{-1}) + z+z^{-1})^2 - 4}}\end{equation} where \[
 w_0(z) = \frac{1}{2}\left(\sqrt{(2(a + a^{-1}) + (z+z^{-1}))^2 - 4} - (2(a + a^{-1}) + (z+z^{-1}))\right)\]
  and where the square root in the denominator refers to the principal branch of the square root. 
\end{lemma} 

\begin{proof}
We have \begin{equation}\label{eq:K11specificedge}
    \mathbb{K}_{a,0,0}^{-1}(\mathbf{w}_0, \mathbf{b}_1+ 2ue_1 + 2ve_2) = \frac{1}{(2\pi i)^2} \int_{\circlecontour_1}\frac{z^u dz}{z}\int_{\circlecontour_1}\frac{dw}{w} \frac{-(a+z)w^v}{P_a(z,w)}
 \end{equation}
We can use the residue theorem to evaluate the inner integral. Note that since $\lim_{w\rightarrow 0} w P_a(z,w) = -a$, the integrand does not have a singularity at 0 for any $v\geq 0$. Also note that for $a < 1$, $P_a(z,w)$ does not have any zeros with $(z,w) \in \circlecontour_1\times \circlecontour_1$. We want to find $w_0(z)$ where $P_a(w_0(z), z) = 0$ and $|w_0(z)| < 1$ for $z \in \circlecontour_1$. The quadratic formula gives \begin{align*}
 w_0(z) &= \frac{1}{2}\left(\sqrt{(2(a + a^{-1}) + (z+z^{-1}))^2 - 4} - (2(a + a^{-1}) + (z+z^{-1}))\right).\end{align*} Note that since $z \in \circlecontour_1$, we have $z + z^{-1} \in \mathbb{R}$. Since $a < 1$ and $ z + z^{-1} \geq -2$, we have $a + a^{-1} > 2$ and so $2(a + a^{-1}) + (z+z^{-1}) > 2$. Hence $w_0(z)$ is real for $z\in \circlecontour_1$. It is clear that $|w_0(z)| < 1$. Since $P_a(z,w)$ is invariant under $w \rightarrow -w$, it is clear that the other root of the quadratic equation is $w_0(z)^{-1}$, given by \[
 w_0(z)^{-1} = \frac{1}{2}\left(-\sqrt{(2(a + a^{-1}) + (z+z^{-1}))^2 - 4} - (2(a + a^{-1}) + (z+z^{-1}))\right).\]
 
To find the residue at $w_0$ we first calculate
\[
\frac{\partial P_a}{\partial w}(z, w) = a(w^{-2} - 1)\] So \[
\frac{1}{\frac{\partial P_a}{\partial w}(z, w_0)} = \frac{w_0}{a(w_0^{-1} - w_0)}.\] Since we have $|w_0(z)| < 1$ for $z \in \circlecontour_1$, this is always finite. Also note that we have $w_0(z)^{-1} - w_0(z) = -\sqrt{(2(a + a^{-1}) + (z+z^{-1}))^2 - 4}$. So we can write
\begin{align*}
 \frac{1}{2\pi i} \int_{\circlecontour_1}  \frac{-(a+z)w^v}{P_a(z, w)}  \frac{dw}{w} &= \frac{w_0(z)}{a(w_0^{-1}(z) - w_0(z))} \frac{-(a+z)w_0(z)^v}{w_0(z)}\\
 &= \frac{(a+z)w_0(z)^v}{a\sqrt{(2(a + a^{-1}) + (z+z^{-1}))^2 - 4}} .\end{align*}  
Substituting this expression into Equation~\ref{eq:K11specificedge} gives the required result.
\end{proof}
Now we will deform the contours to write this as a real integral. This requires some care with the square root. We first consider the case where $u+v\geq0$.

\begin{theorem}\label{thm:kgasinversereal}
 Let $\mathbf{w}_0\in \mathtt{W}_0$ and $\mathbf{b}_1\in \mathtt{B}_1$ be vertices in the same fundamental domain. Take $u,v\in\mathbb{Z}$ with $v\geq 0$ and $u+v\geq 0$. For $a < 1$, let \begin{align}\begin{split} z_1 &= -(a + a^{-1} - 1) + \sqrt{(a+a^{-1} -1)^2 - 1} \\ 
 z_2 &= -(a + a^{-1} + 1) + \sqrt{(a + a^{-1} + 1)^2 - 1}.\label{eq:z1z2}\end{split}\end{align} Also for $z\in\mathbb{R}$ with $z_1 \leq z \leq z_2$ let \begin{equation}
\theta_a(z) = \frac{1}{2z} \left(i \sqrt{4z^2 - (2(a+a^{-1})z + z^2 + 1)^2} - (2(a + a^{-1})z + z^2+1)\right).\label{eq:thetaa}
\end{equation} Then we can write
\begin{equation}
\mathbb{K}_{a,0,0}^{-1}(\mathbf{w}_0, \mathbf{b}_1 + 2ue_1 + 2ve_2)= \frac{1}{2\pi a} \int_{z_1}^{z_2} \frac{(a+z)\,z^u\, (\theta_a(z)^v + \overline{\theta_a(z)}^v)}{ \sqrt{4z^2 - (2(a+a^{-1})z + z^2 + 1)^2}} dz
\end{equation}
\end{theorem}

\begin{proof}
From Lemma~\ref{lemma:K11singlecontourintegral} we have \[\mathbb{K}_{a,0,0}^{-1}(\mathbf{w}_0, \mathbf{b}_1 + 2ue_1 + 2ve_2)= \frac{1}{2\pi i} \int_{\circlecontour_1} \frac{dz}{z} \frac{(a+z)\,z^u\,w_0(z)^v}{a\sqrt{(2(a + a^{-1}) + z+z^{-1})^2 - 4}}\] where $w_0(z)$ is defined in the statement of Lemma~\ref{lemma:K11singlecontourintegral}.
 Since $z + z^{-1}$ is real and greater that $-2$ on the torus, we have $(2(a + a^{-1}) + (z+z^{-1}))^2 - 4 > 0$ for $z \in \circlecontour_1$, so there are no branch cuts on the contour of integration. Let \begin{equation}
 \phi_a(z) = \sqrt{z(2(a + a^{-1}-1) + z+z^{-1})}\sqrt{z(2(a + a^{-1}+1) + z+z^{-1})}
\end{equation} where these square roots are the principal branch. Note that $2(a + a^{-1}-1) + z+z^{-1} - 2$ and $2(a + a^{-1}+1) + z+z^{-1}$ are both real and positive for $|z| = 1$. Then for $z \in \circlecontour_1$ we have \[
\phi_a(z) = z\sqrt{(2(a + a^{-1}) + z+z^{-1})^2 - 4}
\] and therefore \[w_0(z) = \frac{1}{2z}\left(\phi_a(z) - (2(a + a^{-1}) + (z+z^{-1}))\right)\] for $z\in \circlecontour_1$. We can extend $z\sqrt{(2(a + a^{-1}) + z+z^{-1})^2 - 4}$ to the disk $\{|z| \leq 1\}$ by $\phi_a(z)$, except where $\phi_a(z)$ has branch cuts. There are branch cuts when \[
z(2(a + a^{-1}-1) + z+z^{-1}) \in \mathbb{R}_{\leq 0} \text{ or } z(2(a + a^{-1}+1) + z+z^{-1}) \in \mathbb{R}_{\leq 0}.
\] If $z(2(a + a^{-1}\pm 1) + z+z^{-1}) \in \mathbb{R}$ then $\mathrm{Im}(z^2 + 2(a + a^{-1}\pm 1)z ) = 0$, so $(z - \overline{z})(z + \overline{z} + 2(a + a^{-1}\pm 1)) = 0$. Hence either $z \in \mathbb{R}$ or $\mathrm{Re}(z) = -(a + a^{-1}\pm 1) < -1$. So inside the unit disk the only branch cuts are on the real line. We find $z^2 + 2(a + a^{-1} \pm 1)z + 1 < 0$ when $-(a+a^{-1} \pm 1) - \sqrt{(a+a^{-1}\pm 1)^2 - 1} < z < -(a+a^{-1} \pm 1) + \sqrt{(a+a^{-1}\pm 1)^2 - 1}$. We have \begin{align*}-(a+a^{-1} + 1) - \sqrt{(a+a^{-1}+ 1)^2 - 1} &< -(a+a^{-1} -1) - \sqrt{(a+a^{-1}- 1)^2 - 1} \\
 &< -(a+a^{-1} - 1)+\sqrt{(a+a^{-1}- 1)^2 - 1} \\
 &< -(a+a^{-1} + 1) + \sqrt{(a+a^{-1}+ 1)^2 - 1}.\end{align*} For $-(a+a^{-1} - 1) - \sqrt{(a+a^{-1}- 1)^2 - 1} < z < -(a+a^{-1} - 1) + \sqrt{(a+a^{-1}- 1)^2 - 1}$ the branch cuts effectively cancel. For $z < -(a+a^{-1} - 1) - \sqrt{(a+a^{-1}- 1)^2 - 1}$ we do not have $|z| \leq 1$. So the only branch cut inside the contour is along the real axis between $z=z_1$ and $z=z_2$ with \begin{align*}\begin{split} z_1 &= -(a + a^{-1} - 1) + \sqrt{(a+a^{-1} -1)^2 - 1} \\
 z_2 &= -(a + a^{-1} + 1) + \sqrt{(a + a^{-1} + 1)^2 - 1}.\end{split}\end{align*} 
 
We have \[\mathbb{K}_{a,0,0}^{-1}(\mathbf{w}_0, \mathbf{b}_1 + 2ue_1 + 2ve_2)= \frac{1}{2\pi i} \int_{\circlecontour_1} \frac{(a+z)\,z^u\, (\frac{1}{2z}(\phi_a(z) - 2((a + a^{-1})z + z^2+1)))^v}{a\phi_a(z)}dz.\] The integrand is meromorphic in the unit disk except when $z_1 \leq z \leq z_2$. We can show that $(\frac{1}{2z}(\phi_a(z) - 2((a + a^{-1})z + z^2+1)))^v$ has a zero of order $v$ at $z=0$. Since $u+v \geq 0$ and $\phi_a(z)$ has no zeros in the unit disk except on the branch cut, the integrand has no poles outside of the branch cut, and so we can deform the contour to surround the branch cut. As $z$ tends to the branch cut in the upper half plane, we have \begin{align*}
 \phi_a(z) = &\sqrt{z(2(a + a^{-1}-1) + z+z^{-1})}\sqrt{z(2(a + a^{-1}+1) + z+z^{-1})}\\
 &= i\sqrt{z(2(a + a^{-1}-1) + z+z^{-1})}\sqrt{-z(2(a + a^{-1}+1) + z+z^{-1})} \\
  &= i\sqrt{2(a + a^{-1})z - 2z + z^2+1)}\sqrt{-2(a + a^{-1})z - 2z - z^2-1)} \\
 &= i \sqrt{4z^2 - (2(a+a^{-1})z + z^2 + 1)^2}
 \end{align*} while for $z$ tending to the branch cut in the lower half plane we have \[
  \phi_a(z) =  - i \sqrt{4z^2 - (2(a+a^{-1})z + z^2 + 1)^2}
 \] Let \[
\theta_a(z) = \frac{1}{2z} \left(i \sqrt{4z^2 - (2(a+a^{-1})z + z^2 + 1)^2} - (2(a + a^{-1})z + z^2+1)\right).
\]
So we obtain \begin{equation*}
\mathbb{K}_{a,0,0}^{-1}(\mathbf{w}_0, \mathbf{b}_1 + 2ue_1 + 2ve_2)= \frac{1}{2\pi a} \int_{z_1}^{z_2} \frac{(a+z)\,z^u\, (\theta_a(z)^v + \overline{\theta_a(z)}^v)}{ \sqrt{4z^2 - (2(a+a^{-1})z + z^2 + 1)^2}} dz
\end{equation*} as required. 
\end{proof} For the case where $-u + v > 0$, we have a similar result. 
\begin{theorem}\label{thm:kgasinverserealuneg}
 Let $\mathbf{w}_0\in \mathtt{W}_0$ and $\mathbf{b}_1\in \mathtt{B}_1$ be vertices in the same fundamental domain. Take $u,v\in\mathbb{Z}$ with $v\geq 0$ and $-u+v > 0$. For $a < 1$, let $z_1,z_2$ and $\theta_a(z)$ be as in Theorem \ref{thm:kgasinversereal}. Then we can write
\begin{equation}
\mathbb{K}_{a,0,0}^{-1}(\mathbf{w}_0, \mathbf{b}_1 + 2ue_1 + 2ve_2)= \frac{1}{2\pi a} \int_{z_1}^{z_2} \frac{(a+z^{-1})\,z^{-u}\, (\theta_a(z)^v + \overline{\theta_a(z)}^v)}{ \sqrt{4z^2 - (2(a+a^{-1})z + z^2 + 1)^2}} dz
\end{equation}
\end{theorem}
\begin{proof}
 As for the proof of Theorem \ref{thm:kgasinversereal}, we start with the integral formula from \ref{lemma:K11singlecontourintegral}. Then we apply the change of variables $z\rightarrow z^{-1}$. Noting that $w_0(z)$ and the denominator of the integrand are invariant under this change of variables, we obtain 
\[\mathbb{K}_{a,0,0}^{-1}(\mathbf{w}_0, \mathbf{b}_1 + 2ue_1 + 2ve_2)= \frac{1}{2\pi i} \int_{\circlecontour_1} \frac{dz}{z} \frac{(a+z^{-1})\,z^{-u}\,w_0(z)^v}{a\sqrt{(2(a + a^{-1}) + z+z^{-1})^2 - 4}}\] The proof is completed in the same way as for Theorem~\ref{thm:kgasinversereal}.
\end{proof}

\subsection{Asymptotics}\label{sec:Kgasinverseasym}

Now we find the asymptotics of these integrals as $a$ tends to 1 from below. First we find the asymptotics for $u = v = 0$. We will use this in the computation of the asymptotics for the other cases.

\begin{theorem}\label{thm:K11asymptotics}
 Let $\mathbf{w}_0\in \mathtt{W}_0$ and $\mathbf{b}_1\in \mathtt{B}_1$ be vertices in the same fundamental domain. Let $\mathbb{K}_{a,0,0}^{-1}(\mathbf{w}_0, \mathbf{b}_1)$ be as defined in Equation~\ref{eq:gastranslationinvariantoriginal}. Then 
 \[a \mathbb{K}_{a,0,0}^{-1}(\mathbf{w}_0, \mathbf{b}_1) = \frac{1}{4} - \frac{1}{2\pi}\bconst m^{-1/2}(- \log (\bconst m^{-1/2}) +2\log 2 ) + O(m^{-1}\log m). \]
\end{theorem}

\begin{proof}
From Theorem~\ref{thm:kgasinversereal} we have \begin{equation}\label{eq:kgasinverserealunsimplified}
\mathbb{K}_{a,0,0}^{-1}(\mathbf{w}_0, \mathbf{b}_1)= \frac{1}{\pi a} \int_{z_1}^{z_2} \frac{ a+z}{ \sqrt{4z^2 - (2(a+a^{-1})z + z^2 + 1)^2}} dz
\end{equation} Let \begin{equation}\label{eq:Sa}
S(a) = \int_{z_1}^{z_2} \frac{ a+z}{  \sqrt{4z^2 - (2(a+a^{-1})z + z^2 + 1)^2}} dz
\end{equation} so $a \mathbb{K}_{a,0,0}^{-1}(\mathbf{w}_0, \mathbf{b}_1) = S(a)/\pi$. First we make the substitution $z = (t-1)/(t+1)$ and set $t_i = (1+z_i)/(1-z_i)$ to obtain \[
S(a) = \int_{t_1}^{t_2} \left(a + \frac{t-1}{t+1}\right)\frac{1}{\sqrt{-(a + a^{-1} + 2)(a+a^{-1})(t^2 -t_2^2)(t^2 - t_1^2)}}dt.
\] We split this integral into two integrals $S(a) = S_1(a) + S_2(a)$ where \[
S_1(a) = \int_{t_1}^{t_2} \left( \frac{-2t}{t^2-1}\right)\frac{1}{\sqrt{-(a + a^{-1} + 2)(a+a^{-1})(t^2 -t_2^2)(t^2 - t_1^2)}}dt.
\] and
\[
S_2(a) = \int_{t_1}^{t_2} \left(a + \frac{t^2 + 1}{t^2 - 1}\right)\frac{1}{\sqrt{-(a + a^{-1} + 2)(a+a^{-1})(t^2 -t_2^2)(t^2 - t_1^2)}}dt.
\] We will deal with these separately. First we look at $S_1(a)$. We make a change of variables $u = t^2$ to obtain \[
S_1(a) = \int_{t_1^2}^{t_2^2} \left( \frac{1}{1-u}\right)\frac{1}{\sqrt{-(a + a^{-1} + 2)(a+a^{-1})(u -t_2^2)(u - t_1^2)}}du
\] and another change of variables $u = \frac{1}{2}(t_1^2 + t_2^2) + \frac{1}{2}(t_2^2 - t_1^2)\frac{2v}{v^2 + 1}$ to obtain \[
S_1(a) = \int_{-1}^{1} \frac{\sqrt{(a + a^{-1} +1)^2 - 1}}{(a + a^{-1} +1)(v^2 + 1) - 2v} dv.
\] Now we let $\delta = (a + a^{-1} + 1)^{-1}$ and set $y = v-\delta$.
Then we obtain \[
S_1(a) = \int_{-1-\delta}^{1-\delta}\frac{\sqrt{1-\delta^2}}{y^2 + 1 - \delta^2} dy = \frac{\pi}{2}.
\] For $S_2(a)$ we make the substitution $t = t_1t_2/\sqrt{t_2^2 - (t_2^2 - t_1^2)y^2}$ and note that we have $t_1^2 = (a+a^{-1}-2)/(a+a^{-1})$ and $t_2^2 = (a+a^{-1})/(a+a^{-1}+2)$ to obtain
 \[
S_2(a) = \int_0^1 \left(a-1 + \frac{a + a^{-1} - 2}{-1 + \frac{2}{a+a^{-1}} y^2}\right) \frac{1}{(a + a^{-1})\sqrt{\left(1 - \frac{4}{(a+a^{-1})^2}y^2\right)(1-y^2)}}dy.
\] Recall that $c = 1/(a+a^{-1})$. Let $k=2c$. We can write \[
S_2(a) = \frac{a-1}{a+a^{-1}} \int_0^1 \frac{1}{\sqrt{(1 - k^2 y^2)(1-y^2)}}dy 
- \frac{a + a^{-1} - 2}{a+a^{-1}}\int_0^{1}\frac{1}{(1-ky^2)\sqrt{(1-k^2 y^2)(1 - y^2)}} dy.
\] The first term is an elliptic integral. To deal with the second term we first compute \begin{align*}
2\int_0^1 &\frac{1}{(1-ky^2)\sqrt{(1-k^2 y^2)(1 - y^2)}} dy -  \int_0^1 \frac{1}{\sqrt{(1 - k^2 y^2)(1-y^2)}}dy \\
&\hspace{1cm}= \int_0^1 \left(\frac{1 +k y^2}{1 - k y^2}\right) \frac{1}{\sqrt{(1-k^2 y^2)(1 - y^2)}} \\
&\hspace{1cm}= \int_1^{\frac{1+k}{1-k}}u (u^2 -1)\frac{1}{\sqrt{(k+1)^2 - (k-1)^2 u^2}} du \\
& \hspace{1cm}= \int_1^{\left(\frac{1+k}{1-k}\right)^2} \frac{1}{2(1-k)}\frac{1}{\sqrt{\left(\left(\frac{1+k}{1-k}\right)^2 - v\right)(v-1)}} dv \\
&\hspace{1cm}= \frac{\pi}{2(1-k)}
\end{align*} where we use the substitutions $u = (1 + ky^2)/(1 - ky^2)$ and $v = u^2$. Therefore, noting that $1-k = (a+a^{-1}-2)/(a+a^{-1})$ we have \[
S_2(a) = \frac{a - a^{-1}}{2(a + a^{-1})} \int_0^1 \frac{1}{\sqrt{(1 - k^2 y^2)(1-y^2)}} dy - \frac{\pi}{4}
\] So recalling that $a \mathbb{K}_{a,0,0}^{-1}(\mathbf{w}_0, \mathbf{b}_1) = (S_1(a) + S_2(a))/\pi$ we have \begin{equation}
a \mathbb{K}_{a,0,0}^{-1}(\mathbf{w}_0, \mathbf{b}_1) = \frac{1}{4} + \frac{a - a^{-1}}{2\pi (a + a^{-1})} \int_0^1 \frac{1}{\sqrt{(1 - 4 c^2 y^2)(1-y^2)}} dy.\label{eq:kgasinversesimplified}
\end{equation} This integral is a complete elliptic integral of the first kind and its asymptotics are well known. We find that \[
 a \mathbb{K}_{a,0,0}^{-1}(\mathbf{w}_0, \mathbf{b}_1) = \frac{1}{4} - \frac{h}{2\pi}(- \log h +2\log 2 ) + O(h^2\log h). 
\] where $a = 1-h$. The result follows.
\end{proof}

We now make the following definitions. Let \begin{equation}\label{eq:Sxia}
S(a, z) = \int_{z_1}^{z} \frac{ a+\xi}{  \sqrt{4\xi^2 - (2(a+a^{-1})\xi + \xi^2 + 1)^2}} dz.
\end{equation} Note that $S(a, z_2) = S(a)$, defined in Equation~\ref{eq:Sa}, and $S(a, z_1) = 0$. Also note that $S(a,z)$ is continuous in $z$. For $z_1 \leq z \leq z_2$, let \begin{equation}
g^{(u,v)}(a,z)  = \frac{z^u(\theta_a(z)^{v}+ \theta_a(z)^{-v})}{2} = \frac{z^u(\theta_a(z)^{|v|}+ \overline{\theta_a(z)}^{|v|})}{2}\label{eq:guvaz}
\end{equation} where $\theta_a(z)$ is defined in Equation~\ref{eq:thetaa}. Note that $g^{(u,v)}(a,z)$ = $g^{(u,-v)}(a,z)$, and that $g^{(u,v)}(a,z)$ is real for $z_1 \leq z \leq z_2$. 

Below we prove a lemma that will allow us to find the asymptotics of $\mathbb{K}_{a,0,0}^{-1}(\mathbf{w}_0, \mathbf{b}_1 + 2ue_1 + 2ve_2)$ for general $u,v$, from the asymptotics of $g^{(u,v)}(a,z)$ and $S(a,z)$.

\begin{lemma}\label{lemma:intbyparts}
Let $\mathbf{w}_0\in \mathtt{W}_0$ and $\mathbf{b}_1\in \mathtt{B}_1$ be vertices in the same fundamental domain and take $u,v\in\mathbb{Z}$ with $u + |v| \geq 0$. Then we have 
\begin{equation}\label{eq:K11intbypart}
\mathbb{K}_{a,0,0}^{-1}(\mathbf{w}_0, \mathbf{b}_1 + 2ue_1 + 2ve_2) = \frac{1}{\pi a} \left(g^{(u,v)}(a, z_2) S(a, z_2) - \int_{z_1}^{z_2} g^{(u,v)}_2(a,z) S(a, z)dz\right)
\end{equation} where $g^{(u,v)}_2(a,z) = \frac{\partial g^{(u,v)}}{\partial z} (a,z)$.
\end{lemma}

\begin{proof}
Let \[
r(a,z) =\frac{ a+z}{  \sqrt{4z^2 - (2(a+a^{-1})z + z^2 + 1)^2}}\] so $S(a,z) = \int_{z_1}^{z_2} r(a,z')dz'$. Then from Theorem~\ref{thm:kgasinversereal} and Lemma~\ref{lemma:K11invequalities} we have \[
\mathbb{K}_{a,0,0}^{-1}(\mathbf{w}_0, \mathbf{b}_1 + 2ue_1 + 2ve_2) = \frac{1}{\pi a} \int_{z_1}^{z_2} g^{(u,v)}(a,z) r(a,z) dz.
\] We know that $g^{(u,v)}(a,z)$ is a polynomial in $z$ for $u - |v|\geq 0$ and $z_1 \leq z \leq z_2$. Then integration by parts gives \begin{align*}
\int_{z_1}^{z_2} g^{(u,v)}(a,z) r(a,z) dz &= g^{(u,v)}(a,z_2) S(a, z_2) - g^{(u,v)}(a,z_1) S(a, z_1) - \int_{z_1}^{z_2} g^{(u,v)}_2(a,z) \int_{z_1}^z r(a,z') dz' \, dz \\
&= g^{(u,v)}(a,z_2) S(a, z_2) - \int_{z_1}^{z_2} g^{(u,v)}_2(a,z) S(a, z) \, dz
\end{align*} as required.
\end{proof}

We are almost ready to complete our asymptotic analysis. First we state a lemma about the asymptotics of a particular elliptic integral which will appear in the computation of the asymptotics of $S(a,z)$.

\begin{lemma}\label{lemma:Sazellipticasymp}
For $-1 < z < -3+2\sqrt{2}$, let \[
\lambda_z = \sqrt{\frac{a+a^{-1}}{2}}\frac{\sqrt{2z(a+a^{-1}-1) + z^2 + 1}}{1+z}\] for $a$ sufficiently close to 1 such that the argument of the square root is positive. Let $k = 2c = 2/(a+a^{-1})$. Then
\[
\int_0^{\lambda_z} \frac{1}{\sqrt{(1 - k^2 y^2)(1-y^2)}} dy = -\log h + \log\left(\frac{4\sqrt{2}(1+z)}{1-z + \sqrt{-1 -6z - z^2}}\right)  + R(h,z) 
\] as $h \rightarrow 0$, where $a = 1 - h$ and \[|R(h,z)| < A (z+1)^{-2}h^2\log h + A'(z-(-3+2\sqrt{2}))^{-1/2}h^2 + A'' h\log h
\] as $h\to 0$ for some constants $A,A'$ and $A''$ not depending on $z$ or $h$.

\end{lemma}

The proof can be found in Appendix~\ref{sec:Sazellipticasymp}.

\begin{theorem}\label{thm:Kgasinverseasym1}
Let $\mathbf{w}_0\in \mathtt{W}_0$ and $\mathbf{b}_1\in \mathtt{B}_1$ be vertices in the same fundamental domain and take $u,v\in\mathbb{Z}$ with $u + |v| \geq 0$. Then
 \begin{multline}\label{eq:Kgasinverseasym1}
a\mathbb{K}_{a,0,0}^{-1}(\mathbf{w}_0, \mathbf{b}_1 + 2ue_1 + 2ve_2) = \frac{1}{\pi}\int_{-1}^{-3 + 2\sqrt{2}} \frac{g^{(u,v)}(1,z)}{\sqrt{-1-6z-z^2}}dz \\
+ \frac{(-1)^{u + v}}{2\pi} h \log h - \frac{1}{\pi}\left((-3+2\sqrt{2})^u \log 2 + \int_{-1}^{-3 + 2\sqrt{2}} g^{(u,v)}_2(1,z) b(z)dz \right)h + O(h^2 \log h)
\end{multline} where $g^{(u,v)}(a,z)$ is defined in Equation~\ref{eq:guvaz} and \begin{equation}\label{eq:bintegrandpart}b(z) = -\frac{1}{2}\log\left(\frac{4\sqrt{2}(1+z)}{1-z + \sqrt{-1 -6z - z^2}}\right)\end{equation}

\end{theorem}

\begin{proof}
Let $a = 1-h$. We start from Equation~\ref{eq:K11intbypart} and ompute the asymptotics of \[
g^{(u,v)}(a, z_2) S(a, z_2) - \int_{z_1}^{z_2} g^{(u,v)}_2(a,z) S(a, z)dz.\] First we compute \begin{align*}
z_1 &= -1 + \sqrt{2} h + O(h^2) \\
z_2 &= -3 + 2\sqrt{2} + O(h^2)
\end{align*} and $z_2 > -3 + 2\sqrt{2}$ for $0<a<1$. First we note that since from the proof of Theorem~\ref{thm:kgasinversereal} $z_2$ is a root of $2(a+a^{-1}+1)z + z^2 + 1 = 0$, we have $\theta_a(z_2) = \frac{1}{2z_2}(2z_2) = 1$, so \[g^{(u,v)}(a,z_2) =  z_2^u = (-3+2\sqrt{2})^u + O(h^2).\] So from Theorem \ref{thm:K11asymptotics} we have \begin{equation}
g^{(u,v)}(a,z_2) S(a,z_2) =  (-3+2\sqrt{2})^u\left(\frac{\pi}{4} - \frac{h}{2}(- \log h +2\log 2 )\right) + O(h^2\log h).\label{eq:Szaaymp1}
\end{equation}

For the integral term in \ref{eq:K11intbypart}, note that $g^{(u,v)}_2(a,z)$ is a polynomial in $z,z^{-1}, a$ and $a^{-1}$, and $S(a,z)$ is continuous on $[z_1,z_2]$ with $S(a,z_1) = 0$, so the integrand is continuous and evaluates to zero at the lower limit. Also note that the integrand $r(a,z)$ is positive on $(z_1,z_2)$, so $S(a,z) \leq S(a,z_2) \leq \pi/4$ for $h$ sufficiently small. Also $g^{(u,v)}(a,z)$ is bounded in $a$ for $a$ in a closed interval not containing 0 and $z$ in a closed interval not containing 0. So, since $-3 + 2\sqrt{2} < z_2$ and $z_2 = -3 + 2\sqrt{2} + O(h^2)$, we have \[
\left|\int_{-3 + 2\sqrt{2}}^{z_2} g^{(u,v)}_2(a,z) S(a, z)dz\right| = O(h^2)
\]  and so \begin{equation} \int_{z_1}^{z_2} g^{(u,v)}_2(a,z) S(a, z)dz = \int_{z_1}^{-3 + 2\sqrt{2}} g^{(u,v)}_2(a,z) S(a, z)dz + O(h^2)\label{eq:upperz2approx}\end{equation}
Since $g^{(u,v)}(a,z)$ is a function of $a+a^{-1} = 1 + O(h^2)$, we see that $g^{(u,v)}_2(a,z) = g^{(u,v)}_2(1,z) + O(h^2)$. We must find the asymptotics of $S(a,z)$ as $h \rightarrow 0$. To do this we follow the same procedure as in the proof of Theorem~\ref{thm:K11asymptotics}, but keeping track of the upper limits in each integral. We use Lemma \ref{lemma:Sazellipticasymp} to deal with the elliptic integral. We omit the details. We find that for $z_1 \leq z < -3 + 2\sqrt{2}$, \begin{align*}
S(a,z) &= \frac{1}{2}\left(-2\tan^{-1}\left(\frac{4\sqrt{2}z + 3(1+z)\sqrt{-1 -6z - z^2}}{(3+z)(1+3z)}\right) + \sin^{-1}\frac{(1+z)^2}{4z} + 2\tan^{-1}\sqrt{2} \right) \\
&\hspace{1cm}+ \frac{h}{2}\left(\log h - \log\frac{4\sqrt{2}(1+z)}{1-z+ \sqrt{-1 -6z - z^2}}\right) + T(h,z)
\end{align*} where \[|T(h,z)| < C(z+1)^{-2}h^3\log h + C'(z-(-3+2\sqrt{2}))^{-1/2}h^2 + C'' h^2\log h
\] as $h\to 0$ for constants $C,C',C''$. We write \[
S(a,z) = b_0(z) + \frac{ h\log h}{2} + b_2(z) h + T(h,z).
\]
Combining this asymptotic expansion with Equation \ref{eq:upperz2approx}, and integrating the parts of $T(h,z)$ that depend on $z$, noting that $z_1+1 = O(h)$, we have \[
\int_{z_1}^{z_2} g^{(u,v)}_2(a,z) S(a, z)dz = \int_{z_1}^{-3 + 2\sqrt{2}} g^{(u,v)}_2(1,z) \left(b_0(z) + \frac{ h\log h}{2} + b_2(z) h\right)dz + O(h^2\log h)
\] We look at these terms one-by-one. Clearly since $g^{(u,v)}_2(1,z)$ is a polynomial in $z$ and $z^{-1}$, it is bounded with bounded derivative on $[-1, -3 + 2\sqrt{2}]$. Firstly, we can show that $b_0(z)$ is bounded on $[-1, -3 + 2\sqrt{2}]$ with bounded derivative on $[-1, z_1]$. Furthermore $b_0(-1) = 0$. Hence \[
\int_{-1}^{z_1} g^{(u,v)}_2(1,z)b_0(z)dz = O(h^2)\] so \[
\int_{z_1}^{-3 + 2\sqrt{2}} g^{(u,v)}_2(1,z)b_0(z)dz = \int_{-1}^{-3 + 2\sqrt{2}} g^{(u,v)}_2(1,z)b_0(z)dz + O(h^2).\]

For the order $h\log h$ term, we have  \begin{align*}
\int_{z_1}^{-3 + 2\sqrt{2}}  \frac{ g^{(u,v)}_2(1,z) h\log h}{2} dz &= (g^{(u,v)}(1,-3 + 2\sqrt{2}) - g^{(u,v)}(1,z_1))\frac{ h\log h}{2} \\
&= ((-3 + 2\sqrt{2})^u - (-1)^{u + v} )\frac{ h\log h}{2} + O(h^2 \log h)
\end{align*} noting that $g^{(u,v)}(1,z_1) = z_1^u (-1)^{v}$.

For the the order $h$ term, we calculate that $b_2(z) = O(\log(1+z))$ as $z\rightarrow 1$. From this we can show that \[
\int_{-1}^{z_1} g^{(u,v)}_2(1,z)b_2(z)\, dz = O(h\log h).\]

Putting these together, we have \begin{multline}
\int_{z_1}^{z_2} g^{(u,v)}_2(a,z) S(a, z)dz = \int_{-1}^{-3 + 2\sqrt{2}} g^{(u,v)}_2(1,z) \left(b_0(z) + b_2(z) h\right)dz \\
+ (g^{(u,v)}(1,z_2) - (-1)^{u + v} )\frac{ h\log h}{2} + O(h^2\log h)\label{eq:Szaasym2}
\end{multline} We can now simplify this further by using integration by parts again on the constant term. We have \[
\frac{db_0}{dz}(z) = \frac{1}{\sqrt{-1-6z-z^2}} = r(1,z)\] Then \begin{align}\label{eq:constintbyparts}
\begin{split}
\int_{-1}^{-3 + 2\sqrt{2}} g^{(u,v)}_2(1,z)b_0(z) dz &= g^{(u,v)}(1,-3+2\sqrt{2})b_0(-3+2\sqrt{2}) - g^{(u,v)}(1,-1)b_0(-1)\\
&\hspace{2cm} - \int_{-1}^{-3 + 2\sqrt{2}} \frac{g^{(u,v)}(1,z)}{\sqrt{-1-6z-z^2}}dz \\
&= \frac{(-3+2\sqrt{2})^u \pi}{4} - \int_{-1}^{-3 + 2\sqrt{2}} \frac{g^{(u,v)}(1,z)}{\sqrt{-1-6z-z^2}}dz
\end{split}
\end{align}
We are unable to simplify the order $h$ term in the same way, as $b_2(z)$ is undefined at $z = -1$.
Combining Equations \ref{eq:Szaaymp1}, \ref{eq:Szaasym2} and \ref{eq:constintbyparts} we have
 \begin{multline}
g^{(u,v)}(a, z_2) S(a, z_2) - \int_{z_1}^{z_2} g^{(u,v)}_2(a,z) S(a, z)dz = \int_{-1}^{-3 + 2\sqrt{2}} \frac{g^{(u,v)}(1,z)}{\sqrt{-1-6z-z^2}}dz \\
+ \frac{(-1)^{u + v}}{2} h \log h - \left((-3+2\sqrt{2})^u \log 2 + \int_{-1}^{-3 + 2\sqrt{2}} g^{(u,v)}_2(1,z) b_2(z)dz \right)h + O(h^2 \log h).
\end{multline} By Equation~\ref{eq:K11intbypart} to find $a\mathbb{K}_{a,0,0}^{-1}(\mathbf{w}_0, \mathbf{b}_1 + 2ue_1 + 2ve_2)$ we just divide this by $\pi$, then noting that $b(z)=b_2(z)$, the proof is complete.
\end{proof}

We can prove a similar result for when $-u + |v| > 0$.

\begin{theorem}\label{thm:Kgasinverseasym2}
Let $\mathbf{w}_0\in \mathtt{W}_0$ and $\mathbf{b}_1\in \mathtt{B}_1$ be vertices in the same fundamental domain and take $u,v\in\mathbb{Z}$ with $-u + |v| \geq 0$. Then
 \begin{multline}\label{eq:Kgasinverseasym2}
\mathbb{K}_{a,0,0}^{-1}(\mathbf{w}_0, \mathbf{b}_1 + 2ue_1 + 2ve_2) = \frac{1}{\pi}\int_{-1}^{-3 + 2\sqrt{2}} \frac{g^{(-u-1,v)}(1,z)}{\sqrt{-1-6z-z^2}}dz \\
+ \frac{(-1)^{u + v}}{2\pi} h \log h + \frac{1}{\pi}\left((-3+2\sqrt{2})^{-u-1} \log 2 + \int_{-1}^{-3 + 2\sqrt{2}} g^{(-u-1,v)}_2(1,z) b(z)dz \right)h + O(h^2 \log h)
\end{multline} where $g^{(u,v)}(a,z)$ is defined in Equation~\ref{eq:guvaz} and $b(z)$ is as in Equation~\ref{eq:bintegrandpart}.

\end{theorem}
\begin{proof}
From Theorem~\ref{thm:kgasinverserealuneg} we can write \begin{equation}
\mathbb{K}_{a,0,0}^{-1}(\mathbf{w}_0, \mathbf{b}_1 + 2ue_1 + 2ve_2)= \frac{1}{2\pi} \int_{z_1}^{z_2} \frac{(z+a^{-1})\,z^{-u-1}\, (\theta_a(z)^v + \overline{\theta_a(z)}^v)}{ \sqrt{4z^2 - (2(a+a^{-1})z + z^2 + 1)^2}} dz
\end{equation}
In the same way as Lemma~\ref{lemma:intbyparts} we can write \begin{equation}
\mathbb{K}_{a,0,0}^{-1}(\mathbf{w}_0, \mathbf{b}_1 + 2ue_1 + 2ve_2) = \frac{1}{\pi} \left(g^{(-u-1,v)}(a, z_2) S(a^{-1}, z_2) - \int_{z_1}^{z_2} g^{(-u-1,v)}_2(a,z) S(a^{-1}, z)dz\right)
\end{equation}
We find that breaking down $S(a^{-1},z)$ as in Theorem~\ref{thm:K11asymptotics}, the elliptic integral term has the opposite sign from $S(a,z)$, which is responsible for the order $h$ and $h\log h$ terms in the asymptotic expansion. So we have
\begin{equation*}
S(a^{-1},z_2) = \frac{\pi}{4} + \frac{h}{2}(- \log h +2\log 2 ) + O(h^2\log h) 
\end{equation*} and \begin{equation*}
S(a^{-1},z) = b_0(z) - \frac{ h\log h}{2} - b_2(z) h + O(h^2\log h).
\end{equation*} 
The theorem follows from Theorem~\ref{thm:Kgasinverseasym1} by changing the signs of the order $h$ and $h\log h$ terms and replacing $u$ with $-u-1$.
\end{proof}

Now we prove some relations between the coefficients in these asymptotic expansions. We write \begin{equation}\label{eq:Kinvsdef}
\mathbb{K}_{a,0,0}^{-1}(\mathbf{w}_0, \mathbf{b}_1 + 2ue_1 + 2ve_2) = s_0(u,v) + s_1(u,v) h \log h + s_2(u,v) h + O(h^2\log h).
\end{equation}
\begin{lemma}\label{lemma:asymcoeffsrelations}
Let $s_0(u,v)$, $s_1(u,v)$ and $s_2(u,v)$ be as in Equation~\ref{eq:Kinvsdef}. Then for all $u,v\in \mathbb{Z}$ we have:
\begin{enumerate}[noitemsep]
\item $s_i(u,-v) = s_i(u,v)$ for $i=0,1,2$.
\item $s_0(-u-1,v) = s_0(u,v)$ 
\item $s_1(-u-1,v) = -s_1(u,v)$ 
\item $s_2(-u-1,v) = -s_2(u,v) + s_0(u,v)$
\end{enumerate}
Furthermore, we have \[
s_2(u+1, v) = -s_2(u,v) + s_0(u+1,v). \]
\end{lemma}

\begin{proof}
The first relation is clear from Lemma~\ref{lemma:K11invequalities}, where we show that $\mathbb{K}_{a,0,0}^{-1}(\mathbf{w}_0, \mathbf{b}_1 + 2ue_1 + 2ve_2) = \mathbb{K}_{a,0,0}^{-1}(\mathbf{w}_0, \mathbf{b}_1 + 2ue_1 - 2ve_2)$. From Theorems~\ref{thm:Kgasinverseasym1} and \ref{thm:Kgasinverseasym2}, we have \begin{equation}\label{eq:s0integral}
s_0(u,v) = \begin{cases} \frac{1}{\pi}\int_{-1}^{-3 + 2\sqrt{2}} \frac{g^{(u,v)}(1,z)}{\sqrt{-1-6z-z^2}}dz \hspace{0.5cm}\text{ if } u\geq 0 \\
 \frac{1}{\pi}\int_{-1}^{-3 + 2\sqrt{2}} \frac{g^{(-u-1,v)}(1,z)}{\sqrt{-1-6z-z^2}}dz \hspace{0.5cm}\text{ if } u <  0 ,\end{cases}
\end{equation}
\begin{equation}\label{eq:s1integral}
s_1(u,v) = \frac{(-1)^{u + v}}{2\pi},\end{equation} and
\[s_2(u,v) = \begin{cases}  -\frac{1}{\pi}\left((-3+2\sqrt{2})^u \log 2 + \int_{-1}^{-3 + 2\sqrt{2}} g^{(u,v)}_2(1,z) b(z)dz - \int_{-1}^{-3 + 2\sqrt{2}} \frac{g^{(u,v)}(1,z)}{\sqrt{-1-6z-z^2}}dz  \right)  \hspace{0.1cm}\text{ if } u\geq 0 \\
 \frac{1}{\pi}\left((-3+2\sqrt{2})^{-u-1} \log 2 + \int_{-1}^{-3 + 2\sqrt{2}} g^{(-u-1,v)}_2(1,z) b(z)dz \right) \hspace{0.1cm}\text{ if } u <  0 ,\end{cases}
\] noting that there is a factor of $a$ on the left hand side of Equation~\ref{eq:Kgasinverseasym1} but not Equation~\ref{eq:Kgasinverseasym2}. Relations 2--4 in the lemma are clear from these formulas.

For the final relation, we note that $g^{(u+1,v)}_2(1,z) = \frac{d}{dz}(z g^{(u,v)}(1,z)) = g^{(u,v)}(1,z) + z g_2^{(u,v)}(1,z)$ and \[
\frac{d}{dz}(1+z)b(z) = \frac{-1}{\sqrt{-1-6z-z^2}} +b(z).\] We compute \begin{align*}
 \int_{-1}^{-3 + 2\sqrt{2}} g^{(u+1,v)}_2&(1,z) b(z)dz +  \int_{-1}^{-3 + 2\sqrt{2}} g^{(u,v)}_2(1,z) b(z)dz \\
 &=  \int_{-1}^{-3 + 2\sqrt{2}} (1+z) g^{(u,v)}_2(1,z) b(z)dz +  \int_{-1}^{-3 + 2\sqrt{2}}  g^{(u,v)}(1,z) b(z)dz \\
 &= \left[(1+z)b(z) g^{(u,v)}(1,z)\right]^{-3+2\sqrt{2}}_{-1} + \int_{-1}^{-3 + 2\sqrt{2}} \left(\frac{-1}{\sqrt{-1-6z-z^2}} +b(z)\right)  g^{(u,v)}(1,z)dz \\
 &\hspace{1cm} +  \int_{-1}^{-3 + 2\sqrt{2}}  g^{(u,v)}(1,z) b(z)dz\\
 &= -(-2+2\sqrt{2})(-3+2\sqrt{2})^u\log 2 - \int_{-1}^{-3 + 2\sqrt{2}} \frac{g^{(u,v)}(1,z)}{\sqrt{-1-6z-z^2}}dz.
\end{align*} Also we have \[
(-3+2\sqrt{2})^{u+1} \log 2 + (-3+2\sqrt{2})^u \log 2 = (-2+2\sqrt{2})(-3+2\sqrt{2})^u\log 2.
\] Putting these together we see that for $u \geq 0$, we have \[
s_2(u+1,v) + s_2(u,v) = \frac{1}{\pi} \int_{-1}^{-3 + 2\sqrt{2}} \frac{g^{(u+1,v)}(1,z)}{\sqrt{-1-6z-z^2}} = s_0(u+1,v)\] as required. For $u \leq -2$, we have \begin{align*}
s_2(u+1,v) &= -s_2(-u-2,v)+s_0(u+1,v) \\
&= s_2(-u-1, v) - s_0(-u-1,v) + s_0(u+1,v) \\
&= -s_2(u,v) + s_0(u,v) - s_0(-u-1,v) + s_0(u+1,v) \\
&= -s_2(u,v) + s_0(u+1,v) 
\end{align*} Now all that remains is the $u=-1$ case. Relation 4 followed by 2 gives $s_2(0, v) = -s_2(-1,v) + s_0(0,v)$ as required. This proves the lemma.

\end{proof} Now we prove some corollaries. 

\begin{corollary}\label{cor:s2uminus2}
Let $s_2(u,v)$ be as defined in Equation~\ref{eq:Kinvsdef}. Then for all $u,v\in \mathbb{Z}$ we have \[
s_2(u,v) = s_2(-u-2, v).\]
\end{corollary}
\begin{proof}
From the last relation in Lemma~\ref{lemma:asymcoeffsrelations} we have $s_2(u+1,v) = -s_2(u,v) + s_0(u+1,v)$ and from relation 4 we have $s_2(u+1,v) = -s_2(-u-2, v) + s_0(u+1,v)$. Combining these proves the corollary.
\end{proof}
\begin{corollary}\label{cor:s2sumformula}
Let $s_0(u,v)$ and $s_2(u,v)$ be as defined inEquation~\ref{eq:Kinvsdef}. For $u \neq 0$, we can write \[s_2(u,v) = (-1)^u\left(\sum_{i=1}^{|u+1|-1}(-1)^{i}s_0(i,v) + s_2(0,v)\right)\] 
\end{corollary}
\begin{proof}
For $u > 0$, we use the relation $s_2(u, v) = -s_2(u-1,v) + s_0(u,v)$ repeatedly to obtain
\[s_2(u,v) = (-1)^u\left(\sum_{i=1}^{u}(-1)^{i}s_0(i,v) + s_2(0,v)\right).\] The $u < 0$ case follows from Corollary~\ref{cor:s2uminus2}.
\end{proof}

We can use this to write $s_2$ in a more symmetric form. 
\begin{corollary}\label{cor:s2otherform}
Let $s_2(u,v)$ be as defined in Equation~\ref{eq:Kinvsdef}. For $u, v \in \mathbb{Z}$, we can write \begin{multline}s_2(u,v) = \frac{(-1)^u}{\pi}\Bigg(\int_{-1}^{-3 + 2\sqrt{2}} \left(\sum_{i=0}^{|u+1|-1}(-z)^{i}\right) \frac{(\theta_1(z)^{v}+ \theta_1(z)^{-v})}{2\sqrt{-1-6z-z^2}}dz \\
 -\log 2 - \int_{-1}^{-3 + 2\sqrt{2}} g^{(0,v)}_2(1,z) b(z)dz\Bigg)\end{multline} where \[
 \theta_1(z) = \frac{1}{2z} \left(i \sqrt{4z^2 - (4z + z^2 + 1)^2} - (4z + z^2+1)\right),
 \]\[
 g^{(0,v)}_2(1,z) = \frac{v(\theta_1(z)^v - \theta_1(z)^{-v})\theta_1'(z)}{\theta_1(z)}
 \] and $b(z)$ is defined in Equation~\ref{eq:bintegrandpart}.
\end{corollary}

\begin{proof}
We start from Corollary~\ref{cor:s2sumformula}, and substitute the formulas for $s_0(i,v)$ and $s_2(0,v)$ stated in the proof of Lemma~\ref{lemma:asymcoeffsrelations}, and the definition for $g^{(u,v)}(1,z)$ from Equation~\ref{eq:guvaz}, to obtain the formula in the statement of the corollary. It is clear that this formula also holds for $u=0$.
\end{proof}

Now by Lemma~\ref{lemma:K11invequalities}, we can find the asymptotics of $\mathbb{K}_{a,0,0}^{-1}(\mathbf{w}_\e1, \mathbf{b}_\e2 + 2ue_1 + 2ve_2)$ for $\e1,\e2\in \{0,1\}$ and $u,z\in \mathbb{Z}$. In the following lemma we collect these formulas together for $\mathbb{K}_{a,0,0}^{-1}(x,y)$ where $x \in \mathtt{W}$ and $y\in \mathtt{B}$.

\begin{lemma}\label{lemma:Kgasinverseasymall}
Take $x \in \mathtt{W}$ and $y\in \mathtt{B}$. Write $y - x = pe_1 + qe_2$. Let $\zeta(x,y)$ and $\Sigma(x,y)$ be as defined in Equations~\ref{eq:zeta} and \ref{eq:Sigmaxy} respectively. When $p$ is odd and $q$ is even, define 
\[c_0(pe_1 + qe_2) = s_0((p-1)/2,q/2),\]
\[c_1(pe_1 + qe_2) = (-1)^{(p-q-1)/2}s_1((p-1)/2,q/2),\] and 
\[c_2(pe_1 + qe_2) = (-1)^{(p-q-1)/2}\left(s_2 ((p-1)/2,q/2)  - \frac{1}{2}s_0((p-1)/2,q/2)\right).\] When $p$ is even and $q$ is odd, define \[c_i(pe_1 + qe_2) = c_i(qe_1 + pe_2)\] for $i=0,1,2$. Then \begin{multline*}
\mathbb{K}_{a,0,0}^{-1}(x,y) 
= \frac{1}{\Sigma(x,y)}\Big(c_0(y-x) + c_0(y-x) h/2 \\
+ \zeta(x,y)(c_1(y-x) h \log h  + c_2(y-x) h )\Big) + O(h^2\log h).
 \end{multline*}
\end{lemma}

\begin{proof}
First we note that for $\mathbf{w}_\e1\in \mathtt{W}_\e1$ and $\mathbf{b}_\e2\in \mathtt{B}_\e2$ in one fundamental domain, we have \[
\mathbf{b}_1-\mathbf{w}_0 = e_1,\, \mathbf{b}_0-\mathbf{w}_1 = -e_1,\, \mathbf{b}_0-\mathbf{w}_0 = e_2\,\text{ and }\, \mathbf{b}_1-\mathbf{w}_1 = -e_2.\] 

Now suppose that $p$ is odd and $q$ is even, so $\Sigma(x,y) = 1$. Then either $x\in\mathtt{W}_0$ and $y\in\mathtt{B}_1$ or $x\in\mathtt{W}_1$ and $y\in\mathtt{B}_0$. In the first case we have \begin{align*}
\mathbb{K}_{a,0,0}^{-1}(x,y) &= \mathbb{K}_{a,0,0}^{-1}(\mathbf{w}_0, \mathbf{b}_1 + (p-1)e_1 + qe_2) \\
&= s_0((p-1)/2,q/2) + s_1((p-1)/2,q/2) h \log h + s_2 ((p-1)/2,q/2) h + O(h^2\log h).
\end{align*} In the second case we have \begin{align*}
\mathbb{K}_{a,0,0}^{-1}(x,y) &= \mathbb{K}_{a,0,0}^{-1}(\mathbf{w}_1, \mathbf{b}_0 + (p+1)e_1 + qe_2) \\
&= \mathbb{K}_{a,0,0}^{-1}(\mathbf{w}_0, \mathbf{b}_1 - (p+1)e_1 + qe_2) \\
&= s_0(-(p+1)2,q/2) + s_1(-(p+1)/2,q/2) h \log h + s_2 (-(p+1)/2,q/2) h + O(h^2\log h)\\
&= s_0((p-1)/2,q/2) - s_1((p-1)/2,q/2) h \log h \\
&\hspace{1cm}- s_2 ((p-1)/2,q/2) h + s_0((p-1)/2,q/2) h + O(h^2\log h)
\end{align*} where we use Lemma~\ref{lemma:K11invequalities} in the second line and Lemma~\ref{lemma:asymcoeffsrelations} in the fourth line. Now note that for $x\in\mathtt{W}_0$ and $y\in\mathtt{B}_1$ we have $\zeta(x,y) = (-1)^{(p-q-1)/2}$ and for $x\in\mathtt{W}_1$ and $y\in\mathtt{B}_0$ we have $\zeta(x,y) = (-1)^{(p-q-1)/2}$. So we can write \begin{multline*}
\mathbb{K}_{a,0,0}^{-1}(x,y) 
= s_0((p-1)/2,q/2) + \frac{1}{2} s_0((p-1)/2,q/2) h\\
 + \zeta(x,y)(-1)^{(p-q-1)/2} \bigg(s_1((p-1)/2,q/2) h \log h \\
 + (s_2 ((p-1)/2,q/2)  - \frac{1}{2}s_0((p-1)/2,q/2)) h \bigg) + O(h^2\log h).
\end{multline*}

Now suppose that $p$ is even and $q$ is odd, so $\Sigma(x,y) = i$. Then either $x\in\mathtt{W}_0$ and $y\in\mathtt{B}_0$ or $x\in\mathtt{W}_1$ and $y\in\mathtt{B}_1$. In the first case we have \begin{align*}
\mathbb{K}_{a,0,0}^{-1}(x,y) &= \mathbb{K}_{a,0,0}^{-1}(\mathbf{w}_0, \mathbf{b}_0 + pe_1 + (q-1)e_2) \\
&= -i\mathbb{K}_{a,0,0}^{-1}(\mathbf{w}_0, \mathbf{b}_1 + (q-1)e_1 + pe_2).
\end{align*} In the second case we have \begin{align*}
\mathbb{K}_{a,0,0}^{-1}(x,y) &= \mathbb{K}_{a,0,0}^{-1}(\mathbf{w}_1, \mathbf{b}_1 + pe_1 + (q+1)e_2) \\
&= -i\mathbb{K}_{a,0,0}^{-1}(\mathbf{w}_1, \mathbf{b}_0 + (q+1)e_1 + pe_2).
\end{align*} In the first case we have $\zeta(x,y) = (-1)^{(p-q+1)/2} = (-1)^{(q-p-1)/2}$ and in the second case we have $\zeta(x,y) = (-1)^{(p-q-1)/2} = (-1)^{(q-p-1)/2}$. So we can write \begin{multline*}
\mathbb{K}_{a,0,0}^{-1}(x,y) 
= -i\Bigg( s_0((q-1)/2,p/2) + \frac{1}{2} s_0((q-1)/2,p/2) h\\
 + \zeta(x,y)(-1)^{(q-p-1)/2} \bigg(s_1((q-1)/2,p/2) h \log h \\
 + (s_2 ((q-1)/2,p/2)  - \frac{1}{2}s_0((q-1)/2,p/2)) h \bigg)\Bigg) + O(h^2\log h).
\end{multline*}
Comparing with the definitions of $c_i(pe_1,q_e2)$ we see that we have proved the result.

\end{proof}

Now we are ready to prove Theorem~\ref{thm:Kgasinversefullasymptotics}.

\begin{proof}[Proof of Theorem \ref{thm:Kgasinversefullasymptotics}]
We need to show that the definitions of $c_i(pe_1 + qe_2)$ in Lemma~\ref{lemma:Kgasinverseasymall} agree with those in Equations~\ref{eq:c0}--\ref{eq:c2}. We start from Lemma~\ref{lemma:Kgasinverseasymall}. Firstly, note that $g_{(u,v)}(1,z) = z^u k^{(v)}$. Consider $p$ odd and $q$ even. The case $p$ even and $q$ odd follows immediately from this case. 

For $c_0$, by Equation~\ref{eq:s0integral} we have \[c_0(pe_1+qe_2) = \frac{1}{\pi}\int_{-1}^{-3 + 2\sqrt{2}} \frac{g^{(|p|/2-1/2,v)}(1,z)}{\sqrt{-1-6z-z^2}}dz\] from which Equation~\ref{eq:c0} follows.

For $c_1$, by Equation~\ref{eq:s1integral} we have  
\[c_1(pe_1 + qe_2) = (-1)^{(p-q-1)/2}\frac{(-1)^{(p-1)/2+q/2}}{2\pi} = \frac{1}{2\pi}\]

For $c_2$, it is a little more complicated. By Corollary~\ref{cor:s2otherform} we have \begin{align*}c_2(pe_1 + qe_2) &= (-1)^{(p-q-1)/2}\Bigg(\frac{(-1)^{(p-1)/2}}{\pi}\Bigg(\int_{-1}^{-3 + 2\sqrt{2}} \left(\sum_{i=0}^{|p+1|/2-1}(-z)^{i}\right) \frac{(\theta_1(z)^{q/2}+ \theta_1(z)^{-q/2})}{2\sqrt{-1-6z-z^2}}dz \\
 &-\log 2 - \int_{-1}^{-3 + 2 \sqrt{2}} g^{(0,q/2)}_2(1,z) b(z)dz\Bigg)  - \frac{1}{2\pi}\int_{-1}^{-3 + 2\sqrt{2}} \frac{g^{(|(p|/2-1/2,v)}(1,z)}{\sqrt{-1-6z-z^2}}dz\Bigg) \\
  &= \frac{(-1)^{q/2}}{\pi}\Bigg(\int_{-1}^{-3 + 2\sqrt{2}} \left(\sum_{i=0}^{|p+1|/2-1}(-z)^{i}\right) \frac{k^{(v)}(z)}{\sqrt{-1-6z-z^2}}dz -\log 2 \\
 & - \int_{-1}^{-3 + 2 \sqrt{2}} g^{(0,q/2)}_2(1,z) b(z)dz  - \frac{1}{2}\int_{-1}^{-3 + 2\sqrt{2}} \frac{(-1)^{(p-1)/2}z^{|p|/2-1/2}k^{(v)}(z)}{\sqrt{-1-6z-z^2}}dz\Bigg) 
 \end{align*} from which we can show the formula in Equation~\ref{eq:c2} by considering $p>0$ and $p<0$ separately. This completes the proof.
\end{proof}

 \appendix
\section{Some more proofs}\label{sec:moreproofs}

\subsection{Proof of Theorem~\ref{theorem:Ka1}}\label{sec:formuladerivation}
This is a fairly trivial proof starting from the results stated in \cite{chhita2016domino}.

For $x_1, x_2$ even with $0 < x_1,x_2 < 2n$
\begin{equation}
H_{x_1, x_2}(\omega) = \frac{\omega^{2m} G(\omega)^{2m - \frac{x_1}{2}}}{G(\omega^{-1})^{2m - \frac{x_2}{2}}}\label{eq:H}\end{equation} where $G(\omega)$ is defined in Equation~\ref{eq:G}. Let $\circlecontour_r$ denote a positively oriented contour of radius $r$ centered at the origin.
For $n = 4m$, $0 < a < 1$, $x = (x_1,x_2) \in \mathtt{W}_\e1,\, y = (y_1,y_2) \in \mathtt{B}_\e2$ with $\e1,\e2 \in \{0,1\}$, $0 < x_1,x_2,y_1,y_2 < n$, and $\sqrt{2c} < r < 1$, define \begin{multline}\label{eq:B00}
\mathcal{B}_{\e1, \e2} (a, x_1, x_2, y_1, y_2) = \frac{i^{(x_2 - x_1 + y_1  - y_2)/2}}{(2\pi i)^2} \\
\times \int_{\circlecontour_r} \frac{d\w1}{\w1} \int_{\circlecontour_{1/r}}d\w2 \frac{\w2}{\w2^2 - \w1^2} \frac{H_{x_1 + 1, x_2}(\w1)}{H_{y_1, y_2 + 1}(\w2)}\sum_{\g1,\g2=0}^1 Q_{\g1,\g2}^{\e1,\e2}(\w1,\w2),
\end{multline} \begin{multline}\label{eq:B10}
a^{-1}\mathcal{B}_{1-\e1,\e2}(a^{-1},2n-x_1,x_2,2n-y_1,y_2) = -\frac{i^{(x_1-x_2-y_1-y_2)/2}}{(2\pi i)^2} \\
\times \int_{\circlecontour_r}\frac{d\w1}{\w1}\int_{\circlecontour_{1/r}}d\w2\frac{\w2}{\w2^2-\w1^2}\frac{H_{x_1+1,x_2}(\w1)}{H_{2n-y_1,y_2+1}(\w2)}\sum_{\g1,\g2=0}^1 (-1)^{\e2+\g2}Q_{\g1,\g2}^{\e1,\e2}(\w1,\w2) \end{multline}\begin{multline}\label{eq:B01}
a^{-1}\mathcal{B}_{\e1,1-\e2}(a^{-1},x_1,2n-x_2,y_1,2n-y_2) = -\frac{i^{(y_2-y_1-x_2-x_1)/2}}{(2\pi i)^2} \\
\times \int_{\circlecontour_r}\frac{d\w1}{\w1}\int_{\circlecontour_{1/r}}d\w2\frac{\w2}{\w2^2-\w1^2}\frac{H_{x_1+1,2n-x_2}(\w1)}{H_{y_1,y_2+1}(\w2)}\sum_{\g1,\g2=0}^1 (-1)^{\e1+\g1}Q_{\g1,\g2}^{\e1,\e2}(\w1,\w2) \end{multline}\begin{multline}\label{eq:B11}
\mathcal{B}_{1-\e1,1-\e2}(a,2n-x_1,2n-x_2,2n-y_1,2n-y_2) = -\frac{i^{(y_2+y_1+x_2+x_1)/2}}{(2\pi i)^2} \\
\times \int_{\circlecontour_r}\frac{d\w1}{\w1}\int_{\circlecontour_{1/r}}d\w2\frac{\w2}{\w2^2-\w1^2}\frac{H_{x_1+1,2n-x_2}(\w1)}{H_{2n-y_1,y_2+1}(\w2)}\sum_{\g1,\g2=0}^1 (-1)^{\e1+\g1 + \e2 + \g2}Q_{\g1,\g2}^{\e1,\e2}(\w1,\w2) \end{multline}
where $Q_{\g1,\g2}^{\e1,\e2}(\w1,\w2)$ is as defined in Equation~\ref{eq:Q}.

Then we have the following formula for the entries of the inverse Kasteleyn matrix.

\begin{theorem}[Chhita and Johansson \cite{chhita2016domino}]
For $n = 4m$, $0 < a < 1$, $x = (x_1,x_2) \in \mathtt{W}_\e1$ and $y = (y_1,y_2) \in \mathtt{B}_\e2$ with $\e1,\e2 \in \{0,1\}$, $0 < x_1,x_2,y_1,y_2 < n$, the entries of $K_a^{-1}((x_1,x_2),(y_1,y_2))$ are as follows.
\begin{multline}\label{eq:Ka1original}
K_a^{-1}((x_1,x_2),(y_1,y_2)) = \mathbb{K}_{a,0,0}^{-1}((x_1,x_2),(y_1,y_2)) - \Big(\mathcal{B}_{\e1, \e2} (a, x_1, x_2, y_1, y_2) \\
- \frac{i}{a}(-1)^{\e1 + \e2}(\mathcal{B}_{1-\e1, \e2} (1/a, 2n-x_1, x_2, 2n-y_1, y_2) + \mathcal{B}_{\e1, 1-\e2} (1/a, x_1, 2n-x_2, y_1, 2n-y_2))\\
 + \mathcal{B}_{1-\e1, 1-\e2} (a, 2n-x_1, 2n-x_2, 2n-y_1, 2n-y_2)\Big)
\end{multline} where $\mathbb{K}_{a,0,0}^{-1}((x_1,x_2),(y_1,y_2))$ is defined in Equation~\ref{eq:translationinvariantoriginal} and the other terms are defined in Equations~\ref{eq:B00}--\ref{eq:B11}.

\end{theorem}

Now we prove Theorem~\ref{theorem:Ka1}.
\begin{proof}[Proof of Theorem~\ref{theorem:Ka1}.]
First we follow the procedure in \cite[\$3]{chhita2016domino} to deal with the $\w2/(\w2^2 - \w1^2)$ part of the integrands of $\mathcal{B}_{\e1, \e2} (a, x_1, x_2, y_1, y_2)$, $\mathcal{B}_{1-\e1, \e2} (1/a, 2n-x_1, x_2, 2n-y_1, y_2)$, $\mathcal{B}_{\e1, 1-\e2} (1/a,\allowbreak x_1, 2n-x_2, y_1, 2n-y_2))$ and $\mathcal{B}_{1-\e1, 1-\e2} (a, 2n-x_1, 2n-x_2, 2n-y_1, 2n-y_2)$, since this is not convenient for asymptotic analysis.

Observe that \[
\frac{\w2}{\w2^2 - \w1^2} = \frac{1}{2}\left(\frac{1}{\w2-\w1} + \frac{1}{\w2 + \w1}\right).\] We can substitute this expression into the integrands and separate each integral into two double integrals. Then we want to do a change of variables $\w2 \rightarrow -\w2$ for the second integral. Note that by choice of branch cut, we have \begin{equation}
\sqrt{(-\omega)^2 + 2c} = -\sqrt{\omega^2 + 2c}\label{eq:branchcutodd}
\end{equation} and so \begin{equation}
G(-\omega) = -G(\omega)\label{eq:Gisodd}
\end{equation} Using Equation~\ref{eq:Gisodd} and the fact that $y_1 + y_2 \equiv 2\e2+1 \mod 4$, $y_1\equiv 0 \mod 2$ and $y_2 \equiv 1 \mod 4$ we have \begin{align*}
H_{y_1, y_2 + 1}(-\omega) &= \frac{\omega^{2m} G(-\omega)^{2m - \frac{y_1}{2}}}{G(-\omega^{-1})^{2m - \frac{y_2 + 1}{2}}} \\
&= (-1)^{(y_2 -y_1 + 1)/2} H_{y_1, y_2 + 1}(\omega) \\
&= (-1)^{(y_2 +y_1 + 1)/2} H_{y_1, y_2 + 1}(\omega) \\
&= (-1)^{\e2 + 1} H_{y_1, y_2 + 1}(\omega).
\end{align*} Similarly we have \begin{align*}
H_{2n - y_1, y_2 + 1}(-\omega) &= \frac{\omega^{2m} G(-\omega)^{2m - 4m + \frac{y_1}{2}}}{G(-\omega^{-1})^{2m - \frac{y_2 + 1}{2}}} \\
&= (-1)^{(y_1 + y_2 + 1)/2}H_{2n - y_1, y_2 + 1}(\omega) \\
& = (-1)^{\e2 + 1} H_{2n - y_1, y_2 + 1}(\omega) 
\end{align*} So we find \begin{multline*}
\mathcal{B}_{\e1, \e2} (a, x_1, x_2, y_1, y_2) = \frac{i^{(x_2 - x_1 + y_1  - y_2)/2}}{(2\pi i)^2} \int_{\circlecontour_r} \frac{d\w1}{\w1} \int_{\circlecontour_{1/r}}d\w2\, \frac{1}{2}\left(\frac{1}{\w2-\w1} + \frac{1}{\w2 + \w1}\right)\\
\hspace{1cm}\times  \frac{H_{x_1 + 1, x_2}(\w1)}{H_{y_1, y_2 + 1}(\w2)}\sum_{\g1,\g2=0}^1 Q_{\g1,\g2}^{\e1,\e2}(\w1,\w2) \\
= \frac{i^{(x_2 - x_1 + y_1  - y_2)/2}}{(2\pi i)^2}\Bigg( \int_{\circlecontour_r} \frac{d\w1}{\w1} \int_{\circlecontour_{1/r}}d\w2\, \frac{1}{\w2-\w1}\frac{H_{x_1 + 1, x_2}(\w1)}{H_{y_1, y_2 + 1}(\w2)} \frac{1}{2}\sum_{\g1,\g2=0}^1 Q_{\g1,\g2}^{\e1,\e2}(\w1,\w2) \\
\hspace{1cm} +{(2\pi i)^2} \int_{\circlecontour_r} \frac{d\w1}{\w1} \int_{\circlecontour_{1/r}}(-d\w2)\, \frac{1}{-\w2 + \w1} \frac{H_{x_1 + 1, x_2}(\w1)}{H_{y_1, y_2 + 1}(\w2)} \frac{1}{2}(-1)^{\e2 + 1}\sum_{\g1,\g2=0}^1Q_{\g1,\g2}^{\e1,\e2}(\w1,\w2)\Bigg) \\
= \frac{i^{(x_2 - x_1 + y_1  - y_2)/2}}{(2\pi i)^2} \int_{\circlecontour_r} \frac{d\w1}{\w1} \int_{\circlecontour_{1/r}}d\w2\, \frac{1}{\w2-\w1} \frac{H_{x_1 + 1, x_2}(\w1)}{H_{y_1, y_2 + 1}(\w2)} \\
\hspace{1cm}\times\frac{1}{2}\sum_{\g1,\g2=0}^1 (Q_{\g1,\g2}^{\e1,\e2}(\w1,\w2) + (-1)^{1 + \e2}Q_{\g1,\g2}^{\e1,\e2}(\w1,-\w2)) \\
= \frac{i^{(x_2 - x_1 + y_1  - y_2)/2}}{(2\pi i)^2} \int_{\circlecontour_r} \frac{d\w1}{\w1} \int_{\circlecontour_{1/r}}d\w2 \frac{V_{\e1,\e2}^{0,0}(\w1,\w2)}{\w2-\w1} \frac{H_{x_1 + 1, x_2}(\w1)}{H_{y_1, y_2 + 1}(\w2)}  
\end{multline*} Similarly we find \begin{multline*}
a^{-1}\mathcal{B}_{1-\e1,\e2}(a^{-1},2n-x_1,x_2,2n-y_1,y_2) \\
= -\frac{(-1)^\e2 i^{(x_1-x_2-y_1-y_2)/2}}{(2\pi i)^2} \int_{\circlecontour_r} \frac{d\omega_1}{\omega_1} \int_{\circlecontour_{1/r}}d\omega_2 \frac{V^{1,0}_{\e1, \e2} (\omega_1, \omega_2)}{\omega_2 - \omega_1} \frac{H_{x_1+1,x_2}(\omega_1)}{H_{2n-y_1,y_2+1}(\omega_2)}, \end{multline*}
\begin{multline*}
a^{-1}\mathcal{B}_{\e1,1-\e2}(a^{-1},x_1,2n-x_2,y_1,2n-y_2) \\
= -\frac{(-1)^\e1 i^{(y_2-y_1-x_2-x_1)/2}}{(2\pi i)^2} \int_{\circlecontour_r} \frac{d\omega_1}{\omega_1} \int_{\circlecontour_{1/r}}d\omega_2 \frac{V^{0,1}_{\e1, \e2} (\omega_1, \omega_2)}{\omega_2 - \omega_1} \frac{H_{x_1+1,2n-x_2}(\omega_1)}{H_{y_1,y_2+1}(\omega_2)} \end{multline*} and
\begin{multline*}
\mathcal{B}_{1-\e1,1-\e2}(a,2n-x_1,2n-x_2,2n-y_1,2n-y_2) \\
= -\frac{(-1)^{\e1+\e2} i^{(y_2+y_1+x_2+x_1)/2}}{(2\pi i)^2} \int_{\circlecontour_r} \frac{d\omega_1}{\omega_1} \int_{\circlecontour_{1/r}}d\omega_2 \frac{V^{1,1}_{\e1, \e2} (\omega_1, \omega_2)}{\omega_2 - \omega_1} \frac{H_{x_1+1,2n-x_2}(\omega_1)}{H_{2n-y_1,y_2+1}(\omega_2)}. \end{multline*}
 Now we want to replace the $H_{x_1,x_2}$ terms with their $\widetilde{H}_{x_1,x_2}$ equivalents. From Equation~\ref{eq:H} have \begin{align}
 \frac{H_{x_1 + 1, x_2}(\w1)}{H_{y_1, y_2 + 1}(\w2)} &= i^{-(x_1 + x_2 - y_1 - y_2)/2}\frac{\widetilde{H}_{x_1 + 1, x_2}(\w1)}{\widetilde{H}_{y_1, y_2 + 1}(\w2)} \\
 \frac{H_{x_1+1,x_2}(\omega_1)}{H_{2n-y_1,y_2+1}(\omega_2)} &= i^{-(x_1 + x_2 + y_1 - y_2)/2} \frac{\widetilde{H}_{x_1+1,x_2}(\omega_1)}{\widetilde{H}_{2n-y_1,y_2+1}(\omega_2)} \\
 \frac{H_{x_1+1,2n-x_2}(\omega_1)}{H_{y_1,y_2+1}(\omega_2)} &= i^{-(x_1 - x_2 - y_1 - y_2)/2} \frac{\widetilde{H}_{x_1+1,2n-x_2}(\omega_1)}{\widetilde{H}_{y_1,y_2+1}(\omega_2)} \\
  \frac{H_{x_1+1,2n-x_2}(\omega_1)}{H_{2n-y_1,y_2+1}(\omega_2)} &= i^{-(x_1 - x_2 + y_1 - y_2)/2}  \frac{\widetilde{H}_{x_1+1,2n-x_2}(\omega_1)}{\widetilde{H}_{2n-y_1,y_2+1}(\omega_2)}
 \end{align} So we can write \begin{align*}
 \mathcal{B}_{\e1, \e2} (a, x_1, x_2, y_1, y_2) &= \frac{i^{y_1 - x_1}}{(2\pi i)^2} \int_{\circlecontour_r} \frac{d\w1}{\w1} \int_{\circlecontour_{1/r}}d\w2 \frac{V_{\e1,\e2}^{0,0}(\w1,\w2)}{\w2-\w1} \frac{\widetilde{H}_{x_1 + 1, x_2}(\w1)}{\widetilde{H}_{y_1, y_2 + 1}(\w2)} \\
 &= \mathcal{I}^{0,0}_{\e1,\e2} (a, x_1, x_2, y_1, y_2), \end{align*}
 \begin{align*}
a^{-1}\mathcal{B}_{1-\e1,\e2}&(a^{-1},2n-x_1,x_2,2n-y_1,y_2) \\
&= -\frac{(-1)^\e2 i^{-x_2-y_1}}{(2\pi i)^2} \int_{\circlecontour_r} \frac{d\omega_1}{\omega_1} \int_{\circlecontour_{1/r}}d\omega_2 \frac{V^{1,0}_{\e1, \e2} (\omega_1, \omega_2)}{\omega_2 - \omega_1} \frac{\widetilde{H}_{x_1+1,x_2}(\omega_1)}{\widetilde{H}_{2n-y_1,y_2+1}(\omega_2)} \\
&= -(-1)^{\e2}i^{- x_2 - 2y_1  + x_1}\mathcal{I}^{1,0}_{\e1,\e2} (a, x_1, x_2, y_1, y_2), \end{align*}
\begin{align*}
a^{-1}\mathcal{B}_{\e1,1-\e2}&(a^{-1},x_1,2n-x_2,y_1,2n-y_2) \\
&= -\frac{(-1)^\e1 i^{y_2-x_1}}{(2\pi i)^2} \int_{\circlecontour_r} \frac{d\omega_1}{\omega_1} \int_{\circlecontour_{1/r}}d\omega_2 \frac{V^{0,1}_{\e1, \e2} (\omega_1, \omega_2)}{\omega_2 - \omega_1} \frac{\widetilde{H}_{x_1+1,2n-x_2}(\omega_1)}{\widetilde{H}_{y_1,y_2+1}(\omega_2)} \\
&= -(-1)^\e1 i^{y_2- y_1}\mathcal{I}^{0,1}_{\e1,\e2} (a, x_1, x_2, y_1, y_2)\end{align*} and
\begin{align*}
\mathcal{B}_{1-\e1,1-\e2}&(a,2n-x_1,2n-x_2,2n-y_1,2n-y_2) \\
&= -\frac{(-1)^{\e1+\e2} i^{y_2+x_2}}{(2\pi i)^2} \int_{\circlecontour_r} \frac{d\omega_1}{\omega_1} \int_{\circlecontour_{1/r}}d\omega_2 \frac{V^{1,1}_{\e1, \e2} (\omega_1, \omega_2)}{\omega_2 - \omega_1} \frac{\widetilde{H}_{x_1+1,2n-x_2}(\omega_1)}{\widetilde{H}_{2n-y_1,y_2+1}(\omega_2)}\\
&=-(-1)^{\e1+\e2} i^{y_2+x_2-y_1 + x_1} \mathcal{I}^{0,1}_{\e1,\e2} (a, x_1, x_2, y_1, y_2). \end{align*}
So substituting into Equation~\ref{eq:Ka1original} we have 
\begin{multline*}
K_a^{-1}((x_1,x_2),(y_1,y_2)) = \mathbb{K}_{a,0,0}^{-1}((x_1,x_2),(y_1,y_2)) - \Big(\mathcal{I}^{0,0}_{\e1,\e2} (a, x_1, x_2, y_1, y_2) \\
+ (-1)^{\e1}i^{- x_2 - 2y_1  + x_1 + 1}\mathcal{I}^{1,0}_{\e1,\e2} (a, x_1, x_2, y_1, y_2)
+  (-1)^\e2 i^{y_2- y_1+1}\mathcal{I}^{0,1}_{\e1,\e2} (a, x_1, x_2, y_1, y_2)\\
 -(-1)^{\e1+\e2} i^{y_2+x_2-y_1 + x_1} \mathcal{I}^{0,1}_{\e1,\e2} (a, x_1, x_2, y_1, y_2)\Big)
\end{multline*} So all that remains to be shown is that the coefficients of $\mathcal{I}^{j,k}_{\e1,\e2} (a, x_1, x_2, y_1, y_2)$ are 1, $-1$, $-1$ and 1 for $(j,k) = (0,0),$ $(1,0),$ $(0,1)$ and$ (1,1)$ respectively. The first is trivial. For the others, we need to state a few equations. From the definition of $\e1$ and $\e2$ we have \begin{equation}
(-1)^\e1 = i^{x_1 + x_2 -1} \hspace{0.5cm}\text{and}\hspace{0.5cm} (-1)^\e2 = i^{y_1 + y_2 - 1}
\end{equation} Also $x_1, y_2 \equiv 1 \mod 2$ and $x_2, y_1 \equiv 0 \mod 2$ so $i^{2x_1} = -1$, $i^{2x_2} = 1$, $i^{2y_1} = 1$ and $i^{2y_2} = -1$. Hence we find \begin{align*}
(-1)^{\e1}i^{- x_2 - 2y_1  + x_1 + 1} &= i^{2x_1 - 2y_1} \\
&= -1 \\
 (-1)^\e2 i^{y_2- y_1+1} &= i^{2y_2} \\
 &= -1 \\
 -(-1)^{\e1+\e2} i^{y_2+x_2-y_1 + x_1} &= -i^{2x_1+2x_2+2y_2 -2}\\
 &= 1
\end{align*} as required. This concludes the proof of Theorem~\ref{theorem:Ka1}.
\end{proof}

\subsection{Proof of Theorem~\ref{thm:Vexpansion}}\label{sec:Vexpansionproof}
First we will rewrite $V^{j,k}_{\e1, \e2} (\w1, \w2) $ in a way that is easier to work with. 

Define rational functions $\mathrm{z}_{\g1,\g2}^{\e1,\e2}(u,v)$ as follows. First define \begin{align}
\begin{split}
\mathrm{z}_{0,0}^{0,0}(a,u,v) =& \frac{a}{(a^2+1)^2}(2a^6 u^2 v^2 - a^4(1 + u^4 + u^2v^2 - u^4v^2 + v^4 - u^2v^4) \\
&- a^2(1+3u^2 + 3v^2 + 2u^2v^2 + u^4v^2 + u^2v^4 - u^4v^4) \\
&- (1+v^2+u^2+3u^2v^2)) \\
\mathrm{z}_{0,1}^{0,0}(a,u,v) =& \frac{a}{4(a^2 + 1)}(1 + a^2u^2)(2a^2v^2 + 1 + v^2 - u^2 + u^2v^2) \\
\mathrm{z}_{1,0}^{0,0}(a,u,v) =& \frac{a}{4(a^2+1)}(1 + a^2v^2)(2a^2u^2 + 1 - u^2 + v^2 + u^2v^2) \\
\mathrm{z}_{1,1}^{0,0}(a,u,v) =& \frac{a}{4}(2a^2u^2v^2 - 1 + v^2 + u^2 + u^2v^2).\end{split}\label{eq:z00}
\end{align} Then for $\gamma_1,\gamma_2 \in \{0,1\}$ define \begin{align}\begin{split}
\mathrm{z}_{\gamma_1,\gamma_2 }^{0,1}(a,u,v) &= a^3 (\mathrm{z}_{\gamma_1,\gamma_2 }^{0,0}(a^{-1},u,v^{-1})) \\
\mathrm{z}_{\gamma_1,\gamma_2 }^{1,0}(a,u,v) &= a^3( \mathrm{z}_{\gamma_1,\gamma_2 }^{0,0}(a^{-1}, u^{-1},v)) \\
\mathrm{z}_{\gamma_1,\gamma_2 }^{1,1}(a,u,v) &= \mathrm{z}_{\gamma_1,\gamma_2 }^{0,0}(a, u^{-1},v^{-1}).\end{split}\label{eq:zothers}
\end{align} We write $\mathrm{z}_{\g1,\g2}^{\e1,\e2}(u,v) = \mathrm{z}_{\g1,\g2}^{\e1,\e2}(a, u,v)$. 
 Note that we have \begin{equation}
\mathrm{y}_{j,k}^{0,0}(a,1,u,v) = \frac{ z_{j,k}^{0,0}(a,u,v)}{f_{a,1}(u,v)}\label{eq:ya1}
\end{equation}and \begin{equation}
\mathrm{y}_{j,k}^{0,0}(1,a,u,v) = a^3 \frac{ z_{j,k}^{0,0}(a^{-1},u,v)}{f_{a,1}(u,v)}\label{eq:y1a}
\end{equation} where $y_{j,k}^{0,0}(a,b,u,v)$ is defined in Equation~\ref{eq:y00} and $f_{a,b}(u,v)$ defined in Equation~\ref{eq:fa}. We will also make use of the following lemma.
\begin{lemma}[Chhita and Johansson \cite{chhita2016domino}] The function $f_{a,1}$ satisfies \begin{equation}
\frac{1 - \w1^2 \w2^2}{f_{a,1}(G(\w1), G(\w2))} = \frac{1}{4(1+a^2)^2 G(\w1)^2 G(\w2)^2}\label{eq:frelation}
\end{equation}
\end{lemma} Now we can prove
\begin{lemma}
For $j,k,\e1,\e2\in\{0,1\}$, we have
\begin{multline}
V^{j,k}_{\e1, \e2} (\w1, \w2) = \frac{(-1)^{\e1+\e2+\e1\e2}\, G(\w1)^{3\e1-1}\, G(\w2^{-1})^{3\e2-1}}{4(1+a^2)^2 \prod_{i=1,2} \sqrt{\omega_i^2 + 2c}\sqrt{\omega_i^{-2} + 2c}} \\
\times \sum_{\g1, \g2 = 0}^1 (-1)^{\g1(1+\e2 + k) + \g2(1+\e1 + j)} t(\w1)^{\g1} t(\w2^{-1})^{\g2} \mathrm{z}_{\g1,\g2}^{\e1,\e2}(G(\w1),G(\w2^{-1}))\label{eq:Vsimplified}
\end{multline}
\end{lemma}
\begin{proof}
Recall the definition of  $V^{j,k}_{\e1, \e2} (\w1, \w2) $ from Equation~\ref{eq:Voriginal}. First we show that $\mathrm{x}_{\g1, \g2}^{\e1, \e2}(\w1, \w2)$ (defined in Equation~\ref{eq:xg1g2}) is an odd function in each variable. Recall from Equations~\ref{eq:branchcutodd} and \ref{eq:Gisodd} that we have $\sqrt{(-\omega)^2 + 2c} = -\sqrt{\omega^2 + 2c}$ and $G(-\omega) = -G(\omega)$. From Equation~\ref{eq:branchcutodd} we also see that
\begin{equation}
t(-\omega) = t(\omega)
\end{equation} where $t(\omega) = \omega \sqrt{\omega^{-2} + 2c}$ is defined in Equation~\ref{eq:t}. So \begin{equation*}\mathrm{x}_{\g1, \g2}^{\e1, \e2}(\w1, -\w2) = \frac{-G(\w1) G(\w2)}{\prod_{i=1}^2 \sqrt{\omega_i^2 +2c} \sqrt{\omega_i^{-2} + 2c}} \textrm{y}_{\g1, \g2}^{\e1, \e2} (G(\w1), -G(\w2))(1 - \w1^2\w2^2).\end{equation*} It is clear that $\mathrm{z}_{j,k}^{\e1,\e2}(u,v)$ is an even function in each variable, since it only contains terms of even order. We also see that $f_{a,1}(u, v)$ defined in Equation~\ref{eq:fa} satisfies \begin{align*}
f_{a,1}(u, -v) &= (-2a^2 uv - 2uv -a(-1+u^2)(-1+v^2))(-2a^2 uv - 2uv +a(-1+u^2)(-1+v^2))\\
&= (2a^2 uv + 2uv + a(-1+u^2)(-1+v^2))(2a^2 uv + 2uv - a(-1+u^2)(-1+v^2))\\
&= f_{a,1}(u,v)
\end{align*} and similarly for the first variable. So $\mathrm{y}_{j,k}^{\e1,\e2}(u,v)$ is an even function in each variable. Hence \begin{equation*}\mathrm{x}_{\g1, \g2}^{\e1, \e2}(\w1, -\w2) =-\mathrm{x}_{\g1, \g2}^{\e1, \e2}(\w1, \w2)\end{equation*} and similarly $\mathrm{x}_{\g1, \g2}^{\e1, \e2}(-\w1, \w2) = -\mathrm{x}_{\g1, \g2}^{\e1, \e2}(\w1, -\w2)$. So we see that \begin{align*}
Q_{\g1,\g2}^{\e1,\e2}(\w1,-\w2) =& (-1)^{\e1+\e2+\e1\e2 + \g1(1+\e2) + \g2(1+\e1)} t(\w1)^{\g1}t((-\w2)^{-1})^{\g2}\\
&\hspace{1cm}\times G(\w1)^\e1 G((-\w2)^{-1})^\e2 \mathrm{x}_{\g1, \g2}^{\e1, \e2}(\w1, (-\w2)^{-1})\\
=& (-1)^{\e1+\e2+\e1\e2 + \g1(1+\e2) + \g2(1+\e1)}t(\w1)^{\g1}t(\w2^{-1})^{\g2}\\
&\hspace{1cm}\times G(\w1)^\e1 (-1)^\e2 G(\w2^{-1})^\e2 (-\mathrm{x}_{\g1, \g2}^{\e1, \e2}(\w1, \w2^{-1}))\\
=& (-1)^{\e2+1} Q_{\g1,\g2}^{\e1,\e2}(\w1,\w2)
\end{align*} where $Q_{\g1,\g2}^{\e1,\e2}(\w1,\w2)$ is defined in Equation~\ref{eq:Q}. So $V^{\delta_1,\delta_2}_{\e1, \e2} (\w1, \w2)$ can be simplified to \begin{equation}\tilde{V}^{\delta_1,\delta_2}_{\e1, \e2} (\w1, \w2) = \sum_{\gamma_1,\gamma_2=0}^1 (-1)^{\gamma_2\delta_1 + \gamma_1\delta_2}Q_{\gamma_1,\gamma_2}^{\e1,\e2}(\omega_1,\omega_2)
\end{equation} Now we introduce the notation \begin{equation}
\phi(\varepsilon) = \begin{cases} 1 & \text{if } \varepsilon = 0\\
-1 & \text{if } \varepsilon = 1 \end{cases}
\end{equation} and \[
h(\e1,\e2) =  \begin{cases} 0 & \text{if } \e1=\e2\\
1 & \text{if } \e1 \neq \e2 \end{cases}
\]So we can write Equation~\ref{eq:otherys} as \begin{align}
\begin{split}
\mathrm{y}_{j,k}^{\e1,\e2}(a,1,u,v) &= \frac{\mathrm{y}_{j,k}^{0,0}(a^{1-h(\e1,\e2)},a^{h(\e1,\e2)}, u^{\phi(\e1)},v^{\phi(\e2)})}{u^{2\e1} v^{2\e2}} \\
&= a^{3 h(\e1,\e2)}\frac{\mathrm{z}_{j,k}^{0,0}(a^{\phi(\e1)\phi(\e2)}, u^{\phi(\e1)},v^{\phi(\e2)})}{u^{2\e1}v^{2\e2}f_{a,1}(u^{\phi(\e1)},v^{\phi(\e2)})}\\
&= \frac{\mathrm{z}_{j,k}^{\e1,\e2}(u,v)}{u^{2\e1}v^{2\e2}f_{a,1}(u^{\phi(\e1)},v^{\phi(\e2)})}.\end{split}\label{eq:yije1e2_unsimplified}
\end{align}
where we use Equations~\ref{eq:ya1}-\ref{eq:y1a} and Equation~\ref{eq:zothers} in the second and third lines respectively, and recall that  $\mathrm{z}_{\g1,\g2}^{\e1,\e2}(u,v) = \mathrm{z}_{\g1,\g2}^{\e1,\e2}(a, u,v)$. Now we have \begin{align*}
f_{a,1}(u, v^{-1}) &= (2a^2 uv^{-1} + 2uv^{-1} -a(-1+u^2)(-1+v^{-2})) \\
&\hspace{2cm}\times (2a^2 uv^{-1} + 2uv^{-1} +a(-1+u^2)(-1+v^{-2}))\\
&= v^{-4}(2a^2 uv + 2uv -a(-1+u^2)(-1+v^2)) \\
&\hspace{2cm}\times (2a^2 uv + 2uv +a(-1+u^2)(-1+v^2)) \\
&= v^{-4}f_{a,1}(u, v)
\end{align*} and similarly $f_{a,1}(u^{-1}, v) = u^{-4}f_{a,1}(u,v)$. So $f_{a,1}(u^{\phi(\e1)},v^{\phi(\e2)}) = u^{-4\e1}v^{-4\e2} f_{a,1}(u,v)$. Therefore we can simplify Equation~\ref{eq:yije1e2_unsimplified} to \begin{equation}
\mathrm{y}_{j,k}^{\e1,\e2}(u,v) = \frac{u^{2\e1}v^{2\e2}\mathrm{z}_{j,k}^{\e1,\e2}(u,v)}{f_{a,1}(u, v)}
\end{equation} Now using Equation~\ref{eq:frelation} we have \begin{align*}
\textrm{y}_{\g1, \g2}^{\e1, \e2} (G(\w1), G(\w2))(1 - \w1^2\w2^2) &= G(\w1)^{2\e1}G(\w2)^{2\e2}\mathrm{z}_{j,k}^{\e1,\e2}(G(\w1),G(\w2))\\
&\hspace{1cm}\times \frac{(1 - \w1^2\w2^2)}{f_{a,1}(G(\w1), G(\w2))}\\ 
&= \frac{G(\w1)^{2\e1}G(\w2)^{2\e2}\mathrm{z}_{j,k}^{\e1,\e2}(G(\w1),G(\w2))}{{4(1+a^2)^2 G(\w1)^2 G(\w2)^2}} \\
&= \frac{G(\w1)^{2\e1-2}G(\w2)^{2\e2-2}\mathrm{z}_{j,k}^{\e1,\e2}(G(\w1),G(\w2))}{4(1+a^2)^2}
\end{align*} So $\mathrm{x}_{\g1, \g2}^{\e1, \e2}(\w1, \w2)$ defined in Equation~\ref{eq:xg1g2} can be written as
\begin{align*}
\mathrm{x}_{\g1, \g2}^{\e1, \e2}(\w1, \w2) &= \frac{G(\w1)^{2\e1-1} G(\w2)^{2\e2-1}\mathrm{z}_{j,k}^{\e1,\e2}(G(\w1),G(\w2))}{4(1+a^2)^2\prod_{i=1}^2 \sqrt{\omega_i^2 +2c} \sqrt{\omega_i^{-2} + 2c}}
\end{align*} and $Q_{\g1,\g2}^{\e1,\e2}(\w1,\w2)$ defined in Equation~\ref{eq:Q} can be written \begin{multline*}
Q_{\g1,\g2}^{\e1,\e2}(\w1,\w2) = (-1)^{\e1+\e2+\e1\e2 + \g1(1+\e2) + \g2(1+\e1)} t(\w1)^{\g1}t(\w2^{-1})^{\g2} \\
\times \frac{G(\w1)^{3\e1-1} G(\w2^{-1})^{3\e2-1}\mathrm{z}_{j,k}^{\e1,\e2}(G(\w1),G(\w2^{-1}))}{4(1+a^2)^2\prod_{i=1}^2 \sqrt{\omega_i^2 +2c} \sqrt{\omega_i^{-2} + 2c}}\end{multline*} Hence we have \begin{align*}
V^{\delta_1,\delta_2}_{\e1, \e2} (\w1, \w2) &= \sum_{\gamma_1,\gamma_2=0}^1 (-1)^{\gamma_2\delta_1 + \gamma_1\delta_2}Q_{\gamma_1,\gamma_2}^{\e1,\e2}(\omega_1,\omega_2) \\
&= \sum_{\gamma_1,\gamma_2=0}^1 (-1)^{\e1+\e2+\e1\e2+\g1(1+\e2 + \delta_2) + \g2(1+\e1 + \delta_1)} t(\w1)^{\g1}t(\w2^{-1})^{\g2}\\
&\hspace{4cm}\times  \frac{G(\w1)^{3\e1-1} G(\w2^{-1})^{3\e2-1}\mathrm{z}_{j,k}^{\e1,\e2}(G(\w1),G(\w2^{-1}))}{4(1+a^2)^2\prod_{i=1}^2 \sqrt{\omega_i^2 +2c} \sqrt{\omega_i^{-2} + 2c}} 
\end{align*} as required.
\end{proof}

 Write \[
\w1 = i + \bconst^2 m^{-1} w \hspace{0.5cm} \text{and} \hspace{0.5cm} \w2 = i + \bconst^2 m^{-1} z 
\] as in Definition~\ref{def:asymptoticsvars}. From Lemma~\ref{lemma:sqrtexpansions} we have \begin{align*}
 \sqrt{\w1^2 + 2c} &=i\bconst m^{-1/2}\sqrt{1/2-2iw} + O(m^{-1}w), \\
 \sqrt{\w1^{-2} + 2c} &= -i\bconst m^{-1/2}\sqrt{1/2 + 2iw} + O(m^{-1}w),\\
 \sqrt{\w2^2 + 2c} &=i\bconst m^{-1/2}\sqrt{1/2-2iz} + O(m^{-1}z), \\
 \sqrt{\w2^{-2} + 2c} &= -i\bconst m^{-1/2}\sqrt{1/2 + 2iz} + O(m^{-1}z).
\end{align*} and from Equation~\ref{eq:Gasymptotics} and the fact that $1/\sqrt{2c} = 1 + O(m^{-1})$ we have \begin{align*}
G(\w1) &= i  - i\bconst m^{-1/2}\sqrt{1/2-2iw} + O(m^{-1}w) \\
G(\w2^{-1}) &= -i  + i\bconst m^{-1/2}\sqrt{1/2+2iz} + O(m^{-1}z).
\end{align*} Also we see that \begin{align*}
t(\w1) &= \bconst m^{-1/2}\sqrt{1/2 + 2iw} + O(m^{-1}w)\\
t(\w2^{-1}) &= \bconst m^{-1/2}\sqrt{1/2-2iz} + O(m^{-1}z).
\end{align*} Recall that \[
a = 1 - \bconst m^{-1/2}.\] Let \[\sigma(w,z) = \sqrt{1/2-2iw}\sqrt{1/2+2iw}\sqrt{1/2-2iz}\sqrt{1/2+2iz}.\] So the part of $V^{j,k}_{\e1, \e2} (\w1, \w2)$ that is not dependent on $\gamma_1, \gamma_2$ can be written \begin{align}\label{eq:Vasymprefactor}
\begin{split}
\frac{(-1)^{\e1+\e2+\e1\e2}\, G(\w1)^{3\e1-1}\, G(\w2^{-1})^{ 3\e2-1}}{4(1+a^2)^2 \prod_{i=1,2} \sqrt{\omega_i^2 + 2c}\sqrt{\omega_i^{-2} + 2c}}
&= \frac{(-1)^{\e1+\e2+\e1\e2}i^{ 3\e1-1}(-i)^{ 3\e2-1} + O(m^{-1/2})}{4(2 + O(m^{-1/2}))^2 (\bconst^4 m^{-2}\sigma(w,z) + O(m^{-5/2} w^{5/2}))}\\
&= \frac{\bconst^{-4}m^2 (-1)^{\e1+\e2+\e1\e2}i^{-\e1 + \e2}}{16 \sigma(w,z)}(1  + O(m^{-1/2}w^{1/2}))
\end{split}
\end{align} Now we look at the terms $\mathrm{z}_{\g1,\g2}^{\e1,\e2}(u,v)$ where $u = G(\w1)$ and $v = G(\w2^{-1})$. We have \begin{align*}
u^2 = G(\w1)^2 &= -1 + 2 m^{-1/2} \bconst\sqrt{1/2 - 2iw} + O(m^{-1}w), \\
v^2 = G(\w2^{-1})^2 &= -1 + 2 m^{-1/2} \bconst\sqrt{1/2 +2iz} + O(m^{-1}z). 
\end{align*} Write $u^2 = -1 -s$ and $v^2 = -1 -t$ where \[
s = -2  \bconst m^{-1/2}\sqrt{1/2 - 2iw} + O(m^{-1}w) \text{ and } t = -2\bconst m^{-1/2} \sqrt{1/2 +2iz} + O(m^{-1}z).\] Then $u^2v^2 = 1 + s + t + st$. Also \begin{equation}\label{eq:uvsquareinverse}
u^{-2} = -1 + s +O(m^{-1}w) \text{ and } v^{-2} = -1 + t + O(m^{-1}z).\end{equation} Write $h = \bconst m^{-1/2}$ so \[
a = 1 - h.
\] First we look at $\textrm{z}^{0,0}_{0, 0}(a^{\pm 1}, u, v).$ It is defined as
\begin{align*} 
\mathrm{z}_{0,0}^{0,0}(a^{\pm 1}, u,v) =& \frac{1}{4a^{\pm 2}(a^{\pm 2}+1)^2}(2a^{\pm 6} u^2 v^2 - a^{\pm 4}(1 + u^4 + u^2v^2 - u^4v^2 + v^4 - u^2v^4) \\
&- a^{\pm 2}(1+3u^2 + 3v^2 + 2u^2v^2 + u^4v^2 + u^2v^4 - u^4v^4) \\
&- (1+v^2+u^2+3u^2v^2)) 
\end{align*} It turns out that we will need the order $m^{-3/2}w^{3/2}$ terms. We have
\begin{align*}
u^2v^2 &= 1 + s + t + st \\
&= (1 + s + t) + st,
\end{align*}
\begin{align*}
1 + u^4 + u^2v^2 - &u^4v^2 + v^4 -u^2v^4 \\
&= 1 + 1 + 2s + s^2 + 1 + s + t + st - (1 + 2s + s^2)(-1-t)\\
&\hspace{1cm} + 1 + 2t + t^2 - (-1-s)(1 + 2t + t^2) \\
&= 6 + 6s + 6t + 2s^2 + 2t^2 + 5st + s^2t +st^2 + O(m^{-2}w^2)\\
&= 6(1 + s + t) + 2s^2 + 2t^2 + 5st + s^2t +st^2 + O(m^{-2}w^2),
\end{align*}
\begin{align*}
1 + 3u^2 + 3v^2 + &2u^2v^2 + u^4v^2 + u^2 v^4 - u^4v^4 \\
&= 1 - 3 -3s -3 - 3t + 2 + 2s + 2t + 2st \\
&\hspace{1cm}+(1 + 2s + s^2)(-1 -t) + (-1 -s)(1 + 2t + t^2)\\
& \hspace{1cm} - (1 + 2s + s^2)(1 + 2t + t^2) \\
&= -6 - 6s - 6t - 6st - 2s^2 - 2t^2 -3s^2t -3st^2 + O(m^{-2}w^2) \\
&= -6(1 + s + t) - 6st - 2s^2 - 2t^2 -3s^2t -3st^2 + O(m^{-2}w^2),
\end{align*}
\begin{align*}
1 + v^2 + u^2 + 3u^2v^2 &= 1 -1 -t -1 -s + 3 + 3s + 3t + 3st \\
&= 2 + 2s + 2t + 3st \\
&= 2(1 + s + t) + 3st
\end{align*} and also
\begin{align*}
2a^{\pm 6} &= 2 \mp 12h + 30h^2 \mp 40h^3 + O(m^{-2}w^2) \\
-a^{\pm 4} &= -1 \pm 4h -6h^2 \pm 4h^3 +O(m^{-2}w^2) \\
-a^{\pm 2} &= -1 \pm 2h -h^2 
\end{align*} Multiplying and adding up, everything up to order $m^{-3/2}$ cancels, and the $h^2s$, $h^2t$ terms also cancel. We are left with $h^3$ terms: 
\begin{align*}
\pm(-40h^3 + 6\times 4 h^3)  &= \mp 16h^3 \\
&= \mp 16 m^{-3/2} \bconst^3
\end{align*} the $hs^2$, $hst$, $ht^2$ terms: \begin{align*}
\pm h(-12st + 4(5 st &+ 2s^2 + 2t^2) + 2(-6st - 2s^2 - 2t^2))\\
 &= \pm h(-4st + 4s^2 + 4t^2) \\
 &= \pm 16m^{-3/2}\bconst^3\left(-\sqrt{1/2-2iw}\sqrt{1/2 +2iz} + (1/2-2iw) + (1/2+2iz)\right),
\end{align*} and the $s^2t$, $st^2$ terms: \begin{align*}
-s^2t - st^2& -(-3)s^2t - (-3)st^2 = 2(s^2t + st^2) \\
&= -16 m^{-3/2}\bconst^3 \sqrt{1/2 - 2iw}\sqrt{1/2+2iz}\left(\sqrt{1/2-2iw} + \sqrt{1/2+2iz} \right).\end{align*}
Putting these all together we obtain \begin{align*}
\textrm{z}^{0,0}_{0,0}(a^{\pm}, u, v) &= m^{-3/2}\bconst^3  \left(\pm\left(- 1 - \sqrt{1/2-2iw}\sqrt{1/2 +2iz} + (1/2-2iw) + (1/2+2iz)\right)\right. \\
&\hspace{1.5cm}\left.- \sqrt{1/2 - 2iw}\sqrt{1/2+2iz}\left(\sqrt{1/2-2iw} + \sqrt{1/2+2iz} \right)\right) + O(m^{-2})\\
&= m^{-3/2}\bconst^3 \Big(\pm(-2iw +2iz) + \sqrt{1/2-2iw}\sqrt{1/2 +2iz}\Big(\mp 1\\
&\hspace{1.5cm} - \sqrt{1/2-2iw} - \sqrt{1/2+2iz}\Big)\Big)+O(m^{-2}w^2, m^{-2}z^2).
\end{align*} Next we look at $\textrm{z}^{0,0}_{1,0}(u, v)$.
\begin{align*}
\textrm{z}^{0,0}_{1,0}(a^{\pm1}, u, v) &= \frac{(1 + a^{\pm 2}v^2)(2a^{\pm 2}u^2 + (1 - v^2 + u^2 + u^2v^2))}{4a^{\pm 2}(a^{\pm 2} + 1)}\\
&= \frac{1}{4(1\pm h)^2(2\mp 2h+h^2)}(1 + (1 \mp 2h + h^2)(-1-t))(2(1 \mp 2h+h^2)(-1-s) \\
& \hspace{2cm}+ (1 + 1 + t - 1 - s + 1 + t + s + st)) \\
&= \frac{1 + O(m^{-1/2})}{8} (\pm 2h-t + O(m^{-1}z))(4h + 2t - 2s +O(m^{-1}w, m^{-1}z)) \\
&= \frac{1}{4}(\pm 2h - t)(\pm 2h + t - s) + O(m^{-3/2}w^{3/2}, m^{-3/2}z^{3/2}) \\
&= \frac{1}{4}(4h^2 - t^2 + s(\mp 2h+t))+ O(m^{-3/2}w^{3/2}, m^{-3/2}z^{3/2})\\
&=\bconst^2 m^{-1}\left(1 - (1/2+2iz)  \pm \sqrt{1/2 -2iw} + \sqrt{1/2-2iw}\sqrt{1/2+2iz}\right) \\
&\hspace{2cm}+ O(m^{-3/2}w^{3/2}, m^{-3/2}z^{3/2})\\
&=\bconst^2 m^{-1}\left(1/2-2iz  \pm \sqrt{1/2 -2iw} + \sqrt{1/2-2iw}\sqrt{1/2+2iz}\right) \\
&\hspace{2cm}+ O(m^{-3/2}w^{3/2}, m^{-3/2}z^{3/2}).
\end{align*} Similarly, \begin{align*}
\textrm{z}^{0,0}_{0,1}(a^{\pm 1}, u, v) &=\bconst^2 m^{-1}\left(1/2+2iw  \pm \sqrt{1/2 +2iz} + \sqrt{1/2-2iw}\sqrt{1/2+2iz}\right) \\
&\hspace{2cm}+ O(m^{-3/2}w^{3/2}, m^{-3/2}z^{3/2}).
\end{align*} Finally,
\begin{align*}
\textrm{z}^{0,0}_{1,1}(a^{\pm 1}, u,v) &=  \frac{1}{4a^{\pm 2}}(2a^{\pm 2}u^2v^2 + (-1 + v^2 + u^2 + u^2v^2))\\
&= \frac{1}{4}(1 \mp h)(2(1 \mp 2h + h^2)(1 + s + t + st) - 1 - 1 - s - 1 - t + 1 + s + t + st) \\
&= \frac{1}{4}(1 \mp h) (2s + 2t \mp 4h + O(m^{-1}w, m^{-1}z)) \\
&= \frac{1}{2}(s + t \mp 2h) + O(m^{-1}w, m^{-1}z) \\
&= m^{-1/2}\bconst\left(-\sqrt{1/2-2iw} - \sqrt{1/2+2iz} \mp 1\right) + O(m^{-1}w, m^{-1}z).
\end{align*}

Now from the definition of $\textrm{z}^{\e1,\e2}_{\g1,\g2}(u,v)$ in Equation~\ref{eq:zothers} and Equation~\ref{eq:uvsquareinverse} we see that \begin{align*}
\textrm{z}^{\e1,\e2}_{0,0}(u, v) &= \bconst^3 m^{-3/2}\Big((-1)^{\e1+\e2}(-2iw +2iz) + (-1)^{\e1+\e2}\sqrt{1/2-2iw}\sqrt{1/2 +2iz}\\
&\hspace{3cm} \times \left(-(-1)^{\e1+\e2} -(-1)^\e1 \sqrt{1/2-2iw} - (-1)^\e2 \sqrt{1/2+2iz}\right)\Big) \\
&\hspace{3cm}+O(m^{-2}w^2, m^{-2}z^2),
\end{align*} \begin{align*}
\textrm{z}^{\e1,\e2}_{1,0}(u, v) &= \bconst^2 m^{-1}\left(1/2-2iz  +(-1)^\e2\sqrt{1/2 -2iw} + (-1)^{\e1 + \e2}\sqrt{1/2-2iw}\sqrt{1/2+2iz}\right) \\
&\hspace{2cm}+ O(m^{-3/2}w^{3/2}, m^{-3/2}z^{3/2}),
\end{align*} \begin{align*}
\textrm{z}^{\e1,\e2}_{0,1}(u, v) &=\bconst^2 m^{-1}\left(1/2+2iw  +(-1)^\e1\sqrt{1/2 +2iz} + (-1)^{\e1+\e2}\sqrt{1/2-2iw}\sqrt{1/2+2iz}\right) \\
&\hspace{2cm}+ O(m^{-3/2}w^{3/2}, m^{-3/2}z^{3/2}) 
\end{align*} and \begin{align*}
\textrm{z}^{\e1,\e2}_{1,1}(u,v) &= \bconst m^{-1/2}\left(-(-1)^\e1\sqrt{1/2-2iw} - (-1)^\e2\sqrt{1/2+2iz} - (-1)^{\e1+\e2}\right) \\
&\hspace{2cm}+ O(m^{-1}w, m^{-1}z).
\end{align*}  So the terms of the sum in the definition of $V^{j,k}_{\e1, \e2} (\w1, \w2) $ are \begin{multline*}
\mathrm{z}_{0,0}^{\e1,\e2}(G(\w1),G(\w2^{-1})) = \bconst^3m^{-3/2} \Big((-1)^{\e1+\e2}(-2iw +2iz) \\
+ (-1)^{\e1+\e2}\sqrt{1/2-2iw}\sqrt{1/2 +2iz} \big(-(-1)^{\e1+\e2} -(-1)^\e1 \sqrt{1/2-2iw} - (-1)^\e2 \sqrt{1/2+2iz}\big)\Big) \\
+O(m^{-2}w^2, m^{-2}z^2),
\end{multline*}
\begin{multline*}
(-1)^{1+\e2 + k} t(\w1) \mathrm{z}_{1,0}^{\e1,\e2}(G(\w1),G(\w2^{-1})) = \bconst^3m^{-3/2} (-1)^{1 + \e2 + k} \sqrt{1/2 + 2iw}\Big(1/2-2iz \\
+(-1)^\e2\sqrt{1/2 -2iw}+ (-1)^{\e1 + \e2}\sqrt{1/2-2iw}\sqrt{1/2+2iz}\Big)+O(m^{-2}w^2, m^{-2}z^2),
\end{multline*}
\begin{multline*}
(-1)^{1+\e1 + j}t(\w2^{-1}) \mathrm{z}_{0,1}^{\e1,\e2}(G(\w1),G(\w2^{-1})) =
\bconst^3m^{-3/2} (-1)^{1 + \e1 + j} \sqrt{1/2 - 2iz} \Big(1/2+2iw \\
+(-1)^\e1\sqrt{1/2 +2iz} + (-1)^{\e1+\e2}\sqrt{1/2-2iw}\sqrt{1/2+2iz}\Big)+O(m^{-2}w^2, m^{-2}z^2),
\end{multline*}
\begin{multline*}
(-1)^{\e1 + \e2 + j+k} t(\w1) t(\w2^{-1}) \mathrm{z}_{1,1}^{\e1,\e2}(G(\w1),G(\w2^{-1}))= \bconst^3m^{-3/2}(-1)^{\e1+\e2+j+k}\\
\times \sqrt{1/2 + 2iw}\sqrt{1/2-2iz}\left(-(-1)^\e1\sqrt{1/2-2iw} - (-1)^\e2\sqrt{1/2+2iz} - (-1)^{\e1+\e2}\right)\\
+O(m^{-2}w^2, m^{-2}z^2)
\end{multline*} Summing these and simplifying we obtain,
\begin{multline*}
\sum_{\g1, \g2 = 0}^1 (-1)^{\g1(1+\e2 + k) + \g2(1+\e1 + j)} t(\w1)^{\g1} t(\w2^{-1})^{\g2} \mathrm{z}_{\g1,\g2}^{\e1,\e2}(G(\w1),G(\w2^{-1})) =  \bconst^3 m^{-3/2}(-1)^{\e1+\e2}\\
\times \Bigg(-2 i(w-z)- (-1)^{\e1+\e2}\left(\sqrt{1/2-2iw} +  (-1)^{j}\sqrt{1/2-2iz}\right)\left((-1)^{k}\sqrt{1/2+2iw} + \sqrt{1/2+2iz}\right) \\
- \left((-1)^\e1\sqrt{1/2-2iw} + (-1)^{\e2+k}\sqrt{1/2+2iw} + (-1)^\e2\sqrt{1/2+2iz} + (-1)^{\e1+j}\sqrt{1/2-2iz}\right)\\
\times \left(\sqrt{1/2-2iw}\sqrt{1/2+2iz} + (-1)^{j+k}\sqrt{1/2+2iw}\sqrt{1/2-2iz}\right)\Bigg) +O(m^{-2}w^2, m^{-2}z^2).
\end{multline*} So together with Equation~\ref{eq:Vasymprefactor} we have  
\begin{multline*}
V^{j,k}_{\e1, \e2} (\w1, \w2) = m^{1/2}\frac{-(-1)^{\e1\e2}i^{\e2 - \e1}}{16\bconst  \sqrt{1/2-2iw}\sqrt{1/2 + 2iw}\sqrt{1/2-2iz}\sqrt{1/2 + 2iz}}\\
\times \Bigg(2 i(w-z)+ (-1)^{\e1+\e2}\left(\sqrt{1/2-2iw} +  (-1)^{j}\sqrt{1/2-2iz}\right)\left((-1)^{k}\sqrt{1/2+2iw} + \sqrt{1/2+2iz}\right) \\
+\left((-1)^\e1\sqrt{1/2-2iw} + (-1)^{\e2+k}\sqrt{1/2+2iw} + (-1)^\e2\sqrt{1/2+2iz} + (-1)^{\e1+j}\sqrt{1/2-2iz}\right)\\
\times \left(\sqrt{1/2-2iw}\sqrt{1/2+2iz} + (-1)^{j+k}\sqrt{1/2+2iw}\sqrt{1/2-2iz}\right)\Bigg) + O(1).
\end{multline*}  Comparing with the definition of $A^{j,k}_{\e1,\e2}(w,z)$ in Equation~\ref{eq:Aall}, and noting that $i^{\e2-\e1} = (-1)^{\e1+\e2}i^{\e1-\e2}$, we see that \begin{equation*}
V^{j,k}_{\e1,\e2}(\omega_1,\omega_2) = m^{1/2} \frac{(-1)^{\e1\e2}i^{\e1-\e2}}{16\bconst}(A^{j,k}_{\e1,\e2}(w,z)  + O(m^{-1/2}))
\end{equation*} as desired.

\subsection{Proof of Theorem~\ref{thm:arcerrors}}\label{sec:arcerrors}
Recall the definitions of $\circlecontour_{(m)}$ and $\circlecontour_{(m)}'$ from the statement of the theorem. We need to bound the exponential parts of the integrands $|\widetilde{H}_{x_1 + 1, x_2}(\w1)|$, $|\widetilde{H}_{x_1 + 1, 2n - x_2}(\w1)|$ for $\w1 \in \circlecontour_{(m)}$, and
$|\widetilde{H}_{y_1, y_2 + 1}(\w2)|$, $|\widetilde{H}_{2n-y_1, y_2 + 1}(\w2)|$ for $\w2 \in \circlecontour_{(m)}'$ for $m$ sufficiently large. First we prove a few lemmas.

\begin{lemma}
Let $\omega = R_1 e^{i\theta}$ where $R_1 = 1 - \bconst^2 m^{\delta-1} + O(m^{\delta/2-1})$ for $0 < \delta < 1/2$ with $\theta \in [0, \pi/2)$, so $\omega$ is in the first quadrant. Then the square roots $\sqrt{\omega^2 + 2c}$ and $\sqrt{\omega^{-2} + 2c}$ defined in Equation~\ref{eq:sqrtbranchcut} agree with the principal branch of the square root.
\end{lemma}
\begin{proof}
The arguments of $\omega + i\sqrt{2c}$, $\omega - i\sqrt{2c}$, $\omega^{-1} + i\sqrt{2c}$, and $\omega^{-1} + i\sqrt{2c}$ are all in the interval $(-\pi/2, \pi/2)$. Comparing the arguments of the square root in the expression in Equation~\ref{eq:sqrtbranchcut} and for the principal branch gives the result.
\end{proof}

\begin{lemma}\label{lemma:Gincreasing}
Let $\omega = R_1 e^{i\theta}$ where $R_1 = 1 - \bconst^2 m^{\delta-1} + O(m^{\delta/2-1})$ for $0 < \delta < 1/2$ with $\theta \in [0, \pi/2)$. Then $|G(\omega)|$ and $|G(\omega^{-1})|$ both increase with $\theta$.
\end{lemma} \begin{proof}
It is equivalent to show that the logarithms of the above quantities are increasing. We compute \begin{align*}
\frac{d}{d\omega}\log G(\omega) &= -\frac{1}{\sqrt{\omega^2 + 2c}} \\
\frac{d}{d\omega}\log G(\omega^{-1}) &= \frac{1}{\omega^{-2}\sqrt{\omega^2 + 2c}} 
\end{align*} So \begin{align*}
\frac{d}{d\theta}|\log G(R_1 e^{i\theta})| &= \mathrm{Re} \frac{d}{d\theta}\log G(R_1 e^{i\theta}) \\
&= \mathrm{Re} \frac{-iR_1e^{i\theta}}{\sqrt{R_1^2 e^{2i\theta} + 2c}} \\
&= \mathrm{Im} \frac{\omega}{\sqrt{\omega^2 + 2c}}
\end{align*}
Similarly we can show \[
\frac{d}{d\theta}|\log G((R_1 e^{i\theta})^{-1})| 
= -\mathrm{Im} \frac{\omega^{-1}}{\sqrt{\omega^{-2} + 2c}}.
\]
We have $0 \leq \arg(\omega^2) - \arg(\omega^2 + 2c) \leq \arg(\omega^2) \leq \pi$ so $0 \leq \arg(\omega/\sqrt{\omega^2 + 2c}) \leq \pi/2$ and hence $ \mathrm{Im} (\omega/\sqrt{\omega^2 + 2c}) \geq 0$ and so $|\log G(R_1 e^{-i\theta})|$ is increasing as required. We also have $-\pi \leq \arg(\omega^{-2}) \leq\arg(\omega^{-2}) - \arg(\omega^{-2} + 2c) \leq 0$, so $-\pi \leq \arg(\omega^{-1}/\sqrt{\omega^{-2} + 2c}) \leq 0$ and hence $-\mathrm{Im(}\omega^{-1}/\sqrt{\omega^{-2} + 2c}) \geq 0$ so $|\log G((R_1 e^{-i\theta}))^{-1}| $ is increasing as required.
\end{proof}

Since we will have a lot of error terms of different orders to deal with, we state the following inequalities for $0 < \delta < 1/2$.\begin{equation}\label{eq:deltaineq}
\delta - 2 < 2(\delta - 1) < -1 < \delta/2 - 1 < \delta - 1 < -1/2 < (\delta - 1)/2
\end{equation}

\begin{lemma}\label{lemma:omega2cbound}
Let $\omega = R_1 e^{i\theta}$ where $R_1 = 1 - \bconst^2 m^{\delta-1} + O(m^{\delta/2-1})$ for $0 < \delta < 1/2$ with $\theta \in [0, \pi/2)$. Then \[
|\omega^2 + 2c| \geq 2\bconst^2m^{\delta-1} + O(m^{-1}).\]
\end{lemma}
\begin{proof}
Note that $R_1^2 < 2c$. It is clear (e.g. by geometry) that $|\omega^2 + 2c|$ approaches its infimum as $\omega \rightarrow i$. So we have $|\omega^2 + 2c| \geq 2c - R^2 = 2\bconst^2 m^{\delta - 1} + O(m^{-1})$.
\end{proof}

\begin{lemma}\label{lemma:Gomegabounds}
Let $\omega = R_1 e^{i\theta}$ where $R_1 = 1 - \bconst^2 m^{\delta-1} + O(m^{\delta/2-1})$ for $0 < \delta < 1/2$ with $\theta \in [0, \pi/2)$. Then for $m$ sufficiently large we have \[
\sqrt{2} - 1 + O(m^{\delta -1}) \leq |G(\omega)| \leq  1 
\]
\end{lemma}
\begin{proof}
From Lemma~\ref{lemma:Gincreasing}, the infimum of $|G(\omega)|$ on this interval occurs at $\omega = R_1$, while the supremum occurs in the limit $\omega \rightarrow R_1 i$. Recalling that $2c = 1 + O(m^{-1})$, we compute \[
G(R_1) = 1 - \sqrt{2} + O(m^{\delta-1})
\] which shows the lower bound for $|G(\omega)|$. For the upper bound, the limit of $\sqrt{\omega^2 + 2c}$ as $\omega \rightarrow R_1 i$ along the contour $\omega = R_1 e^{i\theta}$ is $\sqrt{2c - R_1^2}$, since $R_1 < \sqrt{2c}$. So \begin{align*}
|G(\omega)|^2 &\leq \left| \frac{1}{2c}\left(R_1 i - \sqrt{2c - R_1^2}\right) \right|^2 \\
&= 1,
\end{align*} which proves the upper bound.
\end{proof}

\begin{lemma}\label{lemma:Gomegainvbounds}
Let $\omega = R_1 e^{i\theta}$ where $R_1 = 1 - \bconst^2 m^{\delta-1} + O(m^{\delta/2-1})$ for $0 < \delta < 1/2$with $\theta \in [0, \pi/2)$. Then for $m$ sufficiently large we have \[
\sqrt{2} - 1 + O(m^{\delta -1}) \leq |G(\omega^{-1})| < 1 
\]
\end{lemma}
\begin{proof}
From Lemma~\ref{lemma:Gincreasing}, the minimum of $|G(\omega^{-1})|$ on this interval occurs at $\omega = R_1$,while the supremum occurs in the limit $\omega \rightarrow R_1 i$. Recalling that $2c = 1 + O(m^{-1})$, we compute \[
G(R_1^{-1}) = 1 - \sqrt{2} + O(m^{\delta-1})
\] which shows the lower bound for $|G(\omega^{-1})|$. For the upper bound, the limit of $\sqrt{\omega^{-2} + 2c}$ as $\omega \rightarrow R_1 i$ along the contour $\omega = R_1 e^{i\theta}$ is $i\sqrt{R_1^{-2} - 2c}$, since $R_1^{-1} > \sqrt{2c}$. So \begin{align*}
|G(\omega^{-1})| &\leq \left| \frac{1}{2c}\left(R_1^{-1} i - i\sqrt{R_1^{-2} - 2c}\right) \right| \\
&= 1 + \bconst^2 m^{\delta-1} - \sqrt{2}\bconst m^{(\delta-1)/2} + O(m^{-\delta/2 - 1/2}) \\
&= 1  - \sqrt{2}\bconst m^{(\delta-1)/2} + O(m^{- 1/2})
\end{align*} where we have made use of the inequalities in Equation~\ref{eq:deltaineq} to deal with the error terms. So for $m$ sufficiently large, we have $|G(\omega^{-1})| < 1$ as required.
\end{proof}

\begin{lemma}\label{lemma:Gratiowbound}
Let $\w1 = R_1 e^{i\theta}$ where $R_1 = 1 - \bconst^2 m^{\delta-1} + O(m^{\delta/2-1}) \in \mathbb{R}$ for $0 < \delta < 1/2$. Then there exists $c_2 > 0$ such that for $m$ sufficiently large \[
\left| \frac{G(\w1)}{G(\w1^{-1})} \right| < 1 + c_2\bconst m^{(\delta-1)/2}\] for all $\theta \neq \pi/2 + k \pi$, $k\in \mathbb{Z}$.
\end{lemma}
\begin{proof}
First note that \[
\left|\frac{G(\omega)}{G(\omega^{-1})}\right| = \left|\frac{G(\overline{\omega})}{G(\overline{\omega}^{-1})}\right|  = \left|\frac{G(-\omega)}{G(-\omega^{-1})}\right|  = \left|\frac{G(-\overline{\omega})}{G(-\overline{\omega}^{-1})}\right| 
\] it is sufficient to consider $\theta \in [0, \pi/2)$. Also note that since $|G(\overline{\w1}^{-1})| = |G(\w1^{-1})|$ we have \[
\left| \frac{G(\w1)}{G(\w1^{-1})} \right| = \frac{\w1 - \sqrt{\w1^2+2c}}{\overline{\w1}^{-1} - \sqrt{\overline{\w1}^{-2} + 2c}}.\] After some rearrangement, we obtain \begin{equation}\label{eq:Gratio1}
\left| \frac{G(\w1)}{G(\w1^{-1})} \right| = \left| 1 + \frac{(\w1 - \overline{\w1}^{-1}) - (\sqrt{\w1^2 + 2c}-\sqrt{\overline{\w1}^{-2} + 2c}) }{\overline{\w1}^{-1} - \sqrt{\overline{\w1}^{-2} + 2c}}\right|.\end{equation}
Now, $\w1 - \overline{\w1}^{-1} = (R_1 - R_1^{-1})e^{i\theta}$ so \[
|\w1 - \overline{\w1}^{-1}| = ( R_1^{-1} - R_1) = 2\bconst^2 m^{\delta-1} + O(m^{\delta/2-1}).
\] Also  $\w1^2 - \overline{\w1}^{-2} = (R_1^2 - R_1^{-2})e^{2i\theta} = 4\bconst^2 m^{\delta-1} + O(m^{\delta/2-1})$. Hence \[
(\w1^2 + 2c) - (\overline{\w1}^{-2}+2c) = (R_1^2 - R_1^{-2})e^{2i\theta} = 4\bconst^2 m^{\delta-1} + O(m^{\delta/2-1})
\] and so we have \begin{equation}\label{eq:pmsqrtapprox}
(\sqrt{\w1^2 + 2c}-\sqrt{\overline{\w1}^{-2} + 2c})(\sqrt{\w1^2 + 2c}+\sqrt{\overline{\w1}^{-2} + 2c}) =  4\bconst^2 m^{\delta-1} + O(m^{\delta/2-1}).\end{equation} Now note that since we are assuming $\theta \in [0, \pi/2)$, we have $\arg(\sqrt{\w1^2 + 2c}) \in [0, \pi/4)$ and $\arg(\sqrt{\w1^{-2} + 2c}) \in [0, \pi/2)$ so \[
|\sqrt{\w1^2 + 2c}+\sqrt{\overline{\w1}^{-2} + 2c}| \geq |\sqrt{\w1^2 + 2c}| \geq \sqrt{2}\bconst m^{(\delta-1)/2}
\] by Lemma~\ref{lemma:omega2cbound}. Thus by Equation~\ref{eq:pmsqrtapprox} we have \[
\sqrt{\w1^2 + 2c}-\sqrt{\overline{\w1}^{-2} + 2c} \leq 2\sqrt{2}m^{(\delta-1)/2} + O(m^{-1/2}).\]
Also, by Lemma~\ref{lemma:Gomegainvbounds} we have 
\[|\overline{\w1}^{-1} - \sqrt{\overline{\w1}^{-2} + 2c}| = |2c\, G(\w1^{-1})| \geq \sqrt{2} - 1 + O(m^{\delta-1}).\] Putting these together we see that \begin{align*}
&\left|\frac{(\w1 - \overline{\w1}^{-1}) - (\sqrt{\w1^2 + 2c}-\sqrt{\overline{\w1}^{-2} + 2c}) }{\overline{\w1}^{-1} - \sqrt{\overline{\w1}^{-2} + 2c}}\right| \\
&\hspace{4cm}\leq \frac{ 2\bconst^2 m^{\delta-1} + O(m^{\delta/2-1}) + 2\sqrt{2}m^{(\delta-1)/2} + O(m^{-1/2})}{\sqrt{2} - 1 + O(m^{\delta-1})} \\
&\hspace{4cm}\leq \frac{2\sqrt{2}m^{(\delta-1)/2}}{\sqrt{2} - 1}  + O(m^{-1/2})
\end{align*} where again we have used the inequalities in Equation~\ref{eq:deltaineq} to deal with the error terms. Hence, from Equation~\ref{eq:Gratio1} we see that there exists some $c_2 > 0$ such that  \[
\left| \frac{G(\w1)}{G(\w1^{-1})} \right| < 1 + c_2\bconst m^{(\delta-1)/2}\] for $m$ sufficiently large as required.
\end{proof}

Now we look at the contours $\Cmainzm$ and $\Cotherzm$.
\begin{lemma}
Let $\w2 = R_2 e^{i\theta}$ where $R_2 = 1 + \bconst^2 m^{\delta-1} + O(m^{\delta/2-1})$ for $0 < \delta < 1/2$. Then there exists $c_2 > 0$ such that for $m$ sufficiently large \[
\left| \frac{G(\w2)}{G(\w2^{-1})} \right| > 1 - c_2\bconst m^{(\delta-1)/2}\] for all $\theta \neq \pi/2 + k \pi$, $k\in \mathbb{Z}$.
\end{lemma}
\begin{proof}
Note that $\overline{\w2}^{-1} = ( 1 - \bconst^2 m^{\delta-1} + O(m^{\delta/2-1})) e^{-i\theta}$. So we can write $\overline{\w2}^{-1} = \w1$ for $\w1$ as in Lemma~\ref{lemma:Gratiowbound}. So we see that we have \[
\left| \frac{G(\overline{\w2}^{-1})}{G(\overline{\w2})} \right|  < 1 + c_2\bconst m^{(\delta-1)/2}
\] for $m$ sufficiently large. Since 
\[
\left| \frac{G(\w2)}{G(\w2^{-1})} \right| = \left| \frac{G(\overline{\w2})}{G(\overline{\w2}^{-1})} \right|\] the result follows.
\end{proof}

Now we are ready to prove Theorem~\ref{thm:arcerrors}
\begin{proof}[Proof of Theorem~\ref{thm:arcerrors}]
 Let $\circlecontour_{(m)}$ and $\circlecontour_{(m)}'$ be as in the statement of the theorem. On $\circlecontour_{(m)}$ we can parametrize $\omega_1$ as $\omega_1 = R_1 e^{i\theta}$ and on $\circlecontour_{(m)}'$ we can parametrize $\omega_2$ as $\omega_2 = R_2 e^{i\theta}$ where \begin{align*}
R_1 &= 1 - \bconst^2 m^{\delta-1} + O(m^{\delta/2-1}) \\
R_2 &= 1 + \bconst^2 m^{\delta-1} + O(m^{\delta/2-1})
\end{align*}and \[\theta \in [-\pi/2 + \theta_0, \pi/2 - \theta_0]\cup[\pi/2 + \theta_0, 3\pi/2 - \theta_0]\] for some $\theta_0 > 0$. We can show using Theorem~\ref{thm:thetalargew} that $\theta_0 = -\bconst^2/\sqrt{2}\alpha m^{\delta/2 - 1} + O(m^{-1})$. The contours $\circlecontour_{(m)}$ and $\circlecontour_{(m)}'$ do not touch the imaginary axis so we can apply the preceding lemmas.
First we look at $\widetilde{H}_{x_1+1, x_2}(\w1)$. Recall that $\alpha < 0$. We have 
 \[
\widetilde{H}_{x_1+1, x_2}(\w1) = \frac{\w1^{2m} (-iG(\w1))^{-\alpha \bconst m^{1/2}+O(1)}}{(i G(\w1^{-1}))^{-\alpha \bconst m^{1/2}+O(1)}}.\] For $\w1 \in \circlecontour_{(m)}$ we have \[
\log |\widetilde{H}_{x_1+1, x_2}(\w1)| = 2m \log |\w1| -\alpha \bconst m^{1/2}\log \left| \frac{G(\w1)}{G(\w1^{-1})} \right| + O(1) 
\] since by Lemma~\ref{lemma:Gomegabounds} and Lemma~\ref{lemma:Gomegainvbounds}, $G(\w1)$ and $G(\w1^{-1})$ are bounded for $\w1 \in \circlecontour_{(m)}$. We have $\log |\w1| = -\bconst^2 m^{\delta - 1} + O(m^{\delta/2-1})$ and $\log | G(\w1)/G(\w1^{-1})| < c_2\bconst m^{(\delta-1)/2} $ for some $c_2 > 0$, by Lemma~\ref{lemma:Gratiowbound}. So  \[
\log |\widetilde{H}_{x_1+1, x_2}(\w1)| < -2\bconst^2 m^\delta + O(m^{\delta/2}) -\alpha \bconst^2 c_2 m^{\delta/2} + O(1) 
\] which simplifies to \[
\log |\widetilde{H}_{x_1+1, x_2}(\w1)| < -2\bconst^2 m^\delta + O(m^{\delta/2}).
\]  Hence there exists $c_x > 0$ such that for all $\w1 \in \circlecontour_{(m)}$, \[
|\widetilde{H}_{x_1+1, x_2}(\w1)| < e^{-c_x m^\delta}
\] for $m$ sufficiently large. 

Next we look at $\widetilde{H}_{x_1+1, 2n - x_2}(\w1)$. We have \[
\widetilde{H}_{x_1+1, 2n - x_2}(\w1) = \widetilde{H}_{x_1+1, x_2}(\w1)(i G(\w1^{-1}))^{-2\alpha \bconst m^{1/2}+O(1)}.\] Since from Lemma~\ref{lemma:Gomegainvbounds} we have $|G(\w1^{-1}))| \leq 1$, and we have $\alpha < 0$ we see that for $m$ sufficiently large and $\w1 \in \circlecontour_{(m)}$ we have \[
|\widetilde{H}_{x_1+1, 2n - x_2}(\w1)| < |\widetilde{H}_{x_1+1, x_2}(\w1)| < e^{-c_x m^\delta}
\] for the $c_x$ that we found above. 

Now we look at $|\widetilde{H}_{y_1, y_2 + 1}(\w2)|$ for $\w2 \in \circlecontour_{(m)}'$. As above, we have \[
\log |\widetilde{H}_{y_1, y_2 + 1}(\w2)| = 2m \log |\w2| -\alpha \bconst m^{1/2}\log \left| \frac{G(\w2)}{G(\w1^{-2})} \right| + O(1).
\] We have $\log |\w2| = \bconst^2 m^{\delta - 1} + O(m^{\delta/2-1})$ and $\log | G(\w2)/G(\w2^{-1})| > -c_2\bconst m^{(\delta-1)/2} $ for some $c_2 > 0$, by Lemma~\ref{lemma:Gratiowbound}. So  \[
\log |\widetilde{H}_{y_1, y_2+1}(\w2)| > 2\bconst^2 m^\delta + O(m^{\delta/2}) +\alpha \bconst^2 c_2 m^{\delta/2} + O(1) 
\] which simplifies to \[
\log |\widetilde{H}_{y_1, y_2+1}(\w2)| > 2\bconst^2 m^\delta + O(m^{\delta/2}).
\]  Hence there exists $c_y > 0$ such that for all $\w2 \in \circlecontour_{(m)}'$, \[
|\widetilde{H}_{y_1, y_2+1}(\w2)| > e^{c_y m^\delta}
\] for $m$ sufficiently large. 

Next we look at $\widetilde{H}_{2n - y_1, y_2+1}(\w2)$. We have \[
\widetilde{H}_{2n - y_1, y_2+1}(\w2) = \widetilde{H}_{ y_1, y_2+1}(\w2)(-i G(\w2))^{2\alpha \bconst m^{1/2}+O(1)}.\] Since from Lemma~\ref{lemma:Gomegainvbounds} we have $|G(\w2))| \leq 1$, and recalling that $\alpha < 0$ we see that for $m$ sufficiently large and $\w2 \in \circlecontour_{(m)}'$ we have \[
|\widetilde{H}_{2n - y_1, y_2+1}(\w2)| > |\widetilde{H}_{ y_1, y_2+1}(\w2)| > e^{c_y m^\delta}
\] for the $c_y$ that we found above. 

Taking $d = \min(c_x,c_y)$ completes the theorem.

\end{proof}

\subsection{Proof of Lemma~\ref{lemma:experrorbound}}\label{sec:experrorbound}

\begin{proof}[Proof of Lemma~\ref{lemma:experrorbound}]
We can write the $O(m^{-1/2}w, m^{-1/2}z)$ error term as $m^{-1/2}R(w,z)$ where $R(w,z)\allowbreak = O(w,z)$ and does not contain any singularities. Recall that on $\widetilde{\mathcal{C}_j}$, we have $|w| < m^{\delta}$ and on $\widetilde{\mathcal{C}_k}'$ we have $|z| < m^{\delta}$ for some $0 < \delta < 1/2$ . Then there exists $c_1 > 0$ such that on $\widetilde{\mathcal{C}_j} \times \widetilde{\mathcal{C}_k}'$, we have $|e^{m^{-1/2}R(z,w)}| < e^{c_1 m^{\delta - 1/2}}$. Then we have \[|e^{m^{-1/2}R(z,w)} - 1| < m^{-1/2}R(z,w)e^{c_1 m^{\delta - 1/2}}.
\]
 Let \[
E = \int_{\widetilde{\mathcal{C}_j}}dw \int_{\widetilde{\mathcal{C}_k}'}dz \frac{A^{j,k}_{\e1,\e2}(w,z)}{z-w} e^{g_{j,k}(w,z)}(e^{m^{-1/2}R(z,w)} - 1) .\] We want to show that $E = O(m^{-1/2})$. We have \begin{align*}
E &\leq \int_{\widetilde{\mathcal{C}_j}}|dw| \int_{\widetilde{\mathcal{C}_k}'}|dz| \left|\frac{A^{j,k}_{\e1,\e2}(w,z)}{z-w}\right| |e^{g_{j,k}(w,z)}|\,|e^{m^{-1/2}R(z,w)} - 1| \\
&\leq  m^{-1/2}e^{c_1 m^{\delta - 1/2}} \int_{\widetilde{\mathcal{C}_j}}|dw| \int_{\widetilde{\mathcal{C}_k}'}|dz| \left|\frac{A^{j,k}_{\e1,\e2}(w,z) R(z,w)}{z-w}\right| |e^{g_{j,k}(w,z)}|\\
&\leq  m^{-1/2}e^{c_1 m^{\delta - 1/2}} \int_{\mathcal{C}_j}|dw| \int_{\mathcal{C}_k'}|dz| \left|\frac{A^{j,k}_{\e1,\e2}(w,z) R(z,w)}{z-w}\right| |e^{g_{j,k}(w,z)}|.
\end{align*} By Lemma~\ref{lemma:Ccontourslargew} and the bounds found in Corollary \ref{cor:expboundcontour}, and because $A^{j,k}_{\e1,\e2}(w,z) R(z,w)$ has no singularities, we see that the double integral converges. Moreover $e^{c_1 m^{\delta - 1/2}}$ decreases as $m \rightarrow \infty$. So $E = O( m^{-1/2})$ as required.

\end{proof}

\subsection{Proof of Lemma~\ref{lemma:Sazellipticasymp}}\label{sec:Sazellipticasymp}
\begin{proof}\label{proof:Sazellipticasymp}
 Let \[
F(\lambda, k) = \int_0^{\lambda} \frac{1}{\sqrt{(1 - k^2 y^2)(1-y^2)}} dy\] denote the incomplete elliptic integral of the first kind, for $\lambda, k \in (0,1)$.

We have $0\leq \lambda_z \leq 1$ for $z_1\leq z \leq z_2$, with $\lambda_{z_1} = 0$ and $\lambda_{z_2}=1$ for all $0<a<1$. Let $z_1$ and $z_2$ be as defined in Equation~\ref{eq:z1z2}, and consider $z$ with $-1 < z_1 \leq z \leq -3+2\sqrt{2} < z_2$. Note that $z_1 = -1 + \sqrt{2}h + O(h^2)$ so for any $z\in (-1,-3+2\sqrt{2})$, by taking $h$ sufficiently small, we have $z \geq z_1$.

We use the following asymptotic formula proven in \cite{Carlson1985} and stated in more convenient notation in \cite{Karp2007}:
\begin{equation}\label{eq:ellipticasymptoticformula}
F(\lambda, k) =  \lambda \log \frac{4}{\sqrt{1-\lambda^2} + \sqrt{1-k^2\lambda^2}} + \theta_1 F(\lambda,k)
\end{equation} with relative error bound \begin{equation}\label{eq:ellipticerrorbound}
\frac{(2 - \lambda^2(1+k^2))\log(1-k^2\lambda^2)}{4\log((1-k^2\lambda^2)/16)} < \theta_1 < \frac{2 - \lambda^2(1+k^2)}{4}.\end{equation} In our case, we have \begin{align*}
\frac{2 - \lambda_z^2(1+k^2)}{4} &= \frac{(1-a)^2}{8a(a+ a^{-1})(1+z)^2}(2(1-z)^2 - 2z(a+a^{-1})^2 - (1+z)^2(a+a^{-1})) \\
&\leq \frac{-z}{(z+1)^2}h^2   + A h
\end{align*} for $z \in (z_1,z_2)$, for some constant $A$ which does not depend on $z$ (we note that $z_1 + 1 = \sqrt{2}h+O(h^2)$ and $z \geq z_1$). We can also show that $(2 - \lambda^2(1+k^2))/4 \geq 0$, so the lower bound in Equation~\ref{eq:ellipticerrorbound} is positive. Now we compute \begin{multline*}
\lambda \log \frac{4}{\sqrt{1-\lambda^2} + \sqrt{1-k^2\lambda^2}} = \frac{\sqrt{\left(\frac{a+a^{-1}}{2}\right) (2(a+a^{-1}-1)z + 1 + z^2)}}{1+z} \\
\times \log \left(\frac{4 a \sqrt{a+a^{-1}}(1+z)}{(1-a)((1-z) + \sqrt{\left(\frac{a+a^{-1}}{2}\right)(-2(a+a^{-1}+1)z -1-z^2)}}\right)
\end{multline*} We can show that \[
 \frac{\sqrt{\left(\frac{a+a^{-1}}{2}\right) (2(a+a^{-1}-1)z + 1 + z^2)}}{1+z} = 1 + R_1(h,z)\] where $|R_1(h,z)| < A_1h^2/(1+z)^2$ as $h\to 0$ for some constant $A_1$, and \begin{multline*}
 \log \left(\frac{4 a \sqrt{a+a^{-1}}(1+z)}{(1-a)((1-z) + \sqrt{\left(\frac{a+a^{-1}}{2}\right)(-2(a+a^{-1}+1)z -1-z^2)}}\right)\\ = -\log h + \log \left(\frac{4\sqrt{2}(1+z)}{1-z+\sqrt{-1-6z-z^2}}\right) -\frac{h}{2} + R_2(h,z)
\end{multline*} where $|R_2(h,z)| < A_2h^2/\sqrt{-1-6z-z^2}$ as $h\to 0$ for some constant $A_2$. Putting these together and again noting that $1+z > \sqrt{2} h$, we have \[
\lambda \log \frac{4}{\sqrt{1-\lambda^2} + \sqrt{1-k^2\lambda^2}} = -\log h + \log \left(\frac{4\sqrt{2}(1+z)}{1-z+\sqrt{-1-6z-z^2}}\right) -\frac{h}{2} + R_3(h,z) 
\] where \[
|R_3(h,z)| < A_3\frac{h^2\log h}{(1+z)^2} + A_3'\frac{h^2}{\sqrt{-1-6z-z^2}}\]
 as $h\to 0$ for some constants $A_3,A_3'$. Now using Equation~\ref{eq:ellipticasymptoticformula} and noting that  $|\theta_1 |< A_4(h^2(z+1)^{-2}) + A_4'h$ for constants $A_4,A_4'$, the result follows.

\end{proof}

\bibliographystyle{alpha}
\bibliography{refs}

\end{document}